\documentclass[11pt]{article}

\usepackage[utf8]{inputenc}     
\usepackage[T1]{fontenc}        
\usepackage{hyperref}           
\usepackage{url}                
\usepackage{booktabs}           
\usepackage{amsfonts}           
\usepackage{mathrsfs}
\usepackage{nicefrac}           
\usepackage{lipsum}             
\usepackage{fancyhdr}           
\usepackage{graphicx}           
\graphicspath{{media/}}         
\usepackage{threeparttable}
\usepackage{xcolor}
\usepackage{soul}
\usepackage{multirow}
\usepackage{multicol}
\usepackage{array}
\usepackage{tikz}
\usepackage{amssymb}
\usepackage{amsmath}
\usepackage{algorithm}
\usepackage{algpseudocode}
\usepackage{enumitem}
\usepackage{todonotes}
\usepackage{etoolbox}
\usepackage{longtable}
\usepackage{pdflscape}
\usepackage{colortbl}
\usepackage[showonlyrefs]{mathtools}
\usepackage{amsthm}
\usepackage{tcolorbox}
\usepackage{comment}
\usepackage{bbm}
\usepackage[backend=biber,style=numeric,doi=true,url=false]{biblatex}
\usepackage{hyperref}
\tcbuselibrary{theorems}
\usepackage{lmodern}

\usepackage{geometry}
\definecolor{chestnut}{rgb}{0.8, 0.36, 0.36}
\definecolor{isabelline}{rgb}{0.96, 0.94, 0.93}

\newcommand{\var}{\operatorname{var}}
\newcommand{\cov}{\operatorname{cov}}

\newcommand{\corr}{\operatorname{corr}}
\newcommand{\spanning}{\operatorname{span}}
\newcommand{\diag}{\operatorname{diag}}
\newcommand{\covop}{\operatorname{\mathcal{C}}}
\newcommand{\e}{\mathbb{E}}

\newcommand{\la}{\langle}
\newcommand{\ra}{\rangle}
\newcommand{\ts}{\hspace*{0.1em}}

\definecolor{internationalorange}{rgb}{1.0, 0.4, 0.05}
\definecolor{darkmagenta}{rgb}{0.55, 0.0, 0.55}
\definecolor{teal}{rgb}{0.0, 0.5, 0.5}

\newcommand{\ba}{\mathbf{a}}
\newcommand{\bA}{\mathbf{A}}
\newcommand{\bb}{\mathbf{b}}
\newcommand{\bB}{\mathbf{B}}
\newcommand{\bC}{\mathbf{C}}

\newcommand{\bbf}{\mathbf{f}}
\newcommand{\bg}{\mathbf{g}}
\newcommand{\bG}{\mathbf{G}}

\newcommand{\bPi}{\mathbf{\Pi}}

\newcommand{\bu}{\mathbf{u}}
\newcommand{\bw}{\mathbf{w}}
\newcommand{\bW}{\mathbf{W}}
\newcommand{\cW}{\mathcal{W}}

\newcommand{\bone}{\mathbf{1}}
\newcommand{\bmu}{\boldsymbol{\mu}}
\newcommand{\bnu}{\boldsymbol{\nu}}
\newcommand{\bpi}{\boldsymbol{\pi}}

\newtheorem{theorem}{Theorem}[section]

\newtheorem{lemma}[theorem]{Lemma}

\newtheorem{definition}[theorem]{Definition}
\newtheorem{example}[theorem]{Example}
\newtheorem{remark}[theorem]{Remark}

\title{Spectral clustering of time-evolving networks using spatio-temporal random walks}

\author{
  Filip Bla\v skovi\' c\thanks{corresponding author, Zuse Institute Berlin, Berlin, Germany. \texttt{blaskovic@zib.de}} \and
  Tim O.F. Conrad\thanks{Zuse Institute Berlin, Berlin, Germany. \texttt{conrad@zib.de}} \and
  Stefan Klus\thanks{Heriot--Watt University, Edinburgh, UK. \texttt{S.Klus@hw.ac.uk}} \and
  Nata\v sa Djurdjevac Conrad\thanks{Zuse Institute Berlin, Berlin, Germany. \texttt{natasa.conrad@zib.de}}
}

\date{}

\begin{document}
\maketitle

\begin{abstract}
Temporal (or time-evolving) networks provide a natural framework for modeling complex systems with time-dependent interactions, where understanding the evolution of community structures is a central challenge. While random walk-based approaches to community detection in static networks are well established through the spectral analysis of associated transfer operators, extending these ideas to temporal networks is nontrivial due to the inherent time-dependence of the underlying dynamics. In this work, we develop a general framework for community detection in temporal networks that is based on \emph{multi-view canonical correlation analysis} (mCCA). We show that the proposed formulation admits a spectral characterization via a time-reversible random walk on an augmented space--time network, providing a clear dynamical interpretation of temporal communities as metastable structures of the process. Furthermore, we analyze key spectral properties of the resulting transfer operators and the interplay between spatial and temporal effects, which allows us to distinguish between structural features and artifacts induced by the snapshot coupling. Finally, we derive a reduced-order model, which preserves the essential spectral properties while significantly improving computational efficiency. We show that the proposed approach effectively detects communities in temporal networks and captures their evolution.
\end{abstract}


\section{Introduction}

Community detection is one of the central problems in network science. In static networks, spectral methods based on random walks provide a powerful framework for identifying communities as metastable regions of an underlying stochastic process. The key idea is that communities correspond to subsets of nodes between which transitions occur only rarely, resulting in slowly decaying dynamical modes that can be characterized through the spectral properties of associated transfer operators \cite{Huisinga_2006, Djurdjevac_2010, Schuette_2013, Sarich_2014, Klus_2023, Klus_2024_1}.

Many real-world networks, however, are inherently dynamic. Social interactions, mobility patterns, communication networks, and contact networks relevant for epidemic spreading evolve over time, giving rise to temporal (or time-evolving) networks \cite{Holme_2012, Masuda_2013, Mancastroppa_2020}. In epidemiological applications, identifying persistent and evolving communities is particularly important, as such structures can strongly influence transmission pathways and the effectiveness of intervention strategies. More generally, however, one seeks to identify communities that persist over extended periods while allowing for structural changes such as community splits, merges, or changes in membership. Extending the dynamical perspective on community detection from static to temporal networks is challenging, since the underlying dynamics are no longer governed by a single transfer operator but by a family of time-dependent operators.

While the notion of a community is well understood in static networks, its extension to temporal networks is considerably more subtle. Time-evolving systems are commonly represented as temporal networks, that is, sequences of graph snapshots that share a common vertex set but exhibit evolving edge structure. From a mathematical perspective, this setting raises several fundamental challenges: how to compare network structure across time, how to define notions of coherence and persistence, how to detect phases of stability during the network evolution, and how to extend spectral methods for community detection from static to temporal networks \cite{Rossetti_2019, Blaskovic_2025}. The goal is to cluster nodes across network snapshots into coherent groups, that is, subsets of nodes that remain densely connected over extended periods of time, see figure~\ref{fig:problem_illustration}.

\begin{figure}
    \centering
    \includegraphics[width=1\linewidth]{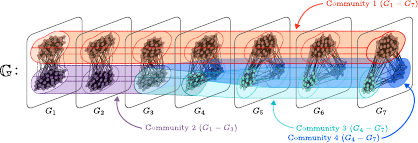}
    \caption{Our goal is to cluster nodes across network snapshots into subsets that remain densely connected over extended time intervals. The nodes highlighted in red remain densely connected internally and only sparsely connected to the rest of the network, forming a persistent community in $\mathbb G$ during the whole evolution of the network. In the first three snapshots, the remaining nodes form a second (purple) community. At time $t=4$ (snapshot $G_4$), this community splits into two smaller communities (light blue and dark blue), which then persist from snapshots $G_4$ to $G_7$. Our goal is to recover these coherent groups of nodes (red, purple, light blue, and dark blue).}
    \label{fig:problem_illustration}
\end{figure}

In this work, we adopt a dynamical perspective on community detection and interpret communities as coherent structures of an underlying random-walk process~\cite{Huisinga_2006, Froyland_2010, Djurdjevac_2012_phd, Schuette_2013, Sarich_2014, Padberg-Gehle_2017}. From a network-theoretic perspective, communities are subsets of nodes that are densely connected internally and only sparsely connected to the remainder of the network. Dynamically, they correspond to regions in which a random walker becomes temporarily trapped. Starting from such a region, the probability of remaining within it is high, while transitions to other parts of the network occur only rarely. This behavior is commonly referred to as \emph{metastability}. More precisely, a random process is metastable if the state space can be partitioned into subsets such that mixing within each subset occurs on a much shorter timescale than transitions between different subsets.

In static networks, this separation of timescales is reflected in the spectral properties of transfer operators associated with the random-walk process~\cite{von_Luxburg_2007, Djurdjevac_2010, Klus_2023, Klus_2024_1}. In particular, for time-reversible random walks, the Koopman operator, which describes the evolution of real-valued functions of the nodes (network observables), is self-adjoint with respect to a weighted inner product. This implies a real spectrum and an orthogonal basis of eigenvectors. The dominant eigenvalues close to one correspond to slowly decaying modes of the dynamics. The associated eigenvectors can be interpreted as network observables that are approximately invariant under the action of the Koopman operator and whose level sets partition the nodes into metastable subsets of the underlying random walk, i.e., into communities~\cite{Huisinga_2006, Budisic_2012, Djurdjevac_2012_phd, Schuette_2013, Klus_2016, Klus_2024}. This operator-theoretic viewpoint provides a principled way to detect and interpret community structures via the dynamics of the underlying random-walk process.

Here, we study the evolution of communities in temporal networks in the same manner, where the process now propagates across snapshots according to the underlying network structure. However, extending the dynamical perspective on community detection from static to temporal networks is non-trivial. Dynamical processes on networks whose structure changes over time are time-dependent and are therefore described by a whole family of Koopman operators. As a consequence, the rich set of tools developed for community detection in static networks is not directly applicable in the temporal setting. Early approaches to community detection in temporal networks proceed by detecting communities independently within each snapshot and matching them across time \cite{Aynaud_2013, Rossetti_2019}. However, such methods face several fundamental challenges, including the ambiguity of community labels, the instability of static community detection algorithms, which may produce significantly different partitions even under minor or no changes in the network \cite{Aynaud_2010}, as well as the dependence on ad hoc matching procedures and parameter tuning. 

In this paper, we turn to \emph{Canonical Correlation Analysis} (CCA), a procedure that aims to find linear transformations of random variables that maximize the correlation between them~\cite{Hotelling_1936}. We note that the computation of dominant eigenvectors of transfer operators associated with static networks can be reformulated as an instance of CCA, in which one seeks to maximize the correlation between observables at different time steps of the random process~\cite{Klus_2024}. We build on the extension of this approach proposed in~\cite{Trower_2025}, where a generalization of CCA, known as multi-view CCA~\cite{Shawe-Taylor_2004}, is employed. A temporal network is represented as a sequence of static networks (snapshots), for example obtained by ``freezing'' the structure of the temporal network at equally spaced time intervals with a desired frequency. Multi-view CCA is then used to maximize a weighted correlation between observables across adjacent snapshots whose level sets are later used to cluster nodes from different snapshots into communities that persist over time. We extend this framework by introducing a general coupling scheme that allows for incorporating correlations between arbitrary pairs of snapshots with different relative weights into the overall objective function. This leads to a more robust method that is more resilient to noisy fluctuations and reduces the emergence of highly correlated observables across snapshots whose level sets do not correspond to meaningful temporal communities. Moreover, this approach enables flexible adaptation to different application scenarios, while preserving the underlying framework. For instance, one can explore whether the data exhibits periodic behavior through a cyclic coupling of the last snapshots of the temporal network to the first snapshots.

The multi-view CCA approach enables us not only to detect sets of nodes that remain densely interconnected over longer time periods, but also to capture changes in community structure throughout the evolution. In contrast to classical CCA, which correlates only two selected time points, this method incorporates the full temporal information and thus retains insight into intermediate structural changes that would otherwise remain invisible when considering only particular time points, such as the beginning and the end of the network evolution.

We show that the multi-view CCA approach with a general coupling scheme reduces to the analysis of spectral properties of a time-reversible random walk on a static network defined on an augmented node set, whose structure encodes both the spatial organization of the temporal network at different time points and the temporal ordering of snapshots. In this representation, a random walker can move not only in space (between nodes) but also through time (between snapshots), which naturally connects our approach to community detection methods for static networks and provides a clear and intuitive interpretation of temporal communities as coherent sets of the associated spatio-temporal random walk.

Our perspective is related to a common approach in the literature, where temporal networks are represented as space--time (multilayer) networks by introducing inter-snapshot edges~\cite{Kivela_2014}, and random processes on such augmented networks are analyzed. In many constructions, only identical nodes are connected across snapshots, typically with uniform inter-snapshot weights~\cite{Kivela_2014, Bazzi_2016, Froyland_2024}, and connections are often restricted to consecutive snapshots~\cite{Taylor_2017}. As a result, the chain-like geometry of the snapshot coupling scheme gives rise to slowly decaying modes, which may be misinterpreted as metastable behavior even in the absence of any intrinsic community structure. Some models allow connections between different nodes across snapshots, but they either rely entirely on user-defined parameters \cite{De_Domenico_2013} or disregard the temporal ordering of snapshots, thereby identifying communities in an aggregated network rather than over contiguous time intervals \cite{Aslak_2018}. In contrast, the space--time network induced by the snapshot coupling and the structure of the temporal network provides a more principled construction. The weights of inter-snapshot connections are determined by the structure of the snapshots themselves and modulated by a flexible notion of temporal proximity. This allows us to simultaneously capture structural and temporal relationships, while at the same time controlling and reducing spurious metastability effects arising from the geometry of the coupling scheme.

The key observation of this paper is that a broad class of multi-view CCA formulations for temporal community detection is equivalent to the spectral analysis of a reversible Markov process on an augmented space-time network. This equivalence enables a rigorous analysis of temporal communities using transfer-operator techniques and provides a principled framework for distinguishing genuine structural features from artifacts induced by temporal coupling.

Our main contributions are threefold:
\begin{itemize}
    \item[(i)] First, we derive a general multi-view CCA model for temporal networks based on Perron--Frobenius and Koopman operators and show that this formulation reduces to a generalized eigenvalue problem. Moreover, we establish a connection to a random walk on an augmented static network representing the original temporal network. This perspective complements existing notions of spatio-temporal random walks on temporal networks.
    \item[(ii)] Second, we analyze the spectral properties of the resulting spatio-temporal dynamics and investigate how temporal effects induced by the snapshot coupling interact with spatial effects arising from the network structure. A principled interpretation of the associated eigenvalues and eigenvectors has largely been missing in the literature, yet is essential for separating spatial and temporal contributions, that is, for distinguishing between effects that arise purely from the coupling mechanism and those that encode meaningful structural events such as community splits, merges, or changes in community membership.
    \item[(iii)] Lastly, we derive a reduced model of the spatio-temporal random walk for community structure analysis. We prove that the aforementioned spectral properties transfer to the projected setting under mild conditions and show that the resulting reduced model remains effective in practice. Since the dimensionality of the full model grows quadratically with both the number of nodes and the number of snapshots, computations quickly become expensive even for moderately large datasets, making this reduction a natural and useful extension. In addition, we provide spectral error bounds that quantify the accuracy of the reduced representation.
\end{itemize}

We first introduce the necessary theoretical background on random walks and transfer operators in section~\ref{sec:theoretical_background}. We then develop the multi-view CCA framework and analyze its spectral properties in section~\ref{sec:community_detection_in_temporal_networks}. We present the reduced-order modeling approach and numerical results in sections~\ref{sec:model_reduction} and~\ref{sec:numerical_examples}. Finally, we provide a brief conclusion and suggest directions for future research in section~\ref{sec:conclusion}.

\section{Theoretical background}
\label{sec:theoretical_background}

In this section, we provide a brief overview of the required concepts. Let $\mathbb{G} = (G_1, \dots, G_M)$ be a temporal network defined over an observation period partitioned into $M$ discrete time steps. Each index $t \in \{1, \dots, M\}$ corresponds to a static snapshot $G_t = (V, E_t, \omega_t)$. In this formulation, the node set $V = \{v_1, \dots, v_N\}$ is fixed across all snapshots, the edge set $E_t \subseteq V\times V$ represents the undirected interactions occurring during time step $t$, and $\omega_t: E_t\rightarrow\mathbb R$ denotes a symmetric function assigning weights to the edges of $G_t$. The (weighted) degree of node $v_i$ of $G_t$ is given by
\begin{equation*}
\deg_t(v_i) = \sum_{j=1}^N \omega_t(v_i,v_j),
\end{equation*}
and the corresponding diagonal degree matrix is defined by
\begin{equation*}
D_t = \operatorname{diag}(\deg_t(v_1), \dots, \deg_t(v_N)).
\end{equation*}

\subsection{Random walks on temporal networks}
\label{sec:random_walks_on_temporal_networks}

On a static network, a standard random walk process is defined as a discrete-time Markov chain that moves between adjacent nodes, where the transition probabilities are proportional to the edge weights. Since the edge set is fixed, the transition rules do not depend on time but only on the current state of the process. We extend this notion from static to temporal networks and consider a random process $\{X_t\}_{t=1}^M$ evolving across the snapshots of a temporal network $\mathbb G$.

Let $A_t\in\mathbb R^{N\times N}$ denote the weighted adjacency matrix of $G_t$ where $(A_t)_{ij}=\omega_t(v_i,v_j)$. Each snapshot $G_t$ induces its own transition matrix
\begin{equation}
\label{eq:transition_matrix}
    S_t = D_t^{-1} A_t,
\end{equation}
which governs the random walk dynamics within that snapshot. Starting in snapshot $G_1$, the process evolves forward in time by moving to the next snapshot at each time step and following the transition rules determined by the transition matrix $S_t$ at time $t$. As a consequence, the resulting process is in general time-inhomogeneous and may no longer be reversible. Therefore, many classical methods for community detection developed for static networks are not directly applicable in this setting. Nevertheless, the underlying intuition remains similar. Although edges in a temporal network may appear or disappear over different snapshots, if a subset of nodes remains densely connected (coherent) throughout the network's evolution, then a random walk process starting in this region will with high probability remain confined to it for as long as this structure persists. This observation is consistent with the interpretation of communities as metastable sets of random processes in static networks.

In what follows, we adopt a transfer operator perspective on network dynamics and extend it from static to temporal networks.

\begin{definition}[Observables and probability distributions]
\label{def:observables}
Let 
\[
\mathbb U = \{ f \mid f \colon V \to \mathbb{R} \}
\]
denote the space of real-valued functions defined on the node set $V$. 
A function $\rho \in \mathbb U$ is called a \emph{probability distribution} if $ \rho(v_i) \ge 0 
$, for all $i = 1, \dots, N $, and
\[
\sum_{i=1}^{N} \rho(v_i) = 1.
\]
In general, functions $f\in\mathbb U$ are referred to as \emph{network observables}.
\end{definition}

\begin{definition}[Standard and weighted inner products]
Let $ f,g \in \mathbb U $ be two network observables. The standard inner product on $\mathbb U$ is given by
\[
\langle f, g \rangle 
=
\sum_{i=1}^{N} f(v_i)\, g(v_i).
\]
Given a probability distribution $\mu$, the corresponding $\mu$-weighted inner product on $\mathbb U$ is defined by
\[
\langle f, g \rangle_{\mu}
=
\sum_{i=1}^{N} f(v_i)\, g(v_i)\, \mu(v_i).
\]
\end{definition}

Analogously to random processes on static networks, the dynamics of the random process $\{X_t\}_{t=1}^M$ on a temporal network can be described from two dual perspectives: through the evolution of its probability distributions and through the evolution of network observables across time snapshots \cite{Otto_2021}. For distinct time points $t$ and $s$ with $t < s$, we now define two transfer operators.

\begin{definition}[Multi step Perron--Frobenius and Koopman operators] 
\label{def:multi-step_transfer_operators}
    Let $\rho \in \mathbb{U}$ be a probability distribution and let $f\in\mathbb U$ be a network observable. Let $S_{ts}=S_tS_{t+1}\cdots S_{s-1}$ for $t<s$, where the entry $(S_{ts})_{ij}$ is the probability that a random walker starting in node $v_i$ at time $t$ will end up in node $v_j$ at time~$s$.
    \begin{itemize}
    \item[i)] The \emph{multi step Perron--Frobenius operator} between times $t$ and $s$, denoted by $\mathcal{P}_{ts} \colon \mathbb{U} \rightarrow \mathbb{U}$, is defined by
    \[
    \mathcal{P}_{ts} \rho(v_i)=\sum_{j=1}^N (S_{ts})_{ji} \, \rho(v_j).
    \]
    \item[ii)] The \emph{multi step Koopman operator} between times $t$ and $s$, denoted by $\mathcal{K}_{ts} \colon \mathbb{U} \rightarrow \mathbb{U}$, is defined by
    \[
    \mathcal{K}_{ts} f(v_i) = \sum_{j=1}^N (S_{ts})_{ij} f(v_j)=\e[f(X_{s})\mid X_t=v_i].
    \]
\end{itemize}
\end{definition}

Matrix representations $P_{ts} \in \mathbb{R}^{N\times N}$ and $K_{ts} \in \mathbb{R}^{N\times N}$ of the operators $ \mathcal{P}_{ts} $ and $ \mathcal{K}_{ts} $, respectively, are given by
\begin{equation}
\label{eq:koopman_mat_rep}
P_{ts} = S_{ts}^{\top} \quad\text{and}\quad K_{ts} = S_{ts}.
\end{equation}
Suppose the ensemble of random walkers is initialized on the first snapshot $G_1$ with probability distribution $\mu_1$. The evolution of probability distributions is then given by
\begin{equation}
\label{eq:evolution_of_prob_distr}
\mu_{s} = \mathcal{P}_{ts} \ts \mu_t,
\end{equation}
where $\mu_t$ and $\mu_s$ denote the probability distributions of the process at times $t$ and $s$ (on snapshots $G_t$ and $G_s$). On the level of observables, the dynamics are described by the family of Koopman operators $\mathcal{K}_{ts}$. Applied to an observable $f$ at time $t$, the operator $\mathcal{K}_{ts}$ propagates it $s-t$ steps forward according to the random walk dynamics between times $t$ and $s$ and computes its expected values at the network nodes at time $s$. We now introduce a reweighted Perron--Frobenius operator that propagates observables from time $t$ to time $s$ with respect to the reference distributions $\mu_t$ and $\mu_s$.

\begin{definition}[Multi step Reweighted Perron--Frobenius operator]
\label{def:multi_step_reweighted_PF}
Let $u \in \mathbb{U}$ and let $S_{ts} = S_t S_{t+1} \cdots S_{s-1}$ for $t < s$.  
Let $\mu_t$, $1 \le t \le M$, denote the probability distributions of the random process on the temporal network at time $t$. The \emph{multi step reweighted Perron--Frobenius operator} between times $t$ and $s$, denoted by $\mathcal{T}_{ts} : \mathbb{U} \to \mathbb{U}$, is defined by
\[
\mathcal{T}_{ts} u(v_i)
=
\frac{1}{\mu_s(v_i)}
\sum_{j=1}^{N} (S_{ts})_{ji} \, \mu_t(v_j) \, u(v_j)
=
\mathbb{E}[u(X_t) \mid X_s = v_i].
\]
\end{definition}

Let $\bmu_t=[\mu_t(v_1),\dots,\mu_t(v_N)]^{\top}$ denote the vector representation of the probability distribution $\mu_t$, and let $D_{\mu_t} = \diag(\bmu_t)$. The matrix representation of the reweighted Perron--Frobenius operator $T_{ts} \in \mathbb{R}^{N \times N}$ is given by
\begin{equation}
\label{eq:reweighted_perron_frobenius_mat_rep}
T_{ts} = D_{\mu_s}^{-1} (S_{ts})^\top D_{\mu_t}.
\end{equation}

The operators $\mathcal{K}_{ts}$ and $\mathcal{T}_{ts}$ (and their respective matrix representations $K_{ts}$ and $T_{ts}$) are adjoint with respect to the inner products induced by the probability distributions $\mu_t$ and $\mu_s$, that is,
\[
\la f, \mathcal K_{ts} g\ra_{\mu_t}=\la \mathcal T_{ts} f,g\ra_{\mu_{s}}.
\]

\begin{remark}
\label{rem:single_vs_multi_step_tr_op}
 Definitions~\ref{def:multi-step_transfer_operators} and~\ref{def:multi_step_reweighted_PF} generalize the standard Perron--Frobenius and Koopman operators for discrete-time dynamics \cite{Klus_2023}, which describe the one-step propagation of probability distributions and observables and can be recovered from our definition by $s=t+1$. In \cite{Trower_2025}, such one-step operators were used to model propagation only between consecutive snapshots. Here, we additionally introduce operators that characterize propagation across multiple snapshots.
\end{remark}

\section{Community detection in temporal networks}
\label{sec:community_detection_in_temporal_networks}
Guided by the intuition from the static network case, we seek observables for each snapshot that remain highly correlated under the underlying random-walk process. More precisely, the goal is to identify observables that preserve information about the coherent behavior of the time-inhomogeneous random walk, such that their level sets partition the nodes across all snapshots into communities that persist over time. This perspective naturally leads to a multi-view extension of the CCA formulation~\cite{Shawe-Taylor_2004}, where we look for a collection of observables $f_t$ for $1 \le t \le M$ (i.e., one observable for each snapshot), that maximize a weighted correlation between the random variables $f_1(X_1), \dots, f_M(X_M)$. In this section, we formalize this idea and begin by defining correlations between network observables at different snapshots.

\subsection{Multi-view canonical correlation analysis (mCCA)}
\label{sec:covariance_operators_and_weighted_correlation_optimization}

To quantify and maximize temporal correlations, we introduce covariance and cross-covariance operators associated with the underlying random process. Building on \cite{Klus_2024, Trower_2025}, we generalize the notion of cross-covariance by defining it between arbitrary time points $t$ and $s$.

\begin{definition}[Covariance and cross-covariance operators]
\label{def:covariance_operators}
Let $ f_t \in \mathbb{U} $, with $ 1\le t\le M $, denote a network observable for the snapshot $G_t$. Let $\mu_t$ be the probability distribution at time $t$ and let $S_{ts}$ denote the transition probability matrix from time $t$ to time $s$ as given in definition~\ref{def:multi-step_transfer_operators}.
\begin{itemize}
    \item[i)] The \emph{covariance operator} at time $t$, $\mathcal{C}_{tt} \colon \mathbb{U} \rightarrow \mathbb{U}$, is defined as
    \[
    \mathcal{C}_{tt} f_t(v_i) = \mu_t(v_i) f_t(v_i).
    \]
    \item[ii)] The \emph{cross-covariance operator} between times $t$ and $s$ for $t<s$, $\mathcal{C}_{ts} \colon \mathbb{U} \rightarrow \mathbb{U}$, is defined as
    \[
    \mathcal{C}_{ts} f_s(v_i) = \sum_{j=1}^N \mu_t(v_i) (S_{ts})_{ij} f_s(v_j).
    \]
\end{itemize}
\end{definition}

We can now express the covariance of an observable $f_t$ at time $t$ as 
\[
\langle f_t, \covop_{tt} f_t \rangle
=
\sum_{i=1}^N \mu_t(v_i) f_t(v_i)^2
=
\mathbb{E}_{\mu_t}[f_t(X_t)^2]
=
\var(f_t)
\]
and the cross-covariance of the observables $f_t$ and $f_s$ at times $t$ and $s$ as
\[
\langle f_t, \covop_{ts} f_s \rangle
=
\sum_{i=1}^N \sum_{j=1}^N 
f_t(v_i)\,\mu_t(v_i)\,(S_{ts})_{ij}\,f_s(v_j)
=
\mathbb{E}_{\mu_{ts}}
\big[ f_t(X_t) f_s(X_{s}) \big]
=
\cov(f_t,f_s),
\]
where 
$
\mu_{ts}(v_i,v_j)
=
\mathbb{P}(X_t = v_i,\, X_{s} = v_j)
=
\mu_t(v_i)\, (S_{ts})_{ij}
$
is the joint probability distribution of the random process at times $t$ and $s$. 

Let $\bbf_t$ and $\bbf_s$ denote the vector representations of the observables $f_t$ and $f_s$ for $t<s$ and let $C_{tt}$ and $C_{ts}$ be the matrix representations of the operators $\mathcal C_{tt}$ and $\mathcal C_{ts}$. Then
\[
\var(f_t)
=
\langle \bbf_t, C_{tt}\bbf_t\rangle
=
\bbf_t^\top C_{tt}\bbf_t
\quad
\text{and}
\quad
\cov(f_t,f_s)
=
\langle \bbf_t, C_{ts}\bbf_s\rangle
=
\bbf_t^\top C_{ts}\bbf_s.
\]
The correlation between observables $f_t$ and $f_s$ is therefore given by
\begin{equation}
\label{eq:correlation}
\corr(f_t,f_s)
=
\frac{\bbf_t^\top C_{ts}\bbf_s}
{\sqrt{\bbf_t^\top C_{tt}\bbf_t}\,
 \sqrt{\bbf_s^\top C_{ss}\bbf_s}}.
\end{equation}
Our objective is to maximize the weighted correlation across all snapshots
\begin{equation}
\label{eq:weighted_correlation}
\max_{f_t\in\mathbb U}
\sum_{\substack{t,s=1 \\ t < s}}^{M}
w_{ts} \, \corr(f_t, f_{s}),
\end{equation}
where $w_{ts}\ge0$, for $1\le t<s\le M$, are coupling weights that quantify the importance of the correlation between observables at times $t$ and $s$. Although only the weights with $t<s$ enter the objective function \eqref{eq:weighted_correlation}, we extend them for convenience to all pairs of indices by defining $w_{st}=w_{ts}$ and $w_{tt}=0$. This convention simplifies several expressions and derivations below. Since our aim is to identify temporal communities and detect events corresponding to structural changes, the coupling weights should encode temporal proximity: snapshots that are close in time should be coupled with higher weights, while those farther apart receive progressively lower weights. 
As persistent communities are characterized by stability over longer periods, coupling only consecutive snapshots may therefore make the solutions of \eqref{eq:weighted_correlation} sensitive to short-term fluctuations and noise. By coupling observables across longer time intervals, we obtain more robust and temporally consistent solutions. Furthermore, depending on the application, alternative snapshot coupling schemes may be more appropriate, e.g. couplings between temporally distant snapshots can be used to reveal recurring or periodic community patterns. We will illustrate these effects through numerical examples in section~\ref{sec:numerical_examples}.

Since the correlation \eqref{eq:correlation} is invariant under scalar multiplication of observables, i.e.,
\[
\corr(f_t,f_s)=\corr(a f_t,b f_s)
\quad \text{for all } a,b\in\mathbb{R},
\]
we may assume, without loss of generality, that all observables $f_t$, with $1\le t\le M$, satisfy $\var(f_t)=1$. We can then rewrite \eqref{eq:weighted_correlation} in the form
\begin{equation}
\label{eq:maximization_problem}
\max_{\bbf_t}\sum_{\substack{t,s=1 \\ t < s}}^{M}w_{ts}\bbf_t^\top C_{ts}\bbf_{s}
\hspace{4mm}\text{ s.t.}\hspace{1mm}\bbf_t^\top C_{tt}\bbf_t=1, \text{ for all } t=1,...,M.
\end{equation}
Following the approach from \cite{Trower_2025}, we use Lagrange multipliers to solve this optimization problem, resulting in
\begin{equation*}
    \mathcal{L}(\bbf_1,\lambda_1,...,\bbf_M,\lambda_M)=\sum_{\substack{t,s=1 \\ t < s}}^{M}w_{ts}\bbf_t^\top C_{ts}\bbf_s -\sum_{t=1}^{M}\frac{r_t}{2}\lambda_t(\bbf_t^\top C_{tt}\bbf_t-1), 
\end{equation*}
where we rescale each Lagrange multiplier $\lambda_t$ by the factor $ \frac{r_t}{2} $, with $r_t=\sum_{s=1}^{M} w_{ts}$. Here, $r_t$ serves as a normalization factor that accounts for the total weight assigned to snapshot $G_t$, while the factor $\frac{1}{2}$ cancels out when differentiating the quadratic term, which simplifies the subsequent computations.
 
For any observables $\bbf_1,\dots,\bbf_M$ that solve \eqref{eq:maximization_problem}, the condition $\partial_{\bbf_t}\mathcal{L} = 0$ must hold for all $1\le t\le M$. Taking derivatives of the Lagrangian function with respect to $\bbf_1,\dots,\bbf_M$, and using $w_{ts}=w_{st}$ we obtain the equations
\begin{equation}
\label{eq:derivatives}
    \begin{split}
        \partial_{\bbf_1}\mathcal L&= w_{12} C_{12}\bbf_2+w_{13}C_{13}\bbf_3+\dots+w_{1M}C_{1M}\bbf_M-r_1\lambda_1C_{11}\bbf_1=0,\\
        \partial_{\bbf_2}\mathcal L&= w_{21} C_{12}^\top\bbf_1+ w_{23} C_{23}\bbf_3+\dots+w_{2M}C_{2M}\bbf_M-r_2\lambda_2C_{22}\bbf_2=0,\\
        \vdots\\
        \partial_{\bbf_M}\mathcal L&= w_{M1}C_{1M}^\top\bbf_1 + w_{M2}C_{2M}^\top \bbf_2+\cdots+w_{M(M-1)}C_{(M-1)M}^\top\bbf_{M-1}-r_M\lambda_MC_{MM}\bbf_M=0.
    \end{split}
\end{equation}

Allowing independent multipliers $\lambda_t$ would lead to a multiparameter eigenvalue problem \cite{Atkinson_1972, Eisenmann_2025}. Such problems fall outside classical spectral theory and typically require more advanced analytical and numerical techniques. While this would be an interesting direction for future research, it is beyond the scope of the present work. 

In what follows, we  assume that $\lambda_t = \lambda$ for all $t=1,\dots,M$, which allows us to rewrite the system in the form of a generalized eigenvalue problem. We introduce the \emph{snapshot coupling matrix} $H=[h_{ts}]\in\mathbb R^{M\times M}$ containing the normalized coupling weights between snapshots 
\begin{equation}
\label{eq:compute_H}
h_{ts} = \frac{w_{ts}}{r_t},
\qquad 1 \le t,s \le M .
\end{equation}
For each $t$, dividing the corresponding equation in \eqref{eq:derivatives} by $r_t$ and rewriting the system in matrix form, we obtain 
\begin{equation}
\label{eq:generalized_eigval_problem}
\bA_H \bbf = \lambda \bB \bbf, 
\end{equation}
where
\begin{equation*}
\bA_H =
\begin{bmatrix}
0 & h_{12}C_{12} & \cdots & h_{1M}C_{1M} \\
h_{21}C_{12}^{\top} & 0 & \ddots & \vdots\\
\vdots & \ddots & \ddots & h_{(M-1)M}C_{(M-1)M} \\
h_{M1}C_{1M}^\top & \cdots & h_{M(M-1)}C_{(M-1)M}^{\top} & 0
\end{bmatrix},
\end{equation*}
\begin{equation*}
\bB =
\begin{bmatrix}
C_{11} & 0 & \cdots & 0 \\
0 & C_{22} & \ddots & \vdots \\
\vdots & \ddots & \ddots & 0 \\
0 & \cdots & 0 & C_{MM}
\end{bmatrix} \,\,\, \text{and} \,\,\, 
\bbf=
\begin{bmatrix}
    \bbf_1\\ 
    \bbf_2\\
    \vdots\\
    \bbf_M
\end{bmatrix}.
\end{equation*}
From definition~\ref{def:covariance_operators}, the matrix representations of the covariance and cross-covariance operators are given by
\[
C_{tt} = D_{\mu_t}
\quad \text{and} \quad
C_{ts} = D_{\mu_t} S_{ts}.
\]
Defining $\bC_H=\bB^{-1}\bA_H$, and using~\eqref{eq:koopman_mat_rep} and~\eqref{eq:reweighted_perron_frobenius_mat_rep}, we obtain

\begin{equation}
\label{eq:k_pf_notation}
    \bC_H=
    \begin{bmatrix}
        0 & h_{12}K_{12} & \cdots & h_{1M}K_{1M}\\
        h_{21}T_{12} & 0 & \ddots & \vdots \\
        \vdots & \ddots & \ddots & h_{(M-1)M}K_{(M-1)M}\\
        h_{M1}T_{1M} & \dots & h_{M(M-1)}T_{(M-1)M} & 0
    \end{bmatrix},
\end{equation}

so that \eqref{eq:generalized_eigval_problem} reduces to the standard eigenvalue problem
\[
\bC_H \bbf = \lambda \bbf.
\]

\begin{remark}
    For $M=2$, we note that $\lambda_1=\lambda_2$ is always satisfied. To see this, we multiply the expressions $\partial_{\bbf_1}$ and $\partial_{\bbf_2}$ from the left by $\bbf_1^{\top}$ and $\bbf_2^{\top}$, respectively, and use
    \[
    \bbf_1^{\top}C_{12}\bbf_2=\bbf_2^{\top}C_{12}^{\top}\bbf_1.
    \]
    We therefore adopt it as a simplifying assumption that allows us to identify clusters by solving an eigenvalue problem. Consequently, solutions of the simplified problem do not necessarily solve the original maximization problem~\eqref{eq:maximization_problem}. However, this restriction comes with the benefit of a natural interpretation and a direct connection to spectral clustering methods as we will show below.
\end{remark}

We can interpret the matrix $\bC_H$ from a transfer-operator perspective. In the standard theory of (one-step) transfer operators in discrete-time dynamics (see remark~\ref{rem:single_vs_multi_step_tr_op}), the Koopman operator evaluates an observable at the current state by taking the conditional expectation of its value at the next state. In this sense, it effectively pulls future information back to the present and is therefore often referred to as a \emph{pull-back} operator. In contrast, the reweighted Perron--Frobenius operator propagates information forward in time by determining the values of an observable at the next state from its values at the current state. It thus transports information forward along the dynamics and is commonly referred to as a \emph{push-forward} operator. These interpretations naturally extend to the multi-step transfer operators introduced here. Acting blockwise with $\bC_H$ on $\bbf$, the $t$th segment of size $N$ of $\bC_H\bbf$ is given by
\begin{equation*}
(\bC_H \bbf)_t
=
\sum_{s=1}^{t-1} h_{ts}\, T_{st}\bbf_s
+
\sum_{s=t+1}^{M} h_{ts}\, K_{ts}\bbf_s,
\end{equation*}
i.e., the observable $\bbf_t$ is transformed into a weighted average of pushed-forward and pulled-back observables from different times, transported to time $t$. Therefore, eigenvectors of $\bC_H$ corresponding to large eigenvalues represent observables that are closest to being invariant under the time-inhomogeneous random process described in section~\ref{sec:theoretical_background}. The relative contribution of past and future snapshots to $\bbf_t$ is given by the coupling matrix $H$.

\subsection{Spatio-temporal random walk}
\label{sec:Spatio-temporal_random_walk}

In this section, we study properties of the matrices $H$ and $\bC_H$ and relate them to random walks between snapshots and on the augmented space--time network. Proofs of all the lemmas stated in this section are provided in appendix~\ref{sec:appendix_1}.

\begin{definition}
We define the \emph{snapshot coupling network} to be an undirected weighted network $G_H = (V_H, E_H, \omega_H)$, with the set of nodes $V_H = \{1,\dots,M\}$, the set of edges $E_H \subset V_H \times V_H$ such that $(t,s)\in E_H$ if and only if $w_{ts} \neq 0$ for any $1 \le t < s \le M$ and the symmetric edge-weight function is given by $\omega_H(t,s)=w_{ts}$. 
\end{definition}

\begin{lemma}
\label{le:H_random_walk}
Let $H$ be the coupling matrix defined in \eqref{eq:compute_H}. Then, there exists a probability distribution $\bpi \in \mathbb{R}^M$ such that the detailed balance condition
\[
D_{\pi} H = H^{\top} D_{\pi}
\]
holds, where $D_{\pi} = \diag(\bpi)$. Consequently, $H$ is the transition matrix of a time-reversible random walk on the coupling network $G_H$. In particular, if $G_H$ is connected, then $\bpi$ is unique.
\end{lemma}

The network $G_H$ can be interpreted as a network defined on the snapshots of the temporal network $\mathbb{G}$, with edge weights given by the snapshot coupling weights $w_{ts}$. Based on the previous lemma, $H$ is thus the transition matrix of the associated random walk on the network $G_H$. 

\begin{lemma}
\label{thm:row_stochasticity}
The matrix $\bC_H$ is row-stochastic. In particular, its spectral radius satisfies
\[
\rho(\bC_H) = 1.
\]
\end{lemma}

Using this result, we can interpret the matrix $\bC_H$ as the transition matrix of a random walk on an augmented state space, where transitions between snapshots both forward and backward in time are possible. Since this process simultaneously captures transitions across nodes (space) and across snapshots (time), we refer to it as a \emph{spatio-temporal random walk}, defined on the augmented network described in definition~\ref{def:space_time_network}. Similar space--time formulations have been used in transfer-operator approaches to coherent-set detection, where coherent sets are identified through the spectral analysis of a process on an augmented space--time domain~\cite{Fackeldey_2019}.

\begin{definition}
\label{def:space_time_network}
The \emph{space--time network} $\mathscr G_H$ is defined on the augmented node set $V \times \{1,\dots,M\}$, so that it has $MN$ nodes in total. Each node is given by a pair $(v_i,t)$, where $1 \le i \le N$ corresponds to a node of $G_t$. Edges connect nodes $(v_i,t)$ and $(v_j,s)$ between snapshots $G_t$ and $G_s$ with nonzero weights if and only if the transition probability between them, given by $\bC_H$, is positive. 
\end{definition}

We will use the terms spatio-temporal random walk on a temporal network and on its associated space--time network interchangeably. A formal definition of $\mathscr G_H$ and further details are given in lemma~\ref{lem:multilayer_network} in appendix~\ref{sec:appendix_2}. An illustration describing the dynamics of the spatio-temporal random walk for a simple temporal network is shown in figure~\ref{fig:spatio_temporal_rw}, following the structure of $\bC_H$ given in \eqref{eq:k_pf_notation}. Starting from a node $v_i$ at snapshot $G_t$ (marked in red in figure~\ref{fig:spatio_temporal_rw}\textbf{a} for $t=3$), the walker first selects a target snapshot $G_s$ (figure~\ref{fig:spatio_temporal_rw}\textbf{b}) according to the coupling matrix $H$, that is, according to a random-walk process on the coupling network $G_H$ (figure~\ref{fig:spatio_temporal_rw}\textbf{c}). Then, it moves to a node within the chosen snapshot according to the transition probabilities given by $K_{ts}$ or $T_{ts}$ depending on the choice of $s$ 
(figure~\ref{fig:spatio_temporal_rw}\textbf{d}). If $t<s$, the spatial transition is governed by the multi-step transition matrices $K_{ts}$, corresponding to the forward propagation of the network observables. 
If $t > s$, it is governed by $T_{ts}$, the adjoint of the former, which also admits the interpretation of backward propagation of observables with respect to the time-inhomogeneous process $\{X_t\}_{t=1}^M$ introduced in section~\ref{sec:random_walks_on_temporal_networks}. The spatio-temporal random walk can therefore be decomposed into two components: a temporal component governed by the transition matrix $H$ describing transitions between snapshots, and a spatial component governed by the transition matrices $K_{ts}$ or $T_{ts}$ describing transitions between nodes. We formalize this decomposition in lemma~\ref{lem:spatial_temporal_part} in appendix~\ref{sec:appendix_1}.

The following lemma additionally shows that the spatio-temporal random walk is time-reversible, allowing us to apply existing results to further investigate its properties in section~\ref{sec:spectral_clustering}.

\begin{figure}
    \centering
    \includegraphics[width=1\linewidth]{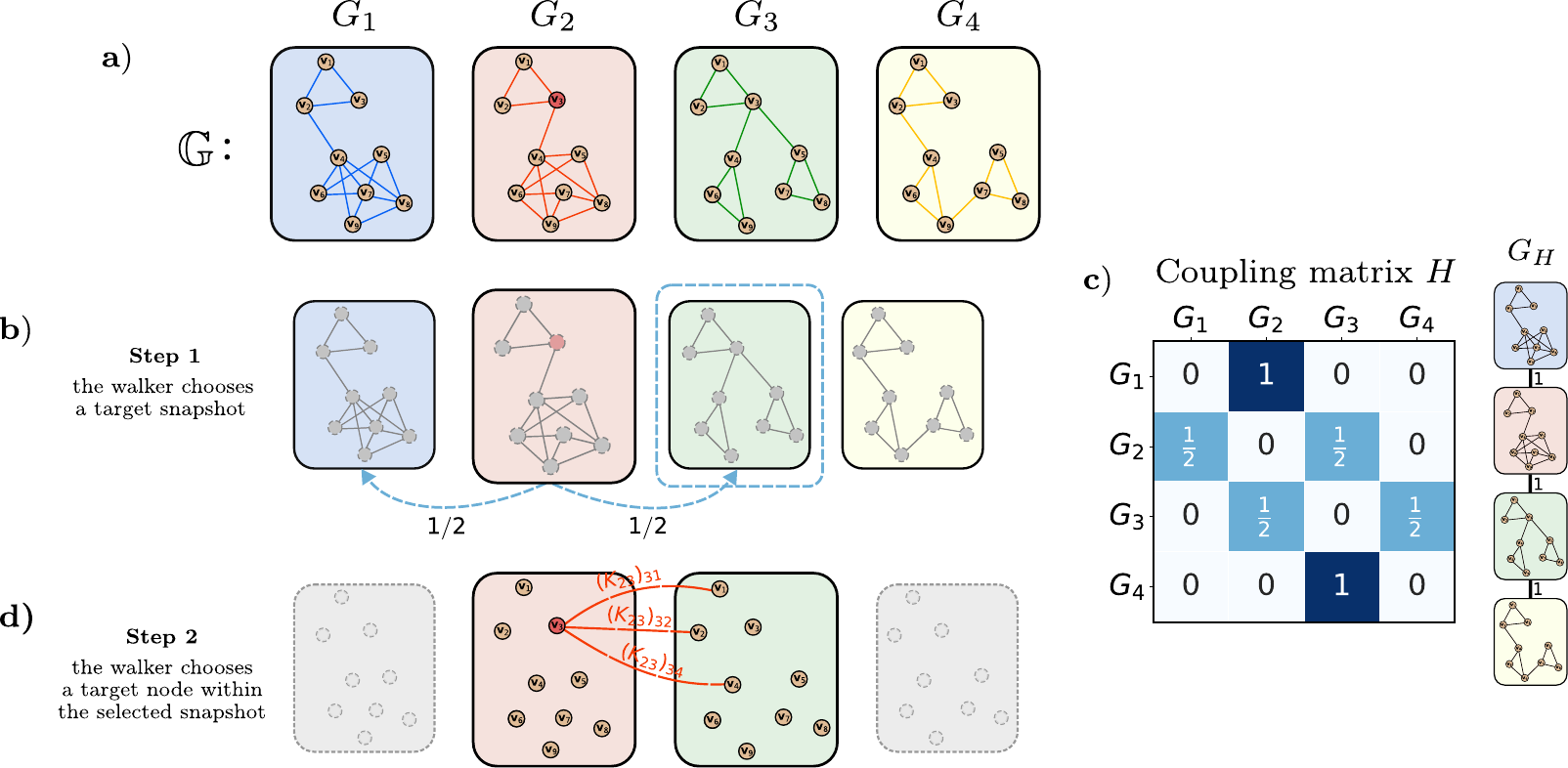}
    \caption{An illustration of a spatio-temporal random walk on the temporal network $\mathbb G$. \textbf{a)} A temporal network $\mathbb G$ consisting of nine nodes and four snapshots. The current position of the random walker is indicated in red at node $v_3$ in snapshot $G_2$. \textbf{b)} Each snapshot is coupled to its immediate neighbors in time, and the walker first chooses a target snapshot with equal probability. In this example, the walker selects snapshot $G_3$. \textbf{c)} Coupling network $G_H$ and coupling matrix $H$. \textbf{d)} Having selected $G_3$ as a target snapshot, the walker jumps to a node in $G_3$ according to the transition probabilities given by the matrix representation of the Koopman operator $K_{23}$. If the chosen snapshot were $G_1$ (i.e., the walker performs a jump backward in time), the transition probabilities would instead be given by the matrix representation of the reweighted Perron--Frobenius operator $T_{12}$.}
    \label{fig:spatio_temporal_rw}
\end{figure}

\begin{lemma}
\label{le:spatio_temporal_time_reversibility}
There exists a stationary distribution $\bnu \in \mathbb{R}^{MN}$ of the random-walk process on $\mathscr G_H$ induced by $\bC_H$. Moreover, the process is time-reversible with respect to $\bnu$, that is, the detailed balance condition
\[
D_{\nu}\bC_H = \bC_H^{\top} D_{\nu}
\]
holds, where $D_{\nu}=\diag(\bnu)$. Consequently, $\bC_H$ is self-adjoint with respect to the $\bnu$-weighted inner product. All eigenvalues of $\bC_H$ are real and lie in the interval $[-1,1]$.
\end{lemma}

\begin{example} \label{ex:guiding_example}
Let us consider a temporal network $\mathbb{G} = (G_1, \dots, G_{20})$ consisting of 60 nodes observed over 20 snapshots, see figure~\ref{fig:guiding_example_part_1}. Each snapshot is generated independently from a stochastic block model (SBM), where the within-community and between-community connection probabilities are $p_{\mathrm{in}} = 0.7$ and $ p_{\mathrm{out}} = 0.05$, respectively. In snapshots $G_1, \dots, G_{10}$, the network exhibits two communities of equal size (30 nodes each). At snapshot $G_{11}$, one of these communities splits into two, and this three-community structure persists for the remaining time. The adjacency matrices of the snapshots are shown in figure~\ref{fig:guiding_example_part_1}\textbf{a}. In this example, the coupling network $G_H$ (figure~\ref{fig:guiding_example_part_1}\textbf{c}) is chosen such that the edge weights between the snapshots $G_t$ and $G_s$ are given by
\begin{equation}
\label{eq:decay_coupling}
w_{ts} = e^{-\alpha (s-t)^2},
\end{equation}
where $\alpha > 0$ is a decay parameter that controls how fast the strength of the connection between snapshots decreases with their temporal distance. We set $\alpha = 0.03$. The corresponding coupling matrix $H$ is defined as the transition matrix of a random walk on $G_H$ (figure~\ref{fig:guiding_example_part_1}\textbf{b}). We then construct the associated space--time network $\mathscr G_H$ (figure~\ref{fig:guiding_example_part_1}\textbf{d}). A random walker on this network moves between nodes in different snapshots according to the spatio-temporal transition matrix $\bC_H$ (figure~\ref{fig:guiding_example_part_1}\textbf{e}).
\end{example}

\begin{figure}[!ht]
    \centering
    \includegraphics[width=1\linewidth]{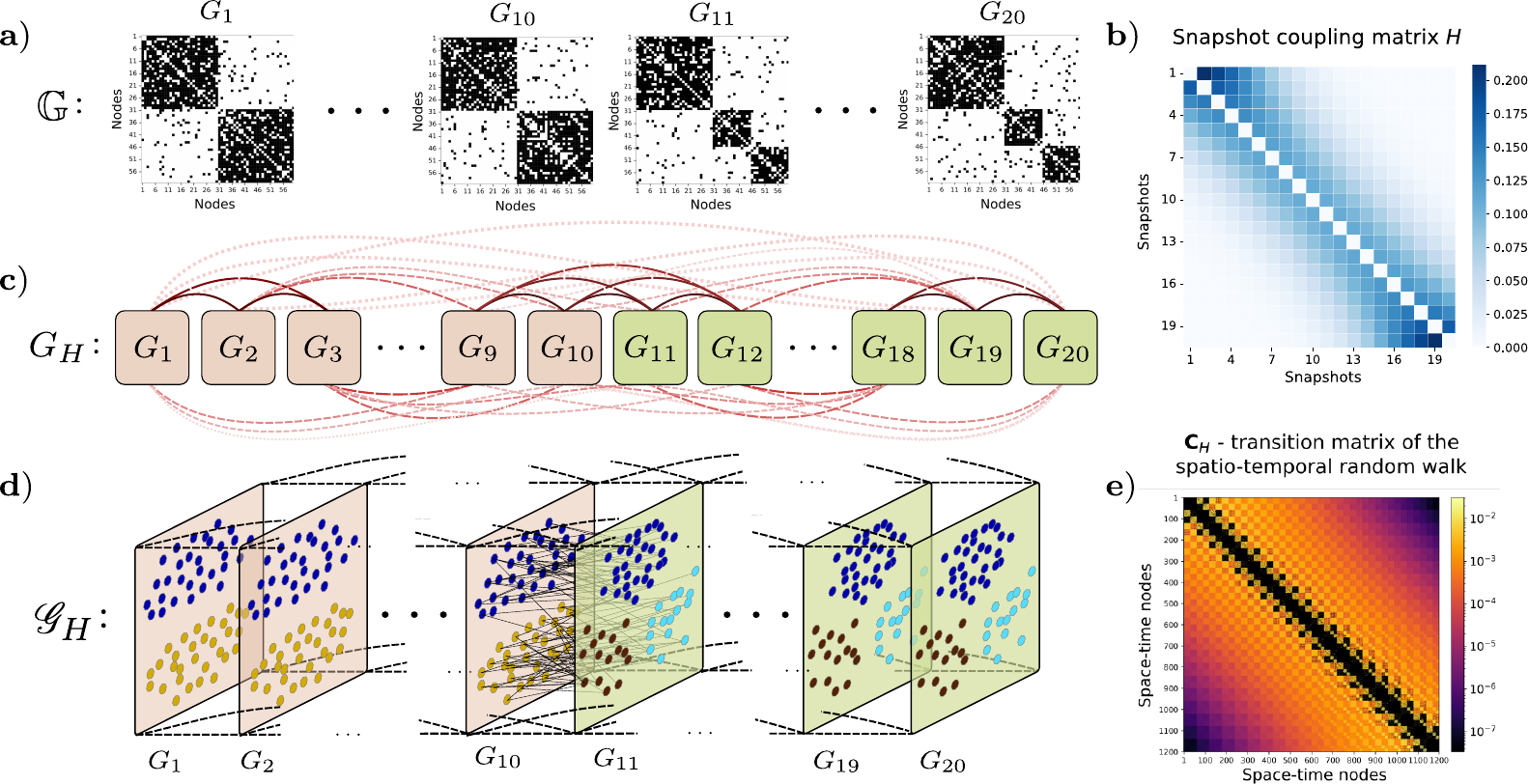}
    \caption{
    Illustration of the construction of a space--time network from a temporal network. 
    \textbf{a)} Adjacency matrices of selected snapshots of the temporal network $\mathbb{G}$. In the first 10 snapshots, the network consists of two equally sized communities (30 nodes each). At $t = 11$, one community splits into two, and this three-community structure persists until the end of the evolution. 
    \textbf{b)} Coupling matrix $H$, defined as the transition matrix of a random walk on the
    \textbf{c)} coupling network $G_H$. Periods during which $\mathbb G$ exhibits a stable community structure are highlighted in orange (snapshots 1–10) and green (snapshots 11–20).
    \textbf{d)} Space--time network $\mathscr G_H$ constructed from $\mathbb{G}$. For clarity of visualization, inter-snapshot edges are shown only between the 10th and 11th snapshots. Communities are color-coded: the dark blue community persists throughout the entire evolution, while the yellow community splits into a brown and a light blue community at $t = 11$.
    \textbf{e)} Transition matrix $\bC_H$ of the spatio-temporal random walk on $\mathscr G_H$.
    }
    \label{fig:guiding_example_part_1}
\end{figure}

\subsection{Spectral clustering of temporal networks}
\label{sec:spectral_clustering}

Motivated by well-established spectral clustering techniques for community detection in static networks \cite{von_Luxburg_2007, Klus_2024_1}, we analyze the spectral properties of the spatio-temporal transition operator $\bC_H$ to identify coherent sets of nodes that persist as densely intraconnected communities over extended time intervals.

In contrast to the purely static setting, eigenvectors of $\bC_H$ may encode two qualitatively different effects: those induced by the temporal coupling between snapshots, and those arising from the internal structure of the individual network snapshots. In \cite{Froyland_2024, Trower_2025}, these effects are analyzed by visually inspecting the leading eigenvectors of the random walk on the space–time network. Our aim is to provide a more principled theoretical understanding and to disentangle the two contributions. To this end, analogously to definition~\ref{def:observables}, we define the space--time observable space by
\[
\mathbb U^M=\{f: V \times \{1,\dots,M\} \to \mathbb{R}\}\mid f(\cdot,t)=f_t\in\mathbb U \text{ for all }1\le t\le M\}.
\]
The vector representation of $f\in\mathbb U^M$ is given by $\bbf = [\bbf_1^\top, \bbf_2^\top, \dots, \bbf_M^\top]^\top\in\mathbb R^{MN}$ and we refer to each observable $\bbf_t\in\mathbb R^N$ as the $t$th snapshot segment of $\bbf$. We now introduce a snapshot segment averaging operator $\bPi_{\mathrm{avg}} \colon \mathbb{U^M} \to \mathbb{U^M}$, which collapses node-level information within snapshot $G_t$ to its $\bmu_t$-weighted mean: 
\begin{equation}
\label{eq:Pi_avg}
\begin{split}
\bPi_{\mathrm{avg}} : (\bbf_1^\top,\dots,\bbf_M^\top)
\;\mapsto\;
\big(
\langle \bbf_1,\bone\rangle_{\mu_1}\bone^\top,\dots,
\langle \bbf_M,\bone\rangle_{\mu_M}\bone^\top
\big)
=
\big(
\mathbb{E}_{\mu_1}[\bbf_1]\bone^\top,\dots,
\mathbb{E}_{\mu_M}[\bbf_M]\bone^\top
\big).
\end{split}
\end{equation}
The matrix representation of $\bPi_{\mathrm{avg}}$ is given by 
\begin{equation}
\label{eq:Pi_avg_matrix}
\bPi_{\mathrm{avg}}=\diag(\bone\bmu_1^{\top},\dots,\bone\bmu_M^{\top}).
\end{equation}
The operator $\bPi_{\mathrm{avg}}$ averages snapshot segments of space--time observables with respect to $\bmu_t$, discarding all intra-snapshot variation.

\begin{lemma}
\label{le:Pi_avg_projection}
    The operator $\bPi_{\mathrm{avg}}$ is the $\bnu$-orthogonal projection onto the subspace of observables that are constant within each snapshot, that is
    \begin{equation*}
    \label{eq:Pi_avg_image}
    \operatorname{Im}(\bPi_{\mathrm{avg}})
    =
    \operatorname{span}\{e_t\otimes\bone \mid 1\le t\le M\},
    \end{equation*}
    where $e_t$ denotes the $t$th standard basis vector of $\mathbb R^M$ and $\otimes$ denotes the Kronecker product. 
\end{lemma}

As we will show below, this operator reveals a fundamental structural separation of the spectrum of $\bC_H$, allowing its eigenvectors to be naturally classified into two distinct types:
\begin{itemize}
    \item \emph{temporal eigenvectors}, which are constant within each snapshot, lie in the range of the operator $\bPi_{\mathrm{avg}}$, and reflect solely the coupling scheme between snapshots;
    \item \emph{spatial eigenvectors}, which lie in the kernel of $\bPi_{\mathrm{avg}}$ so that their snapshot segments $\bbf_t$ have zero $\bmu_t$-mean and capture the internal organization of the individual snapshots.
\end{itemize}

\begin{theorem}
\label{thm:operator_decomposition}
Let $\bC_H$ be a $\bnu$-self-adjoint spatio-temporal transition operator defined as in \eqref{eq:k_pf_notation}, associated with the random walk process on the space--time network $\mathscr G_H$. Let $\bPi_{\mathrm{avg}}$ denote the snapshot segment-averaging operator defined in \eqref{eq:Pi_avg}. Then the operator $\bC_H$ admits the decomposition 
\begin{equation}
\label{eq:theorem_decomposition}
\bC_H
=
\bC_H^{\mathrm{temp}}
+
\bC_H^{\mathrm{spat}},
\end{equation}
where
\[
\bC_H^{\mathrm{temp}}
=
\bPi_{\mathrm{avg}} \bC_H \bPi_{\mathrm{avg}},
\qquad
\bC_H^{\mathrm{spat}}
=
\bPi_{\mathrm{avg}}^{\perp} \bC_H \bPi_{\mathrm{avg}}^{\perp}.
\]
The eigenpairs of $\bC_H$ coincide with the eigenpairs of $\bC_H^{\mathrm{temp}}$ and $\bC_H^{\mathrm{spat}}$. This partitions the eigenpairs of $\bC_H$ into two classes: 
\begin{itemize}
    \item[i)] The $M$ eigenpairs of $\bC_H$ that are obtained as the eigenpairs of $\bC_H^{\mathrm{temp}}$. Their eigenvalues coincide with the eigenvalues of the coupling matrix $H$ and are independent of the internal structure of the network snapshots. The corresponding eigenvectors are completely determined by the eigenvectors of $H$ and are constant within each snapshot.
    \item[ii)] The remaining $MN-M$ eigenpairs of $\bC_H$ that are obtained as the eigenpairs of $\bC_H^{\mathrm{spat}}$. The snapshot segments $\bbf_t$ of the corresponding eigenvectors have zero $\bmu_t$-mean and capture the community structure of the network snapshots. 
\end{itemize}
\end{theorem}

\begin{proof}
    Since $H$ is $\bpi$-self-adjoint (see lemma~\ref{le:H_random_walk}), it admits a basis of $\bpi$-orthonormal eigenvectors. Denote these by $\bu^k=[u_1^k,\dots,u_M^k]^\top\in\mathbb R^M$, $1\le k\le M$, and let $\eta_k$ be the corresponding eigenvalues. We claim that the vectors $\bu^k\otimes\bone\in\mathbb R^{MN}$ are $\bnu$-orthonormal eigenvectors of $\bC_H$. Let us define 
    \begin{equation*}
    B_{ts}=
        \begin{cases}
           K_{ts}, & t<s,\\
           0, &t=s,\\
           T_{st}, &t>s.
        \end{cases}
    \end{equation*}
    Using the block structure of $\bC_H$ and the identity $B_{ts}\bone=\bone$ for all $1\le t,s\le M$ (see the proof of lemma~\ref{thm:row_stochasticity} in appendix~\ref{sec:appendix_1}), we obtain for every $1\le t\le M$
    \[
    (\bC_H(\bu^k\otimes\bone))_t=\sum_{s=1}^Mh_{ts}u^k_sB_{ts}\bone=\sum_{s=1}^Mh_{ts}u^k_s\bone=(H\bu^k)_t\bone=\eta_ku^k_t\bone.
    \]        
    Hence,
    $
    \bC_H(\bu^k\otimes\bone)=\eta_k(\bu^k\otimes\bone)
    $ and
    \begin{equation*}
    \la(\bu^k\otimes\bone),(\bu^l\otimes\bone)\ra_{\nu}=\sum_{t=1}^M\pi_t\la u^k_t\bone,u^l_t\bone\ra_{\mu_t}=\sum_{t=1}^M\pi_tu^k_tu^l_t=\la\bu^k,\bu^l\ra_{\pi}=
    \begin{cases}
        1, & k=l, \\
        0, & k\neq l.
    \end{cases}
    \end{equation*}
    Since $\bC_H$ is $\bnu$-self-adjoint, the spectral decomposition yields
    \[
    \bC_H=\sum_{k=1}^M\eta_k(\bu^k\otimes\bone)(\bu^k\otimes\bone)^{\top}+\sum_{k=1}^{MN-M}\theta_k\bbf^k(\bbf^k)^{\top}.
    \]
    Using the identity
    \[
    (a\otimes b)(c\otimes d)^\top = (ac^\top)\otimes(bd^\top)
    \]
    together with the spectral decomposition of $H$, we obtain
    \[
    \bC_H=H\otimes(\bone\bone^{\top})+\sum_{k=1}^{MN-M}\theta_k\bbf^k(\bbf^k)^{\top}.
    \]
    Let $\bPi_{\mathrm{avg}}^{\perp}=I-\bPi_{\mathrm{avg}}$ denote the $\bnu$-orthogonal projection onto the $\bnu$-orthogonal complement of $\operatorname{Im}(\bPi_{\mathrm{avg}})$, then
    \begin{equation}
    \label{eq:C_H_decomposition}
        \begin{split}
            \bC_H&=(\bPi_{\mathrm{avg}}+\bPi_{\mathrm{avg}}^{\perp})\bC_H(\bPi_{\mathrm{avg}}+\bPi_{\mathrm{avg}}^{\perp})\\&=\bPi_{\mathrm{avg}}\bC_H\bPi_{\mathrm{avg}}+\bPi_{\mathrm{avg}}^{\perp}\bC_H\bPi_{\mathrm{avg}}+\bPi_{\mathrm{avg}}\bC_H\bPi_{\mathrm{avg}}^{\perp}+\bPi_{\mathrm{avg}}^{\perp}\bC_H\bPi_{\mathrm{avg}}^{\perp}.
        \end{split}
    \end{equation}
    Since $\bC_H(\operatorname{Im(\bPi_{\mathrm{avg}})})\subset\operatorname{Im}(\bPi_{\mathrm{avg}})$, we have $\bPi_{\mathrm{avg}}^{\perp}\bC_H\bPi_{\mathrm{avg}}=0$. Then, since $\bPi_{\mathrm{avg}}$ and $\bC_H$ are $\bnu$-self-adjoint, $\bPi_{\mathrm{avg}}^{\perp}\bC_H\bPi_{\mathrm{avg}}$ and $\bPi_{\mathrm{avg}}\bC_H\bPi_{\mathrm{avg}}^{\perp}$ are $\bnu$-adjoint to one another. Hence, $\bPi_{\mathrm{avg}}\bC_H\bPi_{\mathrm{avg}}^{\perp}=0$ as well. Therefore, the mixed terms in \eqref{eq:C_H_decomposition} vanish and we get
    \begin{equation}
    \label{eq:C_H_decomposition_short}
    \bC_H=\bPi_{\mathrm{avg}}\bC_H\bPi_{\mathrm{avg}}+\bPi_{\mathrm{avg}}^{\perp}\bC_H\bPi_{\mathrm{avg}}^{\perp}.
    \end{equation}
    Substituting \eqref{eq:k_pf_notation} and \eqref{eq:Pi_avg_matrix}, we compute
    \[
    \bPi_{\mathrm{avg}}\bC_H\bPi_{\mathrm{avg}}=H\otimes(\bone\bone^{\top}).
    \]
    Combining this with \eqref{eq:C_H_decomposition_short} yields
    \[
    \bPi_{\mathrm{avg}}^{\perp}\bC_H\bPi_{\mathrm{avg}}^{\perp} = \sum_{k=1}^{MN-M}\theta_k\bbf^k(\bbf^k)^{\top}.
    \]
    We define
    \[
    \bC_H^{\mathrm{temp}}=\bPi_{\mathrm{avg}}\bC_H\bPi_{\mathrm{avg}}
    \quad
    \text{and}
    \quad
    \bC_H^{\mathrm{spat}}=\bPi_{\mathrm{avg}}^{\perp}\bC_H\bPi_{\mathrm{avg}}^{\perp}.
    \]
    This yields the decomposition \eqref{eq:theorem_decomposition} of the space--time operator $\bC_H$.  Furthermore, since $\bPi_{\mathrm{avg}}\bbf^k=0$ for all $1\le k\le MN-M$, each snapshot segment $\bbf_t^k$, $1\le t\le M$, of a spatial eigenvector $\bbf^k$ satisfies
    \[
    \e_{\mu_t}[\bbf_t^k]=0.
    \]
    
\end{proof}

Since $\bC_H^{\mathrm{temp}}$ depends only on the coupling matrix $H$ and not on the structure of the snapshots, we refer to it as the \emph{temporal component} and to the corresponding eigenpairs $(\eta_k,\bu^k\otimes\bone)$ as \emph{temporal eigenpairs}. Similarly, we refer to $\bC_H^{\mathrm{spat}}$ as the \emph{spatial component} and to the corresponding eigenpairs $(\theta_k,\bbf^k)$ as \emph{spatial eigenpairs}. Consequently, all eigenpairs of $\bC_H$ can be partitioned into temporal and spatial eigenpairs, obtained as eigenpairs of the temporal and spatial components, respectively. In the following, for consistency of notation, we denote the temporal eigenpairs by $(\lambda_k^{\mathrm{temp}},\bbf^{\mathrm{temp},k})$, for $1\le k\le M$, and the spatial eigenpairs by $(\lambda_k^{\mathrm{spat}},\bbf^{\mathrm{spat},k})$ for $1\le k\le MN-M$. Throughout, superscript indices are used to distinguish different space--time observables, while subscripts denote their snapshot segments.

In particular, by restricting our attention to $\bC_H^{\mathrm{spat}}$, the temporal eigenvectors are filtered out, allowing the leading eigenvectors $\bbf^{\mathrm{spat},k}$ of the spatial component $\bC_H^{\mathrm{spat}}$ to isolate the spatial structure of the temporal network. As a result, they provide a natural framework for analyzing node coherence and community structure. 

\begin{remark}
\label{rem:how_to_select_spatial_evec}
We emphasize that the spatial eigenvectors are not entirely independent of the snapshot coupling scheme. In particular, slowly decaying eigenvectors of the coupling matrix $H$ may significantly interfere with the spatial eigenvectors of $\bC_H$. This effect is not caused by the underlying network structure, but rather by the choice of snapshot coupling, and manifests itself as a modulation of the snapshot segments $\bbf_t^{\mathrm{spat},k}$ by the leading eigenvectors of $H$. Consequently, a careful selection and preprocessing of spatial eigenvectors is required to identify \emph{canonical} representatives that are informative for clustering. This phenomenon was recognized in related approaches~\cite{Froyland_2024, Trower_2025}, but the selection of informative spatial eigenvectors remained an open question. In this paper, we illustrate these effects in the guiding example~\ref{ex:guiding_example} and in the numerical examples in section~\ref{sec:numerical_examples}. Furthermore, we include an additional discussion of the heuristics used to select spatial eigenvectors that serve as good feature coordinates for space--time node clustering in appendix~\ref{sec:appendix_4}.
\end{remark}

In algorithm~\ref{alg:algorithm}, we provide a complete overview of how our approach is applied in practice for community detection in a given temporal network. In figure~\ref{fig:diagram} we provide a schematic illustration of our algorithm.

\begin{algorithm}
\caption{Spectral clustering of temporal networks using spatio-temporal random walks}
\label{alg:algorithm}
	\begin{algorithmic}[1]
        \State \textbf{Input:}\\ $\texttt{adj}$ -- (weighted) adjacency matrices of the $M$ snapshots;\\ 
        $\bmu_1$ -- initial probability distribution (here we choose the uniform distribution); \\$w_{ts}$ 
        -- snapshot coupling weights.
            
        \For {$t=1,2,\dots,M$}
            \State Compute snapshot transition matrices $S_t$ according to \eqref{eq:transition_matrix}. 
            \If {$t>1$}
                \State Compute $\bmu_t^{\top}=\bmu_{t-1}^{\top}S_{t-1}$
            \EndIf
        \EndFor

        \State Compute $H$ as defined in \eqref{eq:compute_H}.

        \State Compute $\bC_H$ using \eqref{eq:koopman_mat_rep}, \eqref{eq:reweighted_perron_frobenius_mat_rep}, and \eqref{eq:k_pf_notation}.

        \State Compute $\bPi_{\mathrm{avg}}$ as given in \eqref{eq:Pi_avg_matrix} and set $\bC^{\mathrm{spat}}_H \gets (I-\bPi_{\mathrm{avg}})\bC_H(I-\bPi_{\mathrm{avg}})$.

        \State \textbf{Apply spectral clustering on} $\bC_H^{\mathrm{spat}}$:
        \State \hspace{2em} Compute leading eigenvalues $\lambda_k^{\mathrm{spat}}$ and eigenvectors $\bbf^{\mathrm{spat},k}$ of $\bC_H^{\mathrm{spat}}$.
        \State \hspace{2em} Among the leading spatial eigenvectors, select the informative ones $\bbf^{\mathrm{spat},i_1},\dots\bbf^{\mathrm{spat},i_l}$ (e.g., using the heuristics given in appendix~\ref{sec:appendix_4}). 
        \State \hspace{2em} Define $\mathbf F=[\bbf^{\mathrm{spat},i_1},\cdots,\bbf^{\mathrm{spat},i_l}]\in\mathbb R^{MN\times l}$ and let $\bbf_{j,:}\in\mathbb R^l, 1\le j\le MN$ denote the $j$th row of $\mathbf F$.
        \State \hspace{2em} Run a clustering algorithm (e.g., $k$-means, DBSCAN) on feature vectors $\{\mathbf F_{j,:}\}_{j=1}^{MN}$.
	\end{algorithmic}
\label{alg:lne}
\end{algorithm}

\begin{figure}
    \centering
    \includegraphics[width=1\linewidth]{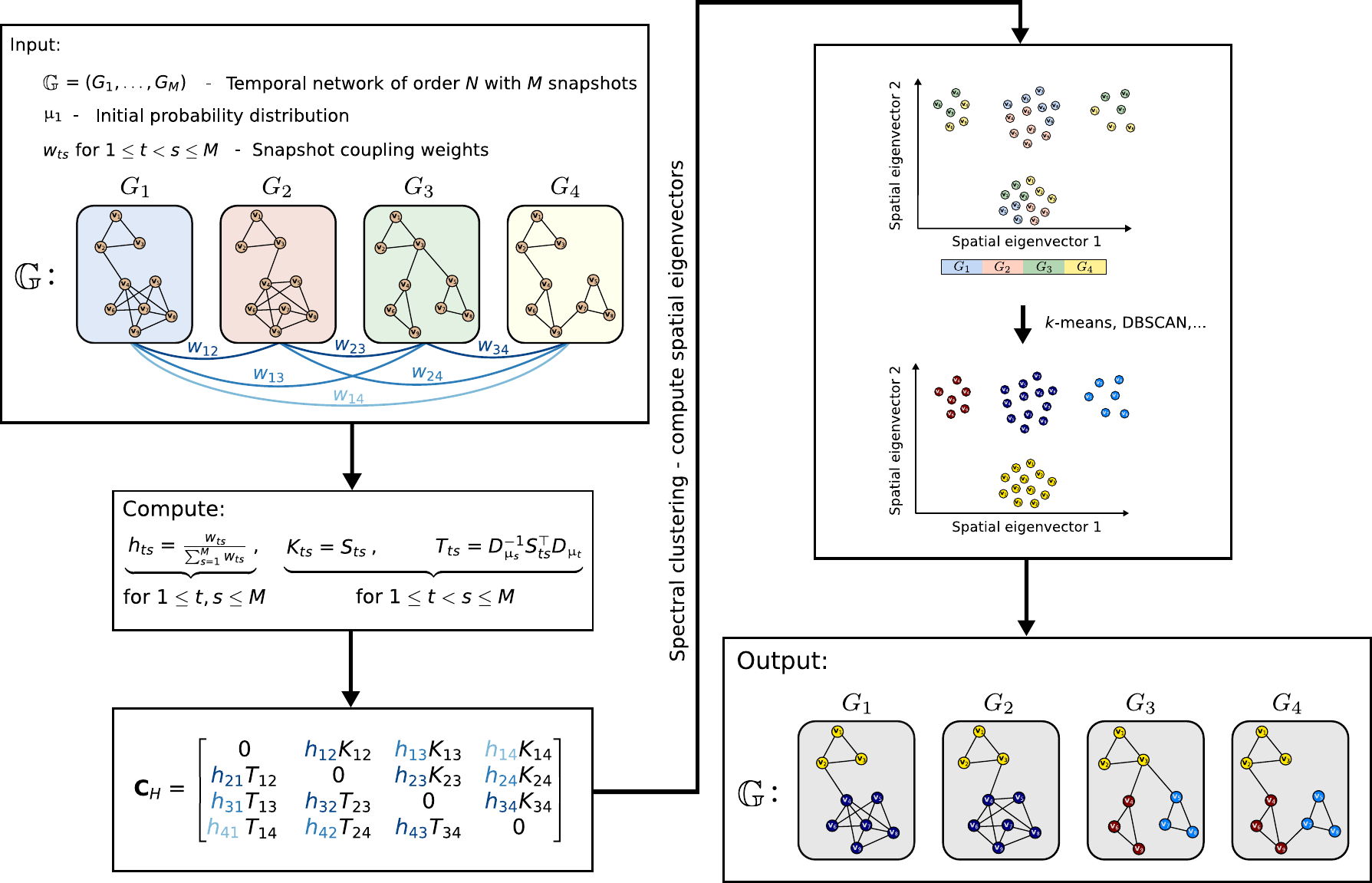}
    \caption{A schematic illustration of algorithm~\ref{alg:algorithm}.}
    \label{fig:diagram}
\end{figure}

\paragraph{Continuation of example~\ref{ex:guiding_example}.} 
\begin{itshape}
In figure~\ref{fig:guiding_example_part_2} we illustrate the spectral clustering of the space--time network from example~\ref{ex:guiding_example}. The leading eigenvalues of $\bC_H$ are shown in figure~\ref{fig:guiding_example_part_2}\textbf{b}, with the spatial eigenvalues highlighted. The first four spatial eigenvectors of $\bC_H$ are displayed in figure~\ref{fig:guiding_example_part_2}\textbf{c}. For visualization purposes, each eigenvector $\bbf^{\mathrm{spat},k}$ is folded into its $M$ snapshot segments, which are displayed using a color gradient ranging from dark blue ($\bbf_1^{\mathrm{spat},k}$) to dark red ($\bbf_M^{\mathrm{spat},k}$). The sign structure of the first spatial eigenvector clearly separates nodes $1$--$30$ from nodes $31$--$60$, which shows that they are members of different communities throughout the entire network evolution. Additionally, the third spatial eigenvector distinguishes between nodes $31$--$45$ and $46$--$60$ during the second half of the evolution, indicating a split of the community around $t=10$. Applying $k$-means to these two feature vectors we successfully recover all communities in the temporal network $\mathbb G$ (figure~\ref{fig:guiding_example_part_2}\textbf{a}). This example also illustrates how slowly decaying modes of the coupling matrix $H$ may interfere with the spatial eigenvectors (see remark~\ref{rem:how_to_select_spatial_evec}). In particular, spatial eigenvectors 1, 2, and 4 in figure~\ref{fig:guiding_example_part_2}\textbf{c} can all be interpreted approximately as modulated versions of an observable separating the first 30 and last 30 nodes, where the modulation is induced by the first (blue), second (orange), and third (green) eigenvectors of $H$, respectively (figure~\ref{fig:guiding_example_part_2}\textbf{e}), which correspond to three largest eigenvalues of $H$ (figure~\ref{fig:guiding_example_part_2}\textbf{d}). This phenomenon complicates the direct application of standard spectral clustering methods, in which $k$-means is typically applied to the spatial eigenvectors corresponding to the dominant eigenvalues of $\bC_H$ preceding the spectral gap.
\end{itshape}

\begin{remark}
    Our method differs from approaches that cluster each snapshot independently and subsequently match the resulting communities across time. In such two-step approaches, the temporal ordering is not considered during the clustering itself, but only afterward through a separate matching procedure. Since community labels have no intrinsic meaning, the user must specify which snapshots are compared and how corresponding communities are identified, for example by choosing suitable matching rules \cite{Aynaud_2010, Cazabet_2019}. Moreover, snapshot-wise clustering can be unstable, as many clustering algorithms are stochastic or sensitive to small changes in the network and may therefore produce substantially different partitions in consecutive snapshots. In contrast, our method incorporates both the spatial structure and the temporal ordering of the snapshots into a single spectral problem, thereby detecting communities across all snapshots simultaneously.
\end{remark}

\begin{figure}
    \centering
    \includegraphics[width=0.95\linewidth]{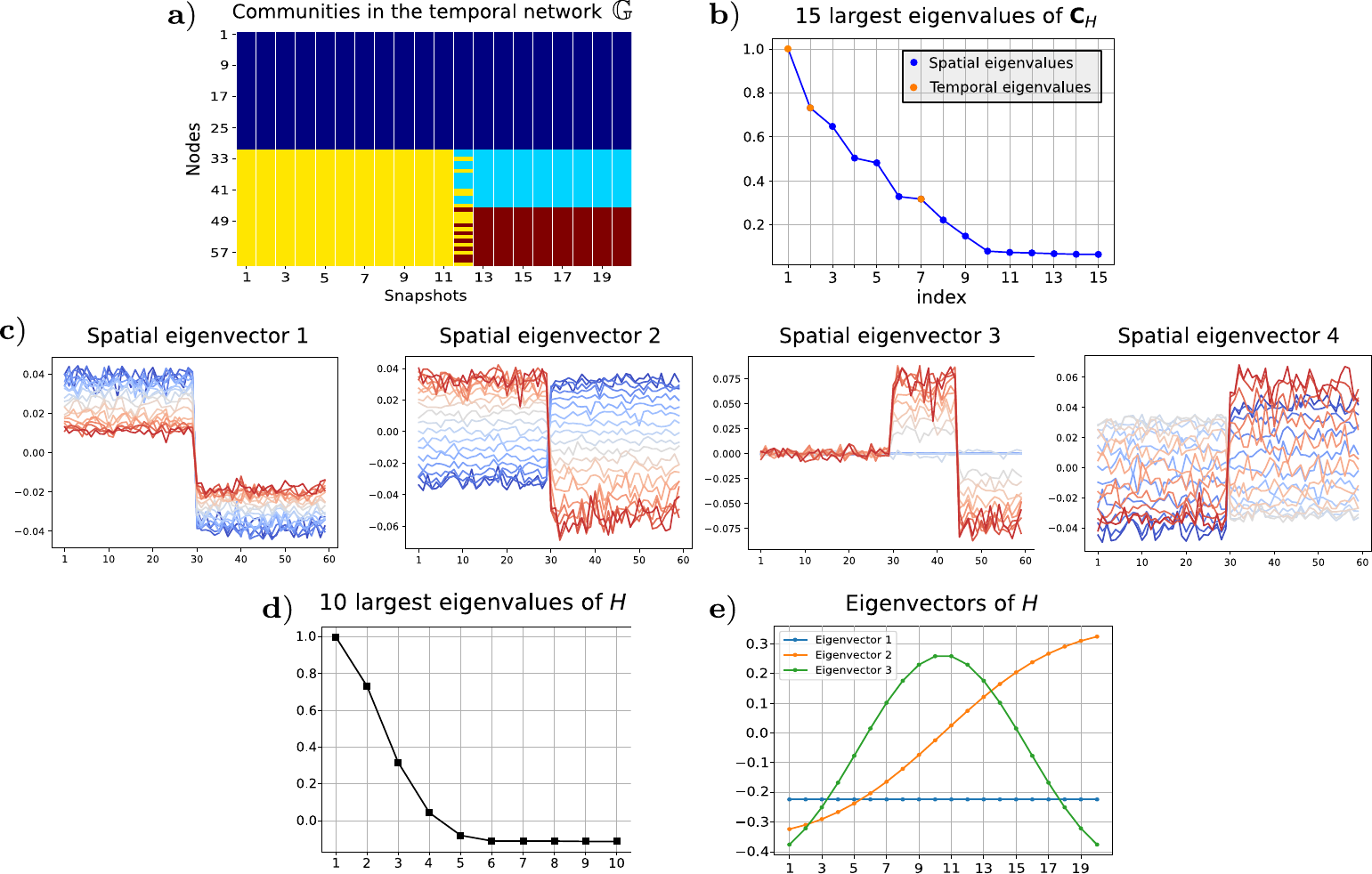}
    \caption{
    Illustration of the community detection via spectral clustering. 
    \textbf{a)} $k$-means clustering of space--time nodes based on the first and third spatial eigenvectors of $\bC_H$. 
    \textbf{b)} Leading eigenvalues of $\bC_H$.
    \textbf{c)} First four spatial eigenvectors of $\bC_H$ (i.e., the four leading eigenvectors of $\bC^{\mathrm{spat}}_H$). For visualization, each eigenvector $\bbf^{\mathrm{spat},k}$ is folded into its $M$ snapshot segments, which are displayed using a color gradient ranging from dark blue ($\bbf_1^{\mathrm{spat},k}$) to dark red ($\bbf_M^{\mathrm{spat},k}$). Each segment $\bbf_t$, $1\le t\le M$, encodes information about the community structure of snapshot $G_t$.
    \textbf{d)} 10 largest eigenvalues of $H$.
    \textbf{e)} Three dominant eigenvectors of $H$.}
    \label{fig:guiding_example_part_2}
\end{figure}

\subsection{Computational complexity}

The construction introduced above leads to a spatio-temporal transition matrix
$\bC_H \in \mathbb R^{MN \times MN}$,
where $N$ denotes the number of nodes and $M$ the number of snapshots. Consequently, the dimension of the associated eigenvalue problem scales with the product $MN$ and can become prohibitively large for long temporal sequences or large networks.

The computational cost depends on the chosen snapshot coupling scheme and on the sparsity of the underlying network snapshots. In the most general setting, the matrix $\bC_H$ consists of $M^2$ blocks of size $N\times N$, yielding a storage complexity of order $O(M^2N^2)$. Although practical implementations can exploit sparsity and the block structure of $\bC_H$, the computation of dominant eigenpairs remains the main computational bottleneck. 

These observations motivate the development of a reduced-order representation of the spatio-temporal random walk. In the next section, we derive a projection-based reduction that preserves the essential spectral properties of $\bC_H$ while significantly reducing the computational effort required for community detection. More details on the computational complexity of our algorithm and its reduced version are given in appendix~\ref{sec:appendix_3}.

\section{Model reduction}
\label{sec:model_reduction} 
In this section, we develop a reduced model by restricting the search space for observables associated with each snapshot to suitably chosen low-dimensional subspaces. Since the snapshot observables that solve our objective~\eqref{eq:weighted_correlation} are expected to be approximately constant on communities, this reduction substantially lowers the memory requirements of the full model while preserving the essential community-level information (see remark~\ref{rem:computational_complexity}).

\subsection{Galerkin projection of covariance operators}
\label{section:covariance_operators}

Let $\cW_t = \mathrm{span}\{ w_1^t,\dots,w_{d_t}^t \} \subset \mathbb{U}$ be a $d_t$-dimensional subspace of observables associated with snapshot $G_t$, equipped with the $\bmu_t$-weighted inner product. Let $W_t = [\bw_1^t,\dots,\bw_{d_t}^t] \in \mathbb{R}^{N\times d_t}$ denote the matrix whose columns are the vector representations of the basis functions spanning $\cW_t$. We aim to solve~\eqref{eq:weighted_correlation} under the constraint that each observable $f_t$ lies in $\cW_t$. Let $\ba_t \in \mathbb{R}^{d_t}$ denote the coordinate vector of $\bbf_t$ with respect to the basis given by the columns of $W_t$, so that $\bbf_t = W_t \ba_t$. The covariance of $f_t$ can then be computed as
\[
\var(f_t)
= \langle f_t,\, \mathcal{C}_{tt} f_t \rangle
= \langle W_t \ba_t,\, C_{tt}(W_t \ba_t) \rangle
= \langle \ba_t,\, W_t^\top C_{tt} W_t \,\ba_t \rangle.
\]
Similarly, we obtain the cross-covariance of observables $f_t\in\cW_t$ and $f_{s}\in\cW_{s}$ for $t<s$ as   
\[
\cov(f_t,f_s)=\la f_t,\mathcal{C}_{ts}f_{s}\ra=\la W_t\ba_t,C_{ts}(W_{s}\ba_{s})\ra=\la \ba_t,W_t^\top C_{ts}W_{s}\ba_{s}\ra.
\] 
We define
\[
\widehat{C}_{tt}=W_t^\top C_{tt} W_t=W_t^{\top}D_{\mu_t}W_t
\]
and 
\[
\widehat{C}_{ts}=W_t^\top\, C_{ts}\, W_{s}=W_t^{\top}D_{\mu_t}S_{ts}W_s.
\]
Now, analogously as in section~\ref{sec:covariance_operators_and_weighted_correlation_optimization} we derive the generalized eigenvalue problem that solves \eqref{eq:maximization_problem} restricted to subspaces of observables $\cW_t$, i.e.,
\begin{equation}
\label{eq:generalized_eval_problem_projected}
    \widehat{\bA}_H\ba=\widehat{\lambda}\widehat{\bB}\ba,
\end{equation}
where
{
\small
\begin{equation*}
    \widehat{\textbf{A}}_H=
    \begin{bmatrix}
    0 & h_{12}\widehat{C}_{12} & \cdots & h_{1M}\widehat{C}_{1M} \\
    h_{21} \widehat{C}_{12}^{\top} & 0 & \ddots & \vdots \\
    \vdots & \ddots & \ddots & h_{(M-1)M}\widehat{C}_{(M-1)M} \\
    h_{M1}\widehat{C}_{1M}^{\top} & \cdots &  h_{M(M-1)}\widehat{C}_{(M-1)M}^{\top} & 0
    \end{bmatrix},
    \hspace{3mm}
    \widehat{\bB}=
    \begin{bmatrix}
    \widehat{C}_{11} & 0 & \cdots & 0 \\
    0 & \widehat{C}_{22} & \ddots & \vdots \\
    \vdots & \ddots & \ddots & 0 \\
    0 & \cdots & 0 & \widehat{C}_{MM}
    \end{bmatrix},
\end{equation*}
}
$\ba=[\ba_1^{\top},\ba_2^{\top},...,\ba_M^{\top}]^{\top}$ and $\ba_t$ is a coefficient vector containing coordinates of $f_t\in\mathcal W_t$ with respect to the basis $\{w_i^t\}_{i=1}^{d_t}$. Since all matrices $W_t$ have full column rank, the matrices $\widehat{C}_{tt}$ are invertible, and therefore so is $\widehat{\bB}$. Defining 
\begin{equation}
\label{eq:C_H_projected}
    \widehat{\bC}_H=\widehat{\bB}^{-1}\widehat{\bA}_H,
\end{equation}
we can equivalently rewrite \eqref{eq:generalized_eval_problem_projected} as a standard eigenvalue problem
\begin{equation}
\label{eq:eigenvalue_problem_reduced}
    \widehat{\bC}_H\ba=\widehat{\lambda}\ba.
\end{equation}

Solutions of \eqref{eq:eigenvalue_problem_reduced} therefore give, for each snapshot $t$, coefficient vectors $\ba_t\in\mathbb{R}^{d_t}$. To interpret them as snapshot observables, we lift each $\ba_t$ back to the full observable space $\mathcal W_t$ by reconstructing
\[
\bbf_t = W_t \ba_t, \qquad t=1,\ldots,M.
\]
The procedure is schematically illustrated in figure~\ref{fig:lifting_illustration}.

\begin{figure}
    \centering
    \includegraphics[width=1\linewidth]{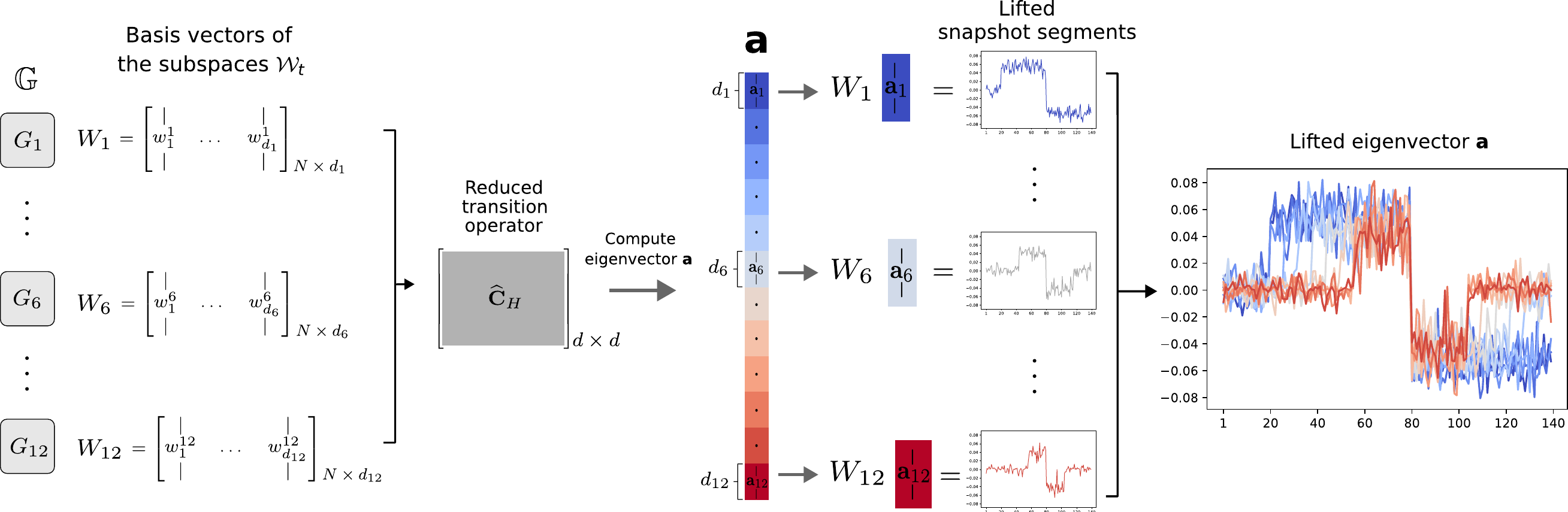}
    \caption{Illustration of the computation of network observables $\bbf$ that maximize the objective function such that their snapshot segments $\bbf_t$ belong to the subspaces $\cW_t$, $t=1,\ldots,M$.
    }
    \label{fig:lifting_illustration}
\end{figure}

If the subspaces $\cW_t$ are chosen so that they contain observables whose level sets reflect the community structure in each snapshot, then the resulting space--time observables $\bbf = [\bbf_1^{\top},\dots,\bbf_M^{\top}]^{\top}$ admit a clear interpretation in terms of how these communities evolve over time, as in the full (unprojected) model. A natural choice is, for example, to take $\mathcal W_t$ as the subspace spanned by the dominant eigenvectors of the transition matrix of a random walk on the static network $G_t$ (see example~\ref{sec:example_2}). In this case, eigenvectors $\ba$ of the reduced operator $\widehat{\bC}_H$ define, through lifting, space--time modes $\bbf$ with $\bbf_t\in \mathcal W_t$ that remain informative of the network community structure and can be directly compared to those of the full operator $\bC_H$, while being significantly cheaper to compute. 

For $1\le t<s\le M$, the building blocks of $\widehat{\bC}_H$ are given by
\begin{equation*}
    [\widehat{\bC}_H]_{ts}=h_{ts}\widehat{C}_{tt}^{-1}\widehat{C}_{ts},
    \quad\hspace{2mm}
    [\widehat{\bC}_H]_{tt}=0,
    \quad\hspace{2mm}
    [\widehat{\bC}_H]_{st}=h_{st}\widehat{C}_{ss}^{-1}\widehat{C}_{ts}^{\top}
\end{equation*}
where
\begin{equation}
\label{eq:Kts_projection}
\widehat{C}_{tt}^{-1}\widehat{C}_{ts}=(W_t^{\top}D_{\mu_t}W_t)^{-1}W_t^{\top}D_{\mu_t}K_{ts}W_{s}=\widehat{K}_{ts},
\end{equation}
and 
\begin{equation}
\label{eq:Tts_projection}
\widehat{C}_{ss}^{-1}\widehat{C}_{ts}^{\top}=(W_{s}^{\top}D_{\mu_{s}}W_{s})^{-1}W_{s}^{\top}D_{\mu_{s}}T_{ts}W_t=\widehat{T}_{ts}.
\end{equation}
Let
\[
\cW = \operatorname{colspan}(\bW)
\]
be the space spanned by the columns of the block diagonal matrix $\bW=\diag(W_1,\dots W_M)\in\mathbb R^{MN\times d}$ for $d=\sum_{t=1}^Md_t$.
More precisely, for each space--time obesrvable $\bbf$ where $\bbf_t=W_t\ba_t$ we can write 
\[
\bbf=\bW\ba.
\]
Then, we obtain
\begin{equation}
\label{eq:C_H_Galerkin}
    \widehat{\bC}_H=(\bW^{\top}D_{\nu}\bW)^{-1}\bW^{\top} D_{\nu}\bC_H\bW.
\end{equation}

The building blocks $\widehat{K}_{ts}\in\mathbb R^{d_t\times d_s}$ and $\widehat{T}_{ts}\in\mathbb R^{d_s\times d_t}$ of matrix $\widehat{\bC}_H$ given in \eqref{eq:Kts_projection} and \eqref{eq:Tts_projection} are exactly Petrov--Galerkin projections of the transfer operators $K_{ts}$ and $T_{ts}$, where trial and test spaces may differ across snapshots. The reduced operator $\widehat{\bC}_H\in\mathbb R^{d\times d}$ \eqref{eq:C_H_Galerkin} can be understood as a Galerkin projection of the full spatio-temporal operator $\bC_H$ onto $\cW$ with respect to the $\bnu$-weighted inner product. 

\begin{remark}
\label{rem:galerkin}
    Since $\bC_H$ acts as a transfer operator on the space--time network, this result is closely related to generalizations of Markov state models, which are obtained as Galerkin projections of transfer operators associated with dynamical processes onto suitably chosen low-dimensional subspaces~\cite{Djurdjevac_2010}. This shows that our reduced model can be interpreted as a coarse-grained version of the spatio--temporal random walk that preserves the dominant spectral properties and long-time dynamics of the original process.
\end{remark}

The operator $\widehat{\bC}_H$ inherits self-adjointness of $\bC_H$ as shown in the following lemma. Proofs of all lemmas stated in this section are provided in appendix~\ref{sec:appendix_1}. 

\begin{lemma}
\label{le:projected_self_adjointness}
Let $\bG_{\cW} = \bW^{\top} D_{\nu}\bW$ denote the Gram matrix induced by the
$\bnu$-weighted inner product on the subspace $\cW$.
Then the projected operator $\widehat{\bC}_H$ is self-adjoint with respect to the $\bG_{\cW}$-weighted inner product, i.e.,
\[
\langle \widehat{\bC}_H\ba,\bb\rangle_{\bG_{\cW}}
=
\langle \ba,\widehat{\bC}_H\bb\rangle_{\bG_{\cW}}
\]
for all $ \ba,\bb\in\mathbb{R}^{d} $, where $\la\ba,\bb\ra_{\bG_{\cW}}=\ba^{\top}\bG_{\cW}\bb$.
\end{lemma}

\begin{remark}
    Defining $\{w_1^t,\ldots,w_{d_t}^t\}$ to be the standard basis $\{e_1,\ldots,e_N\}$ of $\mathbb{R}^N$ for all $1\leq t\leq M$ recovers the full formulation of our approach, since each observable space $\cW_t$ coincides with the full space $\mathbb{U}$.
\end{remark}

Spectral features of $\bC_H$ associated with slowly decaying space--time observables are well approximated by the spectrum of $\widehat{\bC}_H$, provided that the subspaces $\cW_t$ are appropriately chosen. We will later show that, under mild assumptions on the subspaces $\cW_t$, this connection can be made more precise. We begin with a general result about the spectrum of the projected operator $\widehat{\bC}_H$.

\begin{lemma}
\label{lem:coarse_grained_spectral_radius_thm}
    Reduced operator $\widehat{\bC}_H$ defined in \eqref{eq:C_H_Galerkin} has a real spectrum. For arbitrary subspaces $\cW_t$, its spectral radius satisfies
    \[
    \rho(\widehat{\bC}_H) \leq 1.
    \]
\end{lemma}

Since $\bW$ is a block-diagonal matrix, the $\bnu$-orthogonal projector $\bPi_{\cW}$ is given by 
\begin{equation}
\label{eq:Pi_W_diag}
\bPi_{\cW}=\bW(\bW^\top D_{\nu}\bW)^{-1}\bW^\top D_{\nu}=\diag(\bPi_{\cW_1}^{\mu_1},\dots,\bPi_{\cW_M}^{\mu_M}).
\end{equation}
where $\bPi_{\cW_t}^{\mu_t}$ denotes the $\bmu_t$-orthogonal projector onto $\cW_t$. We note that the operators $\widehat{\bC}_H$ and $\bPi_{\cW}\bC_H|_{\cW}$ are equivalent in the sense that they represent the same transformation expressed in different bases: the basis given by the columns of $\bW$ and the standard basis $\{e_n\}_{n=1}^N$, respectively. 

Analogously to theorem~\ref{thm:operator_decomposition}, we use the projected averaging operator $\widehat{\bPi}_{\mathrm{avg}}$ to split the reduced operator $\widehat{\bC}_H$ into temporal and spatial components. Let us define 
\begin{equation}
\label{eq:Pavg_galerkin}
\widehat{\bPi}_{\mathrm{avg}}
=
(\bW^\top D_{\nu}\bW)^{-1}\bW^\top D_{\nu}\,\bPi_{\mathrm{avg}}\,\bW.
\end{equation}
Using \eqref{eq:Pi_avg_matrix} we compute $\widehat{\bPi}_{\mathrm{avg}}$ as block diagonal-matrix with digaonal blocks given as 
\begin{equation*}
    (\widehat{\bPi}_{\mathrm{avg}})_{tt}=(W_t^{\top}D_{\mu_t}W_t)^{-1}W_t^{\top}\mu_t(W_t^{\top}\mu_t)^{\top} \quad \text{ for } 1\le t\le M.
\end{equation*}

In the remainder of the paper, we will additionally impose the mild assumption that $\bone \in \cW_t$ for all $1 \le t \le M$, unless stated otherwise. In the following theorem we show that $\widehat{\bC}_H$ retains the fundamental spectral properties of $\bC_H$ and, analogously to the full operator, admits a decomposition into temporal and spatial components. Moreover, in theorem~\ref{thm:error_bound} we derive an error bound for the leading eigenvalues. This establishes the projected formulation as a computationally efficient but nevertheless accurate surrogate of the full operator. 

\begin{theorem}
\label{thm:C_H_decomposition_projected}
    Let $\widehat{\bC}_H$ be a Galerkin projection of the spatio-temporal transition operator defined in \eqref{eq:C_H_projected} 
    onto a subspace spanned by the columns of $\bW=\diag(W_1,\dots W_M)$. We assume that $\bone\in\operatorname{colspan}(W_t)$ for every $1\le t\le M$. Then the operator $\widehat{\bC}_H$ admits the decomposition
    \[
    \widehat{\bC}_H=\widehat{\bC}_H^{\mathrm{temp}}+\widehat{\bC}_H^{\mathrm{spat}},
    \]
    where $\widehat{\bC}_H^{\mathrm{temp}}$ and $\widehat{\bC}_H^{\mathrm{spat}}$ denote the compressions of $\bC_H^{\mathrm{temp}}$ and $\bC_H^{\mathrm{spat}}$ to the subspace $\cW$, that is, their projections restricted to $\cW$ and represented in the basis given by the columns of $\bW$. Moreover, it holds
    \[
    \widehat{\bC}_H^{\mathrm{temp}}=\widehat{\bPi}_{\mathrm{avg}}\,\widehat{\bC}_H\,\widehat{\bPi}_{\mathrm{avg}},
    \quad
    \widehat{\bC}_H^{\mathrm{spat}}=\widehat{\bPi}_{\mathrm{avg}}^{\perp}\,\widehat{\bC}_H\,\widehat{\bPi}_{\mathrm{avg}}^{\perp}.
    \]
    Furthermore, the spectral radius of the projected transition operator is $\rho(\widehat{\bC}_H)=1$. The eigenpairs of $\widehat{\bC}_H$ coincide with the eigenpairs of $\widehat{\bC}_H^{\mathrm{temp}}$ and $\widehat{\bC}_H^{\mathrm{spat}}$ thereby partitioning the spectrum of $\widehat{\bC}_H$ into two classes:
\begin{itemize}
    \item[i)] The $M$ eigenpairs of $\widehat{\bC}_H$ are obtained from the eigenpairs of $\widehat{\bC}_H^{\mathrm{temp}}$. The corresponding eigenvalues coincide with the temporal eigenvalues of $\bC_H$ (and therefore with the eigenvalues of the coupling matrix $H$) and are independent of the internal structure of the network snapshots. The corresponding eigenvectors are coefficient vectors in the column space basis $\bW$ of the temporal eigenvectors of $\bC_H$.
    
    \item[ii)] The remaining $d - M$ eigenpairs of $\widehat{\bC}_H$ are obtained from the eigenpairs of $\widehat{\bC}_H^{\mathrm{spat}}$. The eigenpairs of $\widehat{\bC}_H$ from this class provide approximations of the leading spatial eigenpairs of the original operator $\bC_H$.
\end{itemize}
\end{theorem}

\begin{proof}
First, we recall that the averaging operator $\bPi_{\mathrm{avg}}$ is the orthogonal projection onto the subspace
\[
\operatorname{Im}(\bPi_{\mathrm{avg}})=\spanning\{\,e_t\otimes\bone \mid 1\le t\le M\,\}.
\]
Since $\bone\in\mathcal W_t$ for every snapshot, we have $\operatorname{Im}(\bPi_{\mathrm{avg}}) \subseteq\mathcal W$, and therefore
\begin{equation}
\label{eq:PwPavg}
\bPi_{\mathcal W}\bPi_{\mathrm{avg}}\bPi_{\mathcal W}=
\bPi_{\mathrm{avg}}\bPi_{\mathcal W}
=
\bPi_{\mathcal W}\bPi_{\mathrm{avg}}
=
\bPi_{\mathrm{avg}}.
\end{equation}
Then, using the decomposition from theorem~\ref{thm:operator_decomposition}, we get
\[
\bPi_{\cW}\bC_H|_{\cW}=\bPi_{\cW}(\bC_H^{\mathrm{temp}} + \bC_H^{\mathrm{spat}})|_{\cW}=\bPi_{\cW}\bC_H^{\mathrm{temp}}|_{\cW} + \bPi_{\cW}\bC_H^{\mathrm{spat}}|_{\cW},
\]
that is,
\[
\widehat{\bC}_H=\widehat{\bC}_H^{\mathrm{temp}}+\widehat{\bC}_H^{\mathrm{spat}}.
\]
Substituting \eqref{eq:C_H_Galerkin} and \eqref{eq:Pavg_galerkin} into $\widehat{\bPi}_{\mathrm{avg}}\,\widehat{\bC}_H\,\widehat{\bPi}_{\mathrm{avg}}$ and $\widehat{\bPi}_{\mathrm{avg}}^{\perp}\,\widehat{\bC}_H\,\widehat{\bPi}_{\mathrm{avg}}^{\perp}$ and using \eqref{eq:PwPavg}, we compute
\[
\widehat{\bPi}_{\mathrm{avg}}\,\widehat{\bC}_H\,\widehat{\bPi}_{\mathrm{avg}}=\widehat{\bC}_H^{\mathrm{temp}}
\quad \text{and} \quad
\widehat{\bPi}_{\mathrm{avg}}^{\perp}\,\widehat{\bC}_H\,\widehat{\bPi}_{\mathrm{avg}}^{\perp}=\widehat{\bC}_H^{\mathrm{spat}},
\]
which proves the claim.

Furthermore, we have 
\[
\bPi_{\cW}\bC_H^{\mathrm{temp}}|_{\cW}= \bPi_{\cW}\bPi_{\mathrm{avg}}\bC_H\bPi_{\mathrm{avg}}|_{\cW}=\bPi_{\mathrm{avg}}\bC_H\bPi_{\mathrm{avg}}|_{\cW}=\bC_H^{\mathrm{temp}}|_{\cW}
\]
so the nonzero eigenvalues of $\bC_H^{\mathrm{temp}}$ and $\widehat{\bC}_H^{\mathrm{temp}}$ coincide, and therefore $1\in\sigma(\widehat{\bC}_H^{\mathrm{temp}})\subset\sigma(\widehat{\bC}_H)$. Using lemma~\ref{lem:coarse_grained_spectral_radius_thm}, we  conclude that $\rho(\widehat{\bC}_H)=1$.
\end{proof}

\begin{remark}[Computational complexity]
\label{rem:computational_complexity}
Assuming that the temporal network is sparse and that snapshots separated by at most $r$ time steps are coupled, the construction of the full operator $\bC_H$ has computational complexity and memory requirements of order $\mathcal O(MrN^2)$. In contrast, the reduced operator $\widehat{\bC}_H$ can be constructed in $\mathcal O(MrNd_{\max}^2)$ time and requires $\mathcal O(MrNd_{\max})$ memory, where
\[
d_{\max}=\max\{d_t\mid1\le t\le M\}.
\]
Typically, $d_{\max}\ll N$, and the reduced formulation scales linearly with the number of nodes $N$ so it provides substantial computational savings. Details of these estimates are given in appendix~\ref{sec:appendix_3}.
\end{remark}

\subsection{Error bound}

Using the results from \cite{Knyazev_2010, Djurdjevac_2012}, we can give error bounds on how accurately the projected operator $\widehat{\bC}_H$ approximates the spatial eigenvalues of $\bC_H$.

\begin{theorem}[Error bound via projection onto reduced subspaces]
\label{thm:error_bound}
Let $1 \le m \le d$, and let
\[
\lambda_1^{\mathrm{spat}} \ge \cdots \ge \lambda_m^{\mathrm{spat}}
\]
be the leading spatial eigenvalues of the spatio-temporal transition matrix $\bC_H$, with corresponding $\bnu$-orthonormal eigenvectors
\[
\bbf^{\mathrm{spat},1}, \dots, \bbf^{\mathrm{spat},m}.
\]
Let $\cW_t$ be projection subspaces such that $\bone\in\cW_t$ for all $1\le t\le M$. Let $\widehat{\lambda}_m^{\mathrm{spat}}$ be $m$th spatial eigenvalue of the reduced operator $\widehat{\bC}_H$ defined in \eqref{eq:C_H_Galerkin}. The eigenvalue error
\[
E_m = \bigl|\lambda_m^{\mathrm{spat}} - \widehat{\lambda}_m^{\mathrm{spat}}\bigr|
\]
satisfies
\[
E_m \;\le\; (\lambda_1^{\mathrm{spat}} + 1)
\sum_{i=1}^m \| \bPi_{\mathcal W}^\perp \bbf^{\mathrm{spat},i} \|_\nu^2,
\]
and equivalently,
\begin{equation}
\label{eq:extended_error}
E_m \;\le\; (\lambda_1^{\mathrm{spat}} + 1)
\sum_{i=1}^m \sum_{t=1}^M \pi_t\,
\| (\bPi_{\mathcal W_t}^{\mu_t})^\perp \bbf_t^{\mathrm{spat},i} \|_{\mu_t}^2.
\end{equation}
\end{theorem}

\begin{proof}
Since $\bPi_{\mathrm{avg}}^{\perp}$ is a projection, the operator
\[
\bC_H^{\mathrm{spat}} = \bPi_{\mathrm{avg}}^{\perp}\bC_H\bPi_{\mathrm{avg}}^{\perp}
\]
is self-adjoint with respect to the $\bnu$-weighted inner product. Let $\mathcal F = \mathrm{span}(\bbf^{\mathrm{spat},1},\dots,\bbf^{\mathrm{spat},m})$. 
As $\mathcal F$ is invariant under $\bC_H^{\mathrm{spat}}$, we may apply theorem 2.4 from~\cite{Knyazev_2010}. 
Let $\theta_1 \ge \dots \ge \theta_m$ be the principal angles between $\mathcal F$ and $\mathcal W$. Let $\lambda_{\min(\mathcal F+\mathcal W)}$ denote the smallest eigenvalue of the projection of $\bC_H^{\mathrm{spat}}$ onto the space $\mathcal F+\mathcal W$. Then 
\begin{equation}
\label{eq:knyazev_clean}
\begin{aligned}
|\lambda_m^{\mathrm{spat}} - \widehat{\lambda}_m^{\mathrm{spat}}|
&\le \sum_{i=1}^m |\lambda_i^{\mathrm{spat}} - \widehat{\lambda}_i^{\mathrm{spat}}| \\
&\le \sum_{i=1}^m (\lambda_i^{\mathrm{spat}} - \lambda_{\min(\mathcal F+\mathcal W)}) \sin^2 \theta_i \\
&\le (\lambda_1^{\mathrm{spat}} - \lambda_{\min(\mathcal F+\mathcal W)}) \sum_{i=1}^m \sin^2 \theta_i.
\end{aligned}
\end{equation}
It remains to express $\sum_{i=1}^m \sin^2 \theta_i$ in a convenient form. Let $\bPi = F F^\top D_\nu$ denote the $\bnu$-orthogonal projection onto $\mathcal F$, where 
$F = [\bbf^{\mathrm{spat},1},\dots,\bbf^{\mathrm{spat},m}]\in\mathbb R^{MN\times m}$. 
Let $\sigma_i(A)$, $\Lambda_i(A)$, and $A^*$ denote the $i$th singular value, the $i$th eigenvalue, and the Hermitian transpose of an operator $A$, respectively. Then the principal angles between the subspaces $\mathcal F$ and $\mathcal W$ are determined by the $m$ largest singular values of $\bPi \bPi_{\mathcal W}$
\[
\cos^2(\theta_{m+1-i}) = \sigma_i^2(\bPi \bPi_{\mathcal W})
= \Lambda_i(\bPi \bPi_{\mathcal W}(\bPi \bPi_{\mathcal W})^*)
= \Lambda_i(\bPi \bPi_{\mathcal W} \bPi)
\quad
\text{ for }1\le i\le m.
\]
Hence,
\[
\sin^2(\theta_{m+1-i}) = 1-\Lambda_i(\bPi \bPi_{\mathcal W} \bPi)= \Lambda_i(\bPi - \bPi \bPi_{\mathcal W} \bPi)
= \Lambda_i(\bPi \bPi_{\mathcal W}^{\perp} \bPi)
\quad
\text{ for }1\le i\le m.
\]
Since $\bPi \bPi_{\mathcal W}^{\perp} \bPi$ has at most rank $m$, we obtain
\begin{equation*}
\sum_{i=1}^m \sin^2(\theta_i)
= \mathrm{tr}(\bPi \bPi_{\mathcal W}^{\perp} \bPi).
\end{equation*}
Substituting $\bPi = FF^\top D_\nu$ and using the identity
$\operatorname{tr}(AB)=\operatorname{tr}(BA)$ together with
$F^\top D_\nu F = I$, we obtain
\[
\operatorname{tr}(\bPi \bPi_{\mathcal W}^{\perp} \bPi)
=
\operatorname{tr}(F^\top D_\nu \bPi_{\mathcal W}^{\perp} F).
\]

Using that $\bPi_{\mathcal W}^{\perp}$ is $\bnu$-self-adjoint, we obtain
\[
F^\top D_\nu \bPi_{\mathcal W}^{\perp} F
= (\bPi_{\mathcal W}^{\perp} F)^\top D_\nu (\bPi_{\mathcal W}^{\perp} F),
\]
and therefore
\begin{equation}
\label{eq:energy_clean}
\sum_{i=1}^m \sin^2(\theta_i)
= \sum_{i=1}^m \|\bPi_{\mathcal W}^{\perp} \bbf^{\mathrm{spat},i}\|_\nu^2.
\end{equation}
Since all eigenvalues of $\bC_H^{\mathrm{spat}}$ are contained in the interval $[-1,1]$ (see lemma~\ref{le:spatio_temporal_time_reversibility}), we have  $\lambda_{\min(\mathcal F+\mathcal W)} \ge -1$ as well. Combining this with \eqref{eq:knyazev_clean} and \eqref{eq:energy_clean} we conclude
\[
|\lambda_m^{\mathrm{spat}} - \widehat{\lambda}_m^{\mathrm{spat}}|
\le (\lambda_1^{\mathrm{spat}} + 1)
\sum_{i=1}^m \|\bPi_{\mathcal W}^{\perp} \bbf^{\mathrm{spat},i}\|_\nu^2.
\]

Finally, inserting the block structure of $\bnu$ and using \eqref{eq:Pi_W_diag} yields \eqref{eq:extended_error}.
\end{proof}

\begin{remark}
The estimate in theorem~\ref{thm:error_bound} is an apriori error bound since the projection errors \eqref{eq:extended_error} depend on the unknown eigenvectors of the full operator $\bC_H$ and are therefore not directly computable from the reduced model. A computable a posteriori indicator of the approximation quality can instead be obtained from the residual of a lifted reduced eigenpair,
\[
r=\bC_H\bW\ba-\widehat{\lambda}\bW\ba=\bPi_{\cW}^{\perp}\bC_H\bW\ba,
\]
where the second equality follows from the fact that the reduced model is a Galerkin projection of $\bC_H$ onto $\cW$ with respect to the $\nu$-weighted inner product. The relation between residual norms and the accuracy of approximate eigenvectors as well as the spectral separation conditions required for corresponding error bounds, are well studied in the literature, e.g., \cite{Knyazev_1997, Nakatsukasa_2019}.
\end{remark}

\section{Numerical examples}
\label{sec:numerical_examples}

In this section, we illustrate different aspects of our framework with the aid of three numerical examples. The first example demonstrates how the choice of snapshot coupling influences the detection of evolving and recurring communities. The second example compares the full and reduced formulations of the method and highlights the effectiveness of the proposed model reduction strategy. Finally, the third example applies our method to a temporal network generated from an opinion dynamics model. This example is based on an interacting-agent dynamical system designed to capture real-world collective behavior and demonstrates the ability of our approach to identify and track evolving communities in complex dynamical systems. In all examples, we use the heuristics described in appendix~\ref{sec:appendix_4} to select the spatial eigenvectors used as feature coordinates for clustering.

\subsection{Example 1 -- recurrent communities}
\label{sec:example_1}
In this example we consider a temporal network $\mathbb{G} = (G_1,\dots, G_{32})$ with 120 nodes consisting of 32 snapshots generated independently using a stochastic block model with intra-community probability $p_{\mathrm{in}} = 0.7$ and inter-community probability $p_{\mathrm{out}} = 0.05$. In snapshots $G_1$--$G_8$, the network exhibits two stable communities $C_1 = \{v_1,\dots,v_{60}\}$, $C_2 = \{v_{61},\dots, v_{120}\}$. At snapshot $G_9$, both communities $C_1$ and $C_2$ start gradually splitting into two smaller communities of equal size, $C_{11} = \{v_1,\dots, v_{30}\}$, $C_{12} = \{v_{31},\dots, v_{60}\}$, and $C_{21} = \{v_{61},\dots, v_{90}\}$, $C_{22} = \{v_{91},\dots, v_{120}\}$, respectively. This splitting process is completed by snapshot $G_{12}$ and the resulting four-community structure then persists over snapshots $G_{13}$--$G_{20}$. Lastly, communities $C_{11}$ and $C_{12}$ gradually merge again during the period $G_{21}$--$G_{24}$ and after snapshot $G_{25}$ the network has three stable communities $C_1$, $C_{21}$, and $C_{22}$, which persist until the end of the evolution. We refer to these three stable regimes of the network evolution as phases~1, 2, and~3. The adjacency matrices of this temporal network are shown in figure~\ref{fig:numerical_example_1}\textbf{a}.

In this example, we employ the snapshot coupling scheme defined in~\eqref{eq:decay_coupling}. Since the stable periods and structural changes occur on a longer time scale than in example~\ref{ex:guiding_example}, we choose a smaller decay parameter $\alpha = 0.01$ in order to strengthen the coupling between snapshots over longer temporal distances. We consider two coupling matrices $H_1$ and $H_2$: a cyclic one (figure~\ref{fig:numerical_example_1}\textbf{b}), where the last snapshots are coupled to the first ones (in this case we set $w_{ts}=\exp(-\alpha\min\{|t-s|,M-|t-s|\}^2)$ for $1\le t,s\le M$), and a non-cyclic one (figure~\ref{fig:numerical_example_1}\textbf{c}), where no such coupling is imposed.

By introducing cyclic coupling, our approach is able to correctly identify the reappearance of community $C_1$ from phase~1 in phase~3 (figure~\ref{fig:numerical_example_1}\textbf{g}). In particular, the first spatial eigenvector in the cyclic coupling case separates the space--time nodes across all three phases into two large groups (figure~\ref{fig:numerical_example_1}\textbf{e}), while spatial eigenvectors~3 and~4 provide additional information about the splitting of communities $C_1$ and $C_2$ and the merging of communities $C_{11}$ and $C_{12}$. This allows our approach to recover community $C_1$ in phase~3. In contrast, without cyclic coupling, this would not be possible (figures~\ref{fig:numerical_example_1}\textbf{h},\textbf{i}). The space--time observables describing the separation of community $C_1$ from the rest of the network in phases~1 and~3 appear as two different spatial eigenvectors (eigenvectors~1 and~2 in figure~\ref{fig:numerical_example_1}\textbf{f}), due to the weak coupling between snapshots from these phases, while eigenvectors 3 and 6 describe the rest of structural changes. Although nodes $v_1,\dots,v_{60}$ form a community in both phases, there is no information contained within leading eigenvectors of $\bC_H$ in this case indicating that these two communities are the same. Depending on whether we want to identify a densely connected group of nodes as the same community, even though it might have undergone structural changes between two periods of stability, we can, depending on the nature of the data, employ cyclic (or more generally long-range) couplings to detect such patterns.

\begin{figure}
    \centering
    \includegraphics[width=0.8\linewidth]{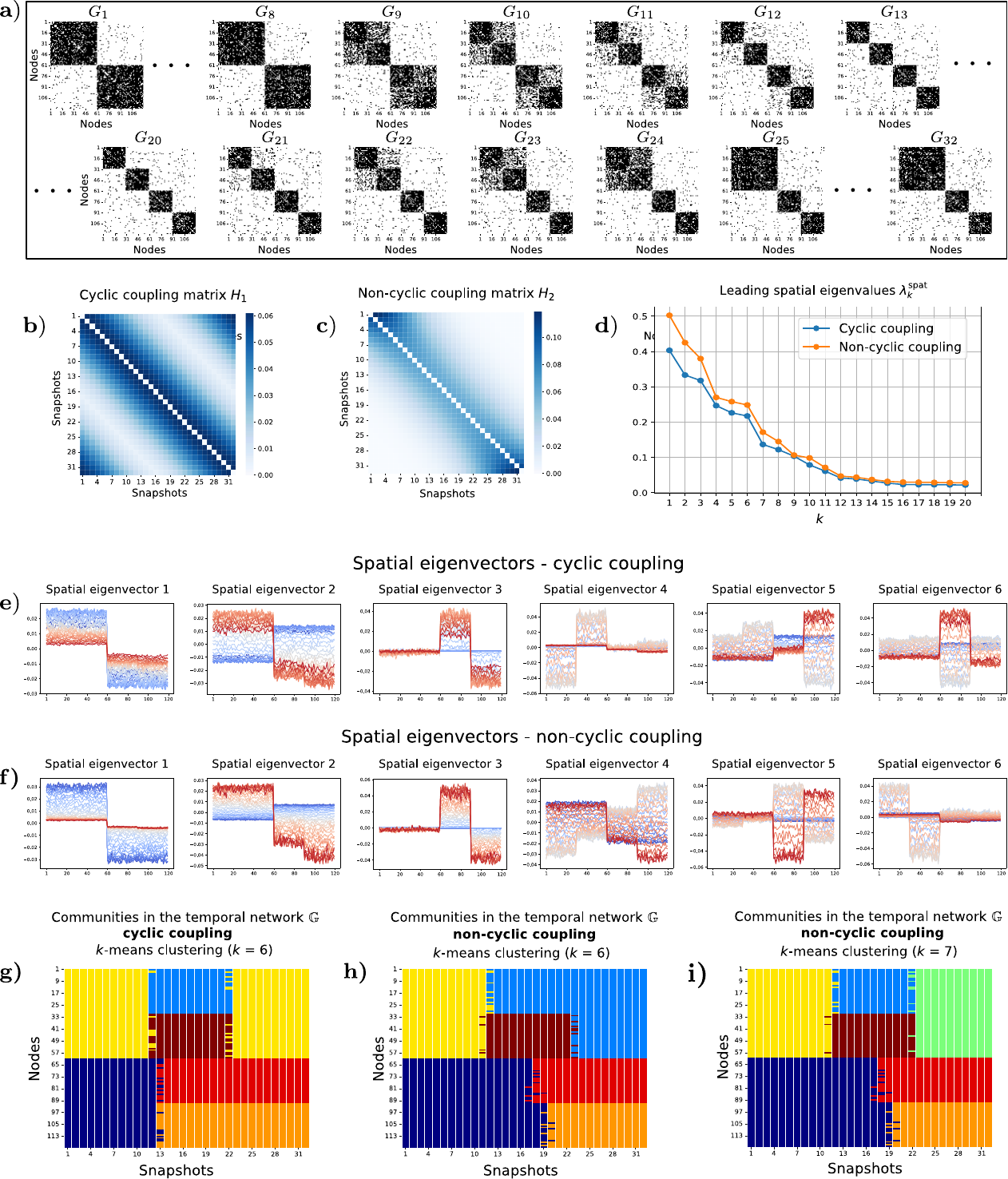}
    \caption{Illustration of our approach on the temporal network $\mathbb{G} = (G_1,\dots, G_{32})$. 
\textbf{a)} Adjacency matrices of the snapshots $G_1$--$G_{32}$ with 120 nodes. The network exhibits three stability phases: in the first phase ($G_1$--$G_8$), the network consists of two stable communities. Between $G_9$ and $G_{12}$, both communities gradually split, leading to a second stability phase with four communities ($G_{13}$--$G_{20}$). During $G_{21}$--$G_{24}$, two of these communities merge again, resulting in a final stable phase with three communities ($G_{25}$--$G_{32}$). 
\textbf{b)} Cyclic coupling matrix $H_1$. 
\textbf{c)} Non-cyclic coupling matrix $H_2$. 
\textbf{d)} Leading spatial eigenvalues of $\bC_{H_i}, i\in\{1,2\}$, for cyclic and non-cyclic couplings. 
\textbf{e)} Leading spatial eigenvectors of $\bC_{H_1}$ for cyclic coupling. 
\textbf{f)} Leading spatial eigenvectors of $\bC_{H_2}$ for non-cyclic coupling. 
\textbf{g)}~$k$-means ($k=6$) clustering of space--time nodes using spatial eigenvectors~1, 3, and~4 as feature coordinates in the cyclic case. 
\textbf{h)} $k$-means ($k=6$) clustering of space--time nodes using spatial eigenvectors~1, 2, 3, and~6 as feature coordinates in the non-cyclic case. 
\textbf{i)} $k$-means ($k=7$) clustering of space--time nodes using spatial eigenvectors~1, 2, 3, and~6 as feature coordinates in the non-cyclic case.}
    \label{fig:numerical_example_1}
\end{figure}

\subsection{Example 2 -- model reduction}
\label{sec:example_2}

In this example we perform a comparative analysis of the full and reduced approaches on a temporal network with 140 nodes. We choose again the coupling scheme defined in \eqref{eq:decay_coupling} with $\alpha=0.03$ (figure~\ref{fig:numerical_example_2}\textbf{b}). Since the community structure of each snapshot is encoded in the dominant eigenvectors of the transition matrix $S_t$ of a random walk defined on it, we use them to define the search spaces $\mathcal W_t$. To further reduce computational requirements, instead of computing the full eigendecomposition of $S_t$, we use an approximation of the dominant eigenvectors. In this example, we employ the Nyström method~\cite{Fowlkes_2004} to obtain good approximations by computing the eigendecomposition of a submatrix defined by a chosen node sample and extending the resulting eigenvectors to the full dimension. As shown in~\cite{Fowlkes_2004}, relatively small node samples are often sufficient to obtain good approximations. In this experiment, we use samples obtained by selecting half of the nodes of the temporal network uniformly at random.

We therefore define the spaces $\mathcal W_t$ to be spanned by the constant vector $\bone \in \mathbb{R}^N$ together with the approximated dominant eigenvectors of $S_t$. In figure~\ref{fig:numerical_example_2}\textbf{e}, we compare the spatial eigenvalues of $\bC_H$ in the full model (blue) with those obtained from projections onto the subspaces $\mathcal W_t$ (green). For comparison, we also include the spatial eigenvalues obtained when $\mathcal W_t$ are spanned by the exact dominant eigenvectors of $S_t$ (orange). In both cases, the spatial eigenvalues of the reduced models are close to those of the full model. In figure~\ref{fig:numerical_example_2}\textbf{c}, we show the leading spatial eigenvectors of $\bC_H$ and in figure~\ref{fig:numerical_example_2}\textbf{d} the network observables obtained by lifting eigenvectors of $\widehat{\bC}_H$. We observe that these network observables closely resemble the spatial eigenvectors of $\bC_H$, while being computed at a significantly lower computational cost.

Guided by the heuristics described in section~\ref{sec:example_1}, we select spatial eigenvectors~1 and~2 as feature vectors for the space--time node clustering. In~figures~\ref{fig:numerical_example_2}\textbf{f},\textbf{h}, we show the corresponding two-dimensional embeddings of space--time nodes, where each node is colored according to the snapshot it belongs to. Applying the $k$-means algorithm with $k=3$, we see in figures~\ref{fig:numerical_example_2}\textbf{g},\textbf{i} that the community structure of the temporal network is correctly identified in both cases, which demonstrates the effectiveness of the proposed reduced model.  

Lastly, to further demonstrate the robustness of our method, we compare our snapshot coupling scheme with the method proposed in~\cite{Trower_2025} which can be understood as a special case of our approach, where only the average correlation between successive snapshots is maximized (see the coupling matrix $H_s$ in figure~\ref{fig:numerical_example_2}\textbf{j}). Since this coupling strategy only incorporates very local temporal interactions, it fails to adequately capture long-range temporal persistence of communities across snapshots. As a result, important temporal information is lost in the eigenvectors of $\bC_{H_s}$. In figure~\ref{fig:numerical_example_2}\textbf{j}, we show that the $k$-means clustering of the space--time nodes with feature coordinates given by the leading spatial eigenvectors of $\bC_{H_s}$ fails to correctly identify the evolving community structure of the temporal network.

\begin{figure}
    \centering
    \includegraphics[width=0.78\linewidth]{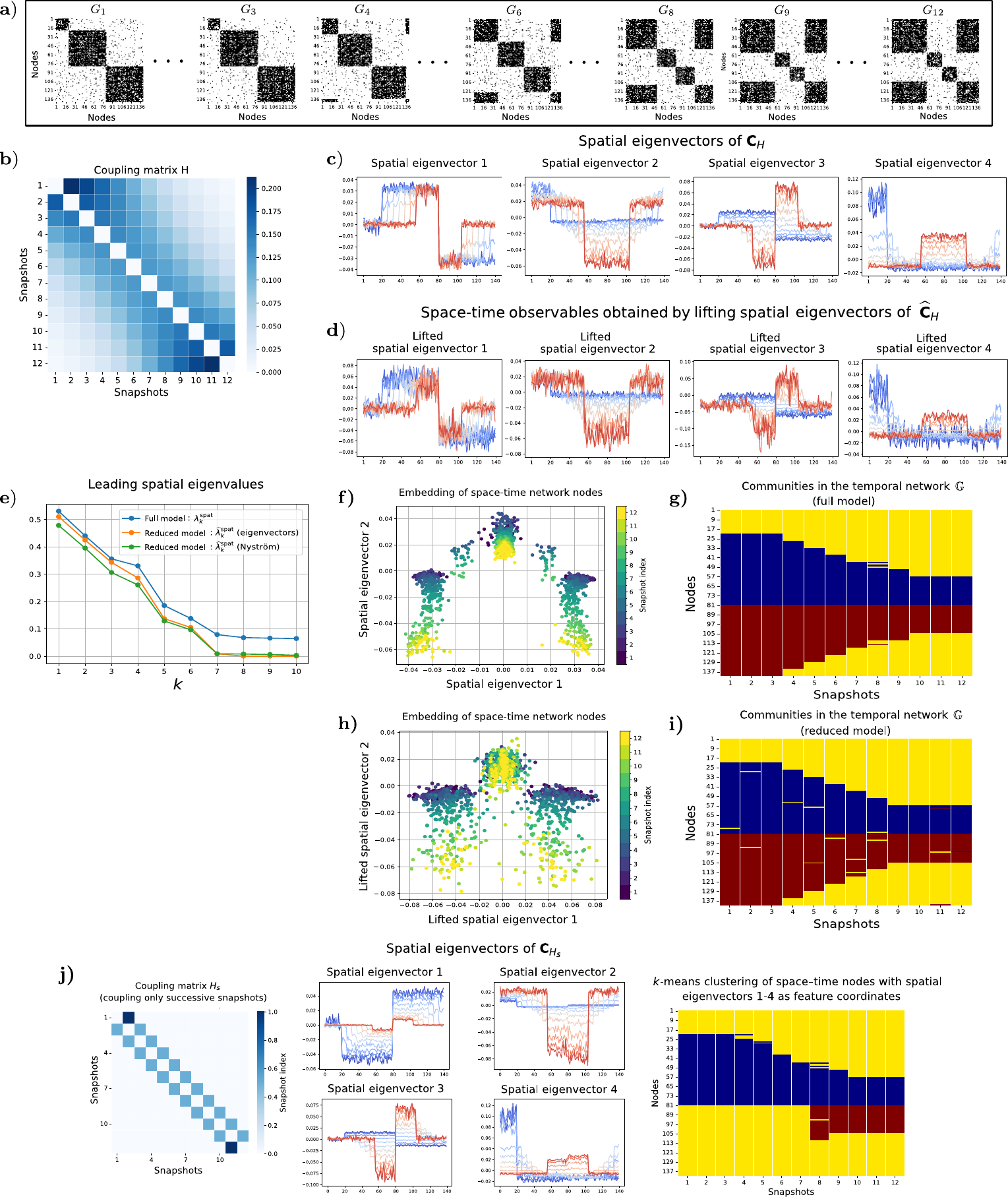}
    \caption{Illustration of our approach on the temporal network $\mathbb G=(G_1,\dots,G_{12})$ with 140 nodes. 
\textbf{a)} Adjacency matrices of snapshots $G_1$--$G_{12}$. In snapshots $G_1$--$G_3$, the network consists of three stable communities $C_1,C_2,C_3\subset V$ with $|C_1|<|C_2|=|C_3|$. Between snapshots $G_4$ and $G_8$, a fixed number of nodes from communities $C_2$ and $C_3$ gradually reassign their membership to $C_1$ at each step, leading to a new stable community structure that persists from $G_9$ to $G_{12}$. \textbf{b)} Coupling matrix. 
\textbf{c)} Leading spatial eigenvectors of $\bC_H$. 
\textbf{d)} Space--time observables obtained by lifting the leading spatial eigenvectors of $\widehat{\bC}_H$. 
\textbf{e)} Leading spatial eigenvalues of the full model together with two reduced models defined on subspaces $\mathcal W_t$ spanned by eigenvectors of the transition matrices $S_t$ and by their Nystr\"om approximations. 
\textbf{f)} and \textbf{h)} Two-dimensional embeddings of the space--time nodes obtained using the first two spatial eigenvectors of $\bC_H$ and the lifted first two spatial eigenvectors of $\widehat{\bC}_H$, respectively. 
\textbf{g)} and \textbf{i)} Community structure of the temporal network obtained by applying $k$-means clustering with $k=3$ to the embeddings of the space--time nodes obtained from the full and reduced model respectively. 
\textbf{j)} Performance of the special case of our approach in which only successive snapshots are coupled through the coupling matrix $H_s$. Leading spatial eigenvectors of $\bC_{H_s}$ together with the $k$-means clustering of the space--time nodes, using spatial eigenvectors 1--4 as feature coordinates.}
    \label{fig:numerical_example_2}
\end{figure}

\paragraph{Quantitative evaluation and method comparison.}
In examples~\ref{sec:example_1} and~\ref{sec:example_2}, we construct time-evolving graphs with well-defined communities. This allows us to quantitatively compare the clustering obtained by our algorithm with a reference clustering determined by the structure of the corresponding probability matrices (see figure~\ref{fig:reference_clustering}\textbf{a,b}) by computing the Adjusted Rand Index (ARI) \cite{Hubert_1985} between them. We use this comparison to demonstrate the accuracy, robustness and stability of our method (see table~\ref{tab:method_comparison}). In both examples, we achieve high ARI values, confirming that the obtained clusterings closely match the corresponding reference clusterings. Furthermore, the low standard deviations show that the results vary only slightly across different network realizations, demonstrating the stability of our method. We also compare our approach with the standard multilayer modularity maximization (MMM) method introduced in \cite{Mucha_2010}, which extends the well-established modularity quality function for static networks to multilayer networks. 
In table~\ref{tab:method_comparison}, we report the mean ARI values between the obtained and reference clusterings, together with their standard deviations, for both methods over 20 realizations of examples~\ref{sec:example_1} and~\ref{sec:example_2}. For more details on the comparative performance of our method and its analysis under different parameter regimes, see appendix~\ref{sec:appendix_5}.

\begin{table}[ht]
    \centering
    \begin{threeparttable}
        \caption{Average ARI values between the clusterings obtained using our
        approach and multilayer modularity maximization and the corresponding
        reference clusterings. The means and standard deviations are computed
        over 20 realizations of the networks in
        examples~\ref{sec:example_1} and~\ref{sec:example_2}.}
        \label{tab:method_comparison}

        \begin{tabular}{
            |>{\centering\arraybackslash}p{2.3cm}
            |>{\centering\arraybackslash}p{2.3cm}
            |>{\centering\arraybackslash}p{2.3cm}
            ||>{\centering\arraybackslash}p{2.3cm}
            |>{\centering\arraybackslash}p{2.3cm}|
        }
            \hline
            & \multicolumn{2}{c||}{\textbf{Our approach}}
            & \multicolumn{2}{c|}{\textbf{Modularity Maximization}} \\
            \cline{2-5}
            & Average ARI & St.\ dev.
            & Average ARI & St.\ dev. \\
            \hline
            \textbf{Example~1}
            & 0.86 & 0.011 & 0.71 & 0.008 \\
            \textbf{Example~2}
            & 0.97 & 0.003 & 1.0 & 0.0 \\
            \hline
        \end{tabular}
    \end{threeparttable}
\end{table}

\subsection{Example 3 -- opinion dynamics}
As a final example, we illustrate the applicability of our framework to clustering phenomena arising in interacting-agent systems, complementing transfer-operator based model reduction approaches~\cite{Wehlitz_2026}. We consider a modified version of the one-dimensional agent-based Hegselmann--Krause opinion dynamics model introduced in~\cite{Cahill_2025}, which extends the classical Hegselmann--Krause model~\cite{Hegselmann_2002} by incorporating interactions between voters and political parties. We generate $N=500$ agents, where the opinion of the $i$th agent at time $t$ is denoted by $x_i(t)\in[0,1]$. In addition to the voters, we assume there exist $N_p=3$ political parties represented by highly influential agents with positions $y_\alpha$ for $\alpha\in\{1,2,3\}$ in the opinion space.

In this setting, voters form their opinions both through interactions with other voters and through interactions with political parties. As in the classical model, voters tend to align with nearby opinions in the opinion space, but they are also attracted toward parties whose political stance is close to their own. Simultaneously, political parties adapt their positions in response to the surrounding voter opinions in order to attract support, while also maintaining sufficient separation from competing parties in order to preserve their political identity. The resulting coupled voter--party dynamics are given by
\begin{equation}
\label{eq:hegselmann_krause}
\begin{split}
\mathrm{d}x_i&=\gamma_{xx}F_{xx}(x_i)\mathrm{d}t+\gamma_{yx}F_{yx}(x_i)\mathrm{d}t+\sigma_x \mathrm{d}W_t^i,\\
\mathrm{d}y_{\alpha}&=\gamma_{xy}F_{xy}(y_{\alpha})\mathrm{d}t-\gamma_{yy}F_{yy}(y_{\alpha})\mathrm{d}t+\sigma_y \mathrm{d}W_t^{\alpha}.
\end{split}
\end{equation}
Here, $W^i, i=1,\dots,N$, and $W^\alpha, \alpha=1,\dots,N_p$, are independent Brownian motions representing random external influences on voters and parties. The coefficients $\sigma_x$ and $\sigma_y$ determine the strength of the stochastic noise in the dynamics. The interaction forces are defined by
\begin{equation*}
\begin{split}
F_{xx}(x_i)&=\frac{1}{N}\sum_{j=1}^{N}\chi_{R_{xx}}(x_j-x_i)(x_j-x_i),\\
F_{yx}(x_i)&=\frac{1}{N_p}\sum_{\beta=1}^{N_p}\chi_{R_{yx}}(y_{\beta}-x_i)(y_{\beta}-x_i),\\
F_{xy}(y_{\alpha})&=\frac{1}{N}\sum_{j=1}^{N}\chi_{R_{xy}}(x_j-y_{\alpha})(x_j-y_{\alpha}),\\
F_{yy}(y_{\alpha})&=\frac{1}{N_p}\sum_{\beta=1}^{N_p}\chi_{R_{yy}}(y_{\beta}-y_{\alpha})(y_{\beta}-y_{\alpha}),
\end{split}
\end{equation*}
where $\chi_R(x)$, $R>0$ denotes the interaction kernel  
\begin{equation*}
\chi_R(x)=
\begin{cases}
    1, & |x|\le R, \\
    0, & |x|> R.
\end{cases}
\end{equation*}
The parameters $R_{xx},R_{yx},R_{xy}$, and $R_{yy}$ denote the interaction radii governing voter--voter, voter--party, party--voter, and party--party interactions, respectively. The constants $\gamma_{xx},\gamma_{xy},\gamma_{yx},\gamma_{yy}\ge0$ determine the strength of the contribution of the corresponding interaction forces to the overall dynamics. A detailed discussion of these parameters and their sociopolitical interpretation can be found in~\cite{Cahill_2025}.

We simulate the system~\eqref{eq:hegselmann_krause} with parameters $R_{xx}=0.04, R_{yx}=0.1,R_{xy}=0.5, R_{yy}=0.05$ and $\gamma_{xx}=0.5,\gamma_{yx}=0.9,\gamma_{xy}=0.01,\gamma_{yy}=0.015$ for 390 time steps using time step size 0.05. The noise coefficients are set to $\sigma_x=0.025, \sigma_y=0.002$. Initially, voters' opinions are sampled uniformly from the interval $[0,1]$, while the initial party positions are given by $y_1(0)=0.15, y_2(0)=0.35$ and $y_3(0)=0.9$. The trajectories of voters in the opinion space are shown by blue lines in figure~\ref{fig:numerical_example_3}\textbf{a}, while the trajectories of the three political parties are shown in red.

Under this parameter regime (see section~4 in~\cite{Cahill_2025}), the system eventually reaches consensus, with both voters and parties converging toward a common opinion. However, the evolution exhibits several clearly distinguishable phases. Initially, three separate opinion communities emerge around the three political parties. Since parties $y_1$ and $y_2$ begin with relatively similar political positions, the corresponding voter groups gradually approach one another and eventually merge as illustrated in figure~\ref{fig:numerical_example_3}\textbf{a}. In contrast, the community associated with party $y_3$, which initially occupies the opposite side of the opinion spectrum, remains separated for a substantially longer period. Finally, the remaining two communities merge, leading to near-consensus across the population.

From this simulation, we construct a weighted temporal network on the set of voters consisting of 40 snapshots of the system taken at times $t=0,10,20,\dots,390$, indicated by dashed lines in figure~\ref{fig:numerical_example_3}\textbf{a}. For each snapshot, edge weights between nodes are determined according to the distance between the corresponding opinions, so that voters with similar opinions are connected by edges of larger weight
\[
w(x_i,x_j)=e^{-\frac{(x_i-x_j)^2}{2\varepsilon^2}},
\]
where we set $\varepsilon=0.05$. The resulting weighted adjacency matrices are then used as input for our framework.

We apply the reduced version of our algorithm with coupling parameter $\alpha=0.01$ in the matrix~$H$. The subspaces $\cW_t$ are chosen to be spanned by the dominant eigenvectors of the transition matrices associated with random walks on the network snapshots. We cluster space--time nodes using spatial eigenvectors 1 and 3 shown in figure~\ref{fig:numerical_example_3}\textbf{c} as feature coordinates. The resulting communities are shown in figure~\ref{fig:numerical_example_3}\textbf{b}. During the initial phase, the algorithm identifies three distinct communities (light blue, orange, and red). Around $t=15$, the orange and red communities merge into a larger dark blue community. Later, around $t=25$, the remaining light blue community merges with the dark blue one. Since the dark blue community already contains a substantially larger number of nodes, the algorithm identifies this event not as a symmetric merge but rather as the absorption of the smaller light blue community into the dominant dark blue community, which then persists until the end of the simulation.

\begin{figure}
    \centering
    \includegraphics[width=1\linewidth]{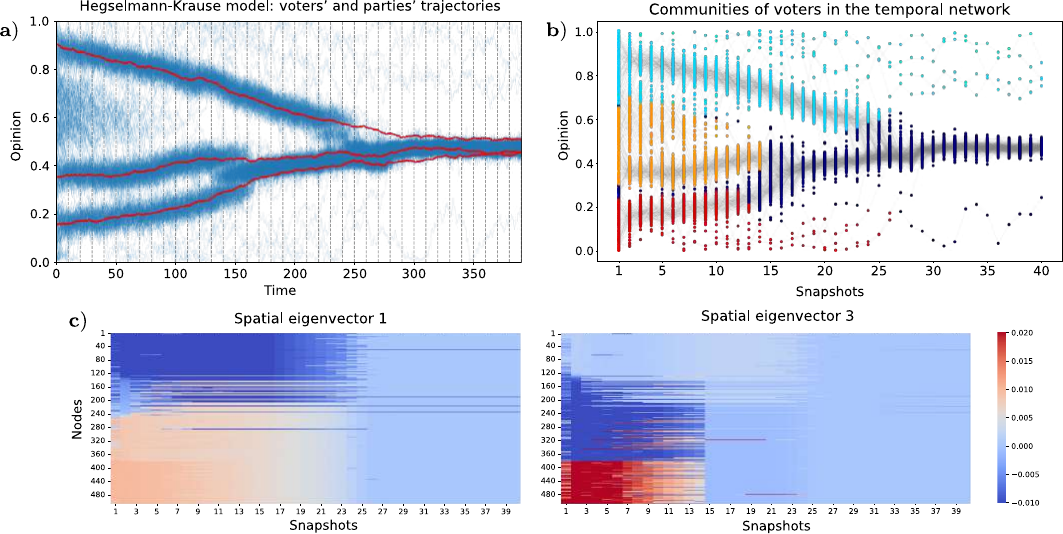}
    \caption{Application of our spectral clustering approach to a temporal network generated from the simulation of the Hegselmann--Krause opinion dynamics model. \textbf{a)} Trajectories of voters (blue) and political parties (red) in the opinion space over time. \textbf{b)} Voters at snapshots $t=0,10,\dots,390$, colored according to the communities identified by our approach. \textbf{c)} Spatial eigenvectors of $\bC_H$ used for clustering in this example. For clarity, the snapshot segments are shown as columns of the heatmaps and the nodes are ordered according to the values of the second snapshot segment of the first eigenvector.}
    \label{fig:numerical_example_3}
\end{figure}

\section{Conclusion}
\label{sec:conclusion}

In this work, we developed a transfer operator-based framework for community detection in temporal networks. Building on multi-view canonical correlation analysis, we introduced a general snapshot coupling scheme that allows correlations between snapshots to be incorporated into a unified objective function. This flexibility enables the method to adapt to different application scenarios while reducing the influence of noisy fluctuations and spurious temporal effects.

A central contribution of this work is the interpretation of the resulting multi-view CCA formulation through the lens of a spatio-temporal random walk on an augmented static network. We showed that the optimization problem reduces to a generalized eigenvalue problem whose spectral properties can be analyzed using known techniques. This perspective provides a natural dynamical interpretation of temporal communities as metastable structures of the associated spatio-temporal process and establishes a direct connection between temporal community detection and well-understood concepts from the analysis of static networks.

Furthermore, we investigated how temporal effects induced by the snapshot coupling scheme interact with spatial effects arising from the network structure. In particular, we have shown that not all dominant eigenvectors necessarily contain meaningful information about community structure and proposed heuristics for identifying informative spatial modes. Both our theoretical and computational analysis contribute to a better understanding of the spectral signatures of evolving communities and highlights the importance of separating temporal artifacts from genuine structural changes.

To improve scalability, we derived a reduced-order formulation based on projections onto low-dimensional subspaces. We showed that the key spectral properties of the full model carry over to the reduced setting and established error bounds that quantify the approximation quality. Numerical experiments show that the reduced model is capable of recovering evolving community structures while significantly reducing computational costs.

Several interesting directions remain for future research. On the theoretical side, a deeper understanding of the spectral properties of the spatio-temporal operator and more principled criteria for selecting informative spatial eigenvectors would be desirable. It would also be interesting to investigate adaptive and data-driven choices of the snapshot coupling matrix, as well as connections to other formulations of temporal and multilayer networks. Finally, applying the proposed methodology to empirical temporal networks from social, biological, and technological domains may provide additional insights into the structure and dynamics of complex systems.

\section*{Code and data availability}
The code and datasets used in this work are available at~\cite{Blaskovic_code_2026}.

%
%


\section*{Funding}
This work has been partially funded by the Deutsche Forschungsgemeinschaft (DFG, German Research Foundation) under Germany's Excellence Strategy – The Berlin Mathematics Research Center MATH+ (EXC-2046/1, EXC-2046/2, project ID: 390685689) and by the Federal Ministry of Research, Technology and Space (BMFTR) project EPISERVE (funding ID: 031L0324A), member of the German Modeling network for severe infectious diseases, MONID).





\printbibliography

@article{Trower_2025,
    author = {Trower, Maia and Djurdjevac Conrad, Natasa and Klus, Stefan},
    title = {Clustering time-evolving networks using the spatiotemporal graph {Laplacian}},
    journal = {Chaos: An Interdisciplinary Journal of Nonlinear Science},
    volume = {35},
    number = {1},
    pages = {013126},
    year = {2025},
    month = {01},
    abstract = {Time-evolving graphs arise frequently when modeling complex dynamical systems such as social networks, traffic flow, and biological processes. Developing techniques to identify and analyze communities in these time-varying graph structures is an important challenge. In this work, we generalize existing spectral clustering algorithms from static to dynamic graphs using canonical correlation analysis to capture the temporal evolution of clusters. Based on this extended canonical correlation framework, we define the spatiotemporal graph Laplacian and investigate its spectral properties. We connect these concepts to dynamical systems theory via transfer operators and illustrate the advantages of our method on benchmark graphs by comparison with existing methods. We show that the spatiotemporal graph Laplacian allows for a clear interpretation of cluster structure evolution over time for directed and undirected graphs.},
    issn = {1054-1500},
    doi = {10.1063/5.0228419},
    url = {https://doi.org/10.1063/5.0228419},
    %eprint = {https://pubs.aip.org/aip/cha/article-pdf/doi/10.1063/5.0228419/20343170/013126\_1\_5.0228419.pdf},
}

@misc{Froyland_2024,
      title={Spectral clustering of time-evolving networks using the inflated dynamic {Laplacian} for graphs}, 
      author={Gary Froyland and Manu Kalia and Péter Koltai},
      year={2024},
      eprint={2409.11984},
      archivePrefix={arXiv},
      primaryClass={cs.SI},
      url={https://arxiv.org/abs/2409.11984}, 
      %note={Available at: \url{https://arxiv.org/abs/2409.11984}}
}

@phdthesis{Djurdjevac_2012_phd,
author = {Djurdjevac, Nataša},
year = {2012},
title = {Methods for analyzing complex networks using random walker approaches},
type = {Dissertation},
school = {Freie Universität Berlin},
url = {http://dx.doi.org/10.17169/refubium-10282},
doi = {http://dx.doi.org/10.17169/refubium-10282},
}

@article{Sarich_2014,
title = {Modularity revisited: A novel dynamics-based concept for decomposing complex networks},
journal = {Journal of Computational Dynamics},
volume = {1},
number = {1},
pages = {191-212},
year = {2014},
issn = {2158-2491},
doi = {10.3934/jcd.2014.1.191},
url = {https://www.aimsciences.org/article/id/d733064e-a5c7-4504-96a6-90cdbbadb208},
author = {Marco Sarich and Djurdjevac Conrad, Natasa and Sharon Bruckner and Tim O. F. Conrad and Christof Schütte}
}

@article{Klus_2023,
  author  = {Klus, Stefan and Djurdjevac Conrad, Natasa},
  title   = {Koopman-Based Spectral Clustering of Directed and Time-Evolving Graphs},
  journal = {Journal of Nonlinear Science},
  volume  = {33},
  number  = {8},
  year    = {2023},
  doi     = {10.1007/s00332-022-09863-0}
}

@ARTICLE{Fowlkes_2004,
  author={Fowlkes, C. and Belongie, S. and Chung, F. and Malik, J.},
  journal={IEEE Transactions on Pattern Analysis and Machine Intelligence}, 
  title={Spectral grouping using the {Nystr\" om} method}, 
  year={2004},
  volume={26},
  number={2},
  pages={214-225},
  doi={10.1109/TPAMI.2004.1262185},
}

@article{von_Luxburg_2007,
  author  = {von Luxburg, Ulrike},
  title   = {A Tutorial on Spectral Clustering},
  journal = {Statistics and Computing},
  volume  = {17},
  number  = {4},
  pages   = {395--416},
  year    = {2007},
  doi     = {10.1007/s11222-007-9033-z},
}

@article{Aslak_2018,
  title = {Constrained information flows in temporal networks reveal intermittent communities},
  author = {Aslak, Ulf and Rosvall, Martin and Lehmann, Sune},
  journal = {Phys. Rev. E},
  volume = {97},
  issue = {6},
  pages = {062312},
  numpages = {10},
  year = {2018},
  month = {6},
  publisher = {American Physical Society},
  doi = {10.1103/PhysRevE.97.062312},
  url = {https://link.aps.org/doi/10.1103/PhysRevE.97.062312},
}

@article{Bazzi_2016,
author = {Bazzi, Marya and Porter, Mason A. and Williams, Stacy and McDonald, Mark and Fenn, Daniel J. and Howison, Sam D.},
title = {Community Detection in Temporal Multilayer Networks, with an Application to Correlation Networks},
journal = {Multiscale Modeling \& Simulation},
volume = {14},
number = {1},
pages = {1-41},
year = {2016},
doi = {10.1137/15M1009615},
URL = {https://doi.org/10.1137/15M1009615},
%eprint = {https://doi.org/10.1137/15M1009615},
abstract = { Networks are a convenient way to represent complex systems of interacting entities. Many networks contain “communities” of nodes that are more densely connected to each other than to nodes in the rest of the network. In this paper, we investigate the detection of communities in temporal networks represented as multilayer networks. As a focal example, we study time-dependent financial-asset correlation networks. We first argue that the use of the “modularity” quality function---which is defined by comparing edge weights in an observed network to expected edge weights in a “null network''---is application-dependent. We differentiate between “null networks” and “null models” in our discussion of modularity maximization, and we highlight that the same null network can correspond to different null models. We then investigate a multilayer modularity-maximization problem to identify communities in temporal networks. Our multilayer analysis depends only on the form of the maximization problem and not on the specific quality function that one chooses. We introduce a diagnostic to measure persistence of community structure in a multilayer network partition. We prove several results that describe how the multilayer maximization problem measures a trade-off between static community structure within layers and larger values of persistence across layers. We also discuss some computational issues that the popular “Louvain” heuristic faces with temporal multilayer networks and suggest ways to mitigate them.},
}

@article{Klus_2024,
doi = {10.1088/2632-072X/ad28fe},
url = {https://doi.org/10.1088/2632-072X/ad28fe},
year = {2024},
month = {2},
publisher = {IOP Publishing},
volume = {5},
number = {1},
pages = {015014},
author = {Klus, Stefan and Trower, Maia},
title = {Transfer operators on graphs: spectral clustering and beyond},
journal = {Journal of Physics: Complexity},
}

@inbook{Djurdjevac_2010,
author = {Natasa Djurdjevac and Marco Sarich and Christof Schütte},
title = {On Markov State Models for Metastable Processes},
booktitle = {Proceedings of the International Congress of Mathematicians 2010 (ICM 2010)},
publisher = {Proceedings of the International Congress of Mathematicians},
year = {2010},
pages = {3105-3131},
doi = {10.1142/9789814324359_0182},
URL = {https://www.worldscientific.com/doi/abs/10.1142/9789814324359_0182},
%eprint = {https://www.worldscientific.com/doi/pdf/10.1142/9789814324359_0182},
    abstract = { Abstract We consider Markov processes on large state spaces and want to find low-dimensional structure-preserving approximations of the process in the sense that the longest timescales of the dynamics of the original process are reproduced well. Recent years have seen the advance of so-called Markov state models (MSM) for processes on very large state spaces exhibiting metastable dynamics. It has been demonstrated that MSMs are especially useful for modelling the interesting slow dynamics of biomolecules (cf. Noe et al, PNAS(106) 2009) and materials. From the mathematical perspective, MSMs result from Galerkin projection of the transfer operator underlying the original process onto some low-dimensional subspace which leads to an approximation of the dominant eigenvalues of the transfer operators and thus of the longest timescales of the original dynamics. Until now, most articles on MSMs have been based on full subdivisions of state space, i.e., Galerkin projections onto subspaces spanned by indicator functions. We show how to generalize MSMs to alternative low-dimensional subspaces with superior approximation properties, and how to analyse the approximation quality (dominant eigenvalues, propagation of functions) of the resulting MSMs. To this end, we give an overview of the construction of MSMs, the associated stochastics and functional-analysis background, and its algorithmic consequences. Furthermore, we illustrate the mathematical construction with numerical examples. }
}

@article{Djurdjevac_2012,
author = {Djurdjevac, Natasa and Sarich, Marco and Sch\"{u}tte, Christof},
title = {Estimating the Eigenvalue Error of {M}arkov State Models},
journal = {Multiscale Modeling \& Simulation},
volume = {10},
number = {1},
pages = {61-81},
year = {2012},
doi = {10.1137/100798910},
URL = {https://doi.org/10.1137/100798910},
%eprint = {https://doi.org/10.1137/100798910},
abstract = { We consider a continuous-time, ergodic Markov process on a large continuous or discrete state space. The process is assumed to exhibit a number of metastable sets. Markov state models (MSMs) are designed to represent the effective dynamics of such a process by a Markov chain that jumps between the metastable sets with the transition rates of the original process. MSMs have been used for a number of applications, including molecular dynamics (cf. [F. Noé et al., Proc. Natl. Acad. Sci. USA, 106 (2009), pp. 19011–19016]), for more than a decade. The rigorous and fully general (no zero temperature limit or comparable restrictions) analysis of their approximation quality, however, has only recently begun. Our first article on this topics [M. Sarich, F. Noé, and Ch. Schütte, Multiscale Model. Simul., 8 (2010), pp. 1154–1177] introduces an error bound for the difference in propagation of probability densities between the MSM and the original process on long timescales. Herein we provide upper bounds for the error in the eigenvalues between the MSM and the original process, which means that we analyze how well the longest timescales in the original process are approximated by the MSM. Our findings are illustrated by numerical experiments.}
}

@book{Schuette_2013,
  author    = {Christof Sch{\"u}tte and Marco Sarich},
  title     = {Metastability and Markov State Models in Molecular Dynamics: Modeling, Analysis, Algorithmic Approaches},
  series    = {Courant Lecture Notes in Mathematics},
  volume    = {24},
  publisher = {American Mathematical Society},
  address   = {Providence, RI},
  year      = {2013},
  isbn      = {978-0-8218-4359-8},
  doi       = {10.1090/cln/024},
  mrnumber  = {3155191}
}

@Inbook{Huisinga_2006,
author="Huisinga, Wilhelm
and Schmidt, Bernd",
title="Metastability and Dominant Eigenvalues of Transfer Operators",
bookTitle="New Algorithms for Macromolecular Simulation",
year="2006",
publisher="Springer Berlin Heidelberg",
address="Berlin, Heidelberg",
pages="167--182",
isbn="978-3-540-31618-3",
doi="10.1007/3-540-31618-3_11",
url="https://doi.org/10.1007/3-540-31618-3_11"
}

@book{Atkinson_1972,
  author    = {F. V. Atkinson},
  title     = {Multiparameter Eigenvalue Problems. Volume I: Matrices and Compact Operators},
  series    = {Mathematics in Science and Engineering},
  volume    = {82},
  publisher = {Academic Press},
  address   = {New York and London},
  year      = {1972},
  url       = {https://doi.org/10.1137/1015089}
}

@article{Eisenmann_2025,
author = {Eisenmann, Henrik},
title = {A {Newton} Method for Solving Locally Definite Multiparameter Eigenvalue Problems by Multi-index},
journal = {SIAM Journal on Matrix Analysis and Applications},
volume = {46},
number = {2},
pages = {906-933},
year = {2025},
doi = {10.1137/24M1652775},
URL = {https://doi.org/10.1137/24M1652775},
%eprint = {https://doi.org/10.1137/24M1652775},
abstract = { Abstract. We present a new approach to compute eigenvalues and eigenvectors of locally definite multiparameter eigenvalue problems by their signed multi-index. The method has the interpretation of a semismooth Newton method applied to certain functions that have a unique zero. We can therefore show local quadratic convergence, and for certain extreme eigenvalues even global linear convergence of the method. Local definiteness is a weaker condition than right and left definiteness, which is often considered for multiparameter eigenvalue problems. These conditions are naturally satisfied for multiparameter Sturm–Liouville problems that arise when separation of variables can be applied to multidimensional boundary eigenvalue problems. }
}

@article{Otto_2021,
  author  = {Otto, Samuel E. and Rowley, Clarence W.},
  title   = {Koopman Operators for Estimation and Control of Dynamical Systems},
  journal = {Annual Review of Control, Robotics, and Autonomous Systems},
  year    = {2021},
  volume  = {4},
  pages   = {59--87},
  doi     = {10.1146/annurev-control-071020-010108}
}

@article{Klus_2016,
title = {On the numerical approximation of the {P}erron--{F}robenius and {K}oopman operator},
journal = {Journal of Computational Dynamics},
volume = {3},
number = {1},
pages = {51-79},
year = {2016},
issn = {2158-2491},
doi = {10.3934/jcd.2016003},
url = {https://www.aimsciences.org/article/id/99a4878a-5263-40bf-b4e8-bbf7c2730a2a},
author = {Stefan  Klus and Péter Koltai and Christof  Schütte}
}

@article{Budisic_2012,
title = {Geometry of the ergodic quotient reveals coherent structures in flows},
journal = {Physica D: Nonlinear Phenomena},
volume = {241},
number = {15},
pages = {1255-1269},
year = {2012},
issn = {0167-2789},
doi = {10.1016/j.physd.2012.04.006},
url = {https://www.sciencedirect.com/science/article/pii/S0167278912001108},
author = {Marko Budišić and Igor Mezić},
}

@article{Hotelling_1936,
  author  = {Hotelling, Harold},
  title   = {Relations Between Two Sets of Variates},
  journal = {Biometrika},
  volume  = {28},
  number  = {3/4},
  pages   = {321--377},
  year    = {1936},
  url     = {https://doi.org/10.1093/biomet/28.3-4.321},
  doi     = {10.1093/biomet/28.3-4.321}
}

@book{Shawe-Taylor_2004,
  author    = {Shawe-Taylor, John and Cristianini, Nello},
  title     = {Kernel Methods for Pattern Analysis},
  publisher = {Cambridge University Press},
  address   = {Cambridge, United Kingdom},
  year      = {2004},
  isbn      = {0521813972},
  url       = {https://doi.org/10.1017/CBO9780511809682},
  doi       = {10.1017/CBO9780511809682},
}

@article{Holme_2012,
title = {Temporal networks},
journal = {Physics Reports},
volume = {519},
number = {3},
pages = {97-125},
year = {2012},
note = {Temporal Networks},
issn = {0370-1573},
doi = {10.1016/j.physrep.2012.03.001},
url = {https://www.sciencedirect.com/science/article/pii/S0370157312000841},
author = {Petter Holme and Jari Saramäki}
}

@article{Masuda_2013,
  author    = {Naoki Masuda and Petter Holme},
  title     = {Predicting and controlling infectious disease epidemics using temporal networks},
  journal   = {F1000Prime Reports},
  year      = {2013},
  volume    = {5},
  pages     = {6},
  doi       = {10.12703/P5-6},
  pmid      = {23513178},
  pmcid     = {PMC3590785},
  publisher = {F1000 Research Ltd}
}

@article{Mancastroppa_2020,
  title = {Active and inactive quarantine in epidemic spreading on adaptive activity-driven networks},
  author = {Mancastroppa, Marco and Burioni, Raffaella and Colizza, Vittoria and Vezzani, Alessandro},
  journal = {Phys. Rev. E},
  volume = {102},
  issue = {2},
  pages = {020301(R)},
  numpages = {6},
  year = {2020},
  month = {8},
  publisher = {American Physical Society},
  doi = {10.1103/PhysRevE.102.020301},
  url = {https://link.aps.org/doi/10.1103/PhysRevE.102.020301}
}

@article{Kivela_2014,
    author = {Kivelä, Mikko and Arenas, Alex and Barthelemy, Marc and Gleeson, James P. and Moreno, Yamir and Porter, Mason A.},
    title = {Multilayer networks},
    journal = {Journal of Complex Networks},
    volume = {2},
    number = {3},
    pages = {203-271},
    year = {2014},
    month = {07},
    abstract = {In most natural and engineered systems, a set of entities interact with each other in complicated patterns that can encompass multiple types of relationships, change in time and include other types of complications. Such systems include multiple subsystems and layers of connectivity, and it is important to take such ‘multilayer’ features into account to try to improve our understanding of complex systems. Consequently, it is necessary to generalize ‘traditional’ network theory by developing (and validating) a framework and associated tools to study multilayer systems in a comprehensive fashion. The origins of such efforts date back several decades and arose in multiple disciplines, and now the study of multilayer networks has become one of the most important directions in network science. In this paper, we discuss the history of multilayer networks (and related concepts) and review the exploding body of work on such networks. To unify the disparate terminology in the large body of recent work, we discuss a general framework for multilayer networks, construct a dictionary of terminology to relate the numerous existing concepts to each other and provide a thorough discussion that compares, contrasts and translates between related notions such as multilayer networks, multiplex networks, interdependent networks, networks of networks and many others. We also survey and discuss existing data sets that can be represented as multilayer networks. We review attempts to generalize single-layer-network diagnostics to multilayer networks. We also discuss the rapidly expanding research on multilayer-network models and notions like community structure, connected components, tensor decompositions and various types of dynamical processes on multilayer networks. We conclude with a summary and an outlook.},
    issn = {2051-1310},
    doi = {10.1093/comnet/cnu016},
    url = {https://doi.org/10.1093/comnet/cnu016},
    %eprint = {https://academic.oup.com/comnet/article-pdf/2/3/203/9130906/cnu016.pdf},
}

@article{Rossetti_2019,
author = {Rossetti, Giulio and Cazabet, R\'{e}my},
title = {Community Discovery in Dynamic Networks: A Survey},
year = {2018},
issue_date = {March 2019},
publisher = {Association for Computing Machinery},
address = {New York, NY, USA},
volume = {51},
number = {2},
issn = {0360-0300},
url = {https://doi.org/10.1145/3172867},
doi = {10.1145/3172867},
journal = {ACM Comput. Surv.},
month = feb,
articleno = {35},
numpages = {37}
}

@article{De_Domenico_2013,
  title = {Mathematical Formulation of Multilayer Networks},
  author = {De Domenico, Manlio and Sol\'e-Ribalta, Albert and Cozzo, Emanuele and Kivel\"a, Mikko and Moreno, Yamir and Porter, Mason A. and G\'omez, Sergio and Arenas, Alex},
  journal = {Phys. Rev. X},
  volume = {3},
  issue = {4},
  pages = {041022},
  numpages = {15},
  year = {2013},
  month = {12},
  publisher = {American Physical Society},
  doi = {10.1103/PhysRevX.3.041022},
  url = {https://link.aps.org/doi/10.1103/PhysRevX.3.041022}
}

@article{Taylor_2017,
    author = {Taylor, Dane and Myers, Sean A. and Clauset, Aaron and Porter, Mason A. and Mucha, Peter J.},
    title = {Eigenvector-Based Centrality Measures for Temporal Networks},
    journal = {Multiscale Modeling \& Simulation},
    volume = {15},
    number = {1},
    pages = {537-574},
    year = {2017},
    doi = {10.1137/16M1066142},
    URL = {https://doi.org/10.1137/16M1066142},
    %eprint = {https://doi.org/10.1137/16M1066142},
    abstract = { Numerous centrality measures have been developed to quantify the importances of nodes in time-independent networks, and many of them can be expressed as the leading eigenvector of some matrix. With the increasing availability of network data that changes in time, it is important to extend such eigenvector-based centrality measures to time-dependent networks. In this paper, we introduce a principled generalization of network centrality measures that is valid for any eigenvector-based centrality. We consider a temporal network with \$N\$ nodes as a sequence of \$T\$ layers that describe the network during different time windows, and we couple centrality matrices for the layers into a supracentrality matrix of size \$NT\times NT\$ whose dominant eigenvector gives the centrality of each node \$i\$ at each time \$t\$. We refer to this eigenvector and its components as a joint centrality, as it reflects the importances of both the node \$i\$ and the time layer \$t\$. We also introduce the concepts of marginal and conditional centralities, which facilitate the study of centrality trajectories over time. We find that the strength of coupling between layers is important for determining multiscale properties of centrality, such as localization phenomena and the time scale of centrality changes. In the strong-coupling regime, we derive expressions for time-averaged centralities, which are given by the zeroth-order terms of a singular perturbation expansion. We also study first-order terms to obtain first-order-mover scores, which concisely describe the magnitude of the nodes' centrality changes over time. As examples, we apply our method to three empirical temporal networks: the United States Ph.D. exchange in mathematics, costarring relationships among top-billed actors during the Golden Age of Hollywood, and citations of decisions from the United States Supreme Court. } 
}

@article{Klus_2024_1,
doi = {10.1088/2632-072X/ad9e60},
url = {https://doi.org/10.1088/2632-072X/ad9e60},
year = {2024},
month = {12},
publisher = {IOP Publishing},
volume = {5},
number = {4},
pages = {041001},
author = {Klus, Stefan and Djurdjevac Conrad, Nataša},
title = {Dynamical systems and complex networks: a {K}oopman operator perspective},
journal = {Journal of Physics: Complexity}
}

@article{Knyazev_2010,
author = {Knyazev, Andrew V. and Argentati, Merico E.},
title = {Rayleigh–{R}itz Majorization Error Bounds with Applications to {FEM}},
journal = {SIAM Journal on Matrix Analysis and Applications},
volume = {31},
number = {3},
pages = {1521-1537},
year = {2010},
doi = {10.1137/08072574X},
URL = {https://doi.org/10.1137/08072574X},
%eprint = {https://doi.org/10.1137/08072574X},
abstract = { The Rayleigh–Ritz (RR) method finds the stationary values, called Ritz values, of the Rayleigh quotient on a given trial subspace as approximations to eigenvalues of a Hermitian operator A. If the trial subspace is A-invariant, the Ritz values are exactly some of the eigenvalues of A. Given two subspaces \$\mathcal{X}\$ and \$\mathcal{Y}\$ of the same finite dimension, such that \$\mathcal{X}\$ is A-invariant, the absolute changes in the Ritz values of A with respect to \$\mathcal{X}\$ compared to the Ritz values with respect to \$\mathcal{Y}\$ represent the RR absolute eigenvalue approximation error. Our first main result is a sharp majorization-type RR error bound in terms of the principal angles between \$\mathcal{X}\$ and \$\mathcal{Y}\$ for an arbitrary A-invariant \$\mathcal{X}\$, which was a conjecture in [SIAM J. Matrix Anal. Appl., 30 (2008), pp. 548–559]. Second, we prove a novel type of RR error bound that deals with the products of the errors, rather than the sums. Third, we establish majorization bounds for the relative errors. We extend our bounds to the case \$\dim\mathcal{X}\leq\dim\mathcal{Y}<\infty\$ in Hilbert spaces and apply them in the context of the finite element method. }
}

@inproceedings{Aynaud_2010,
  TITLE = {{Static community detection algorithms for evolving networks}},
  AUTHOR = {Aynaud, Thomas and Guillaume, Jean-Loup},
  URL = {https://inria.hal.science/inria-00492058},
  BOOKTITLE = {{WiOpt'10: Modeling and Optimization in Mobile, Ad Hoc, and Wireless Networks}},
  ADDRESS = {Avignon, France},
  PAGES = {508-514},
  YEAR = {2010},
  MONTH = May,
  HAL_ID = {inria-00492058},
  HAL_VERSION = {v1},
}

@incollection{Aynaud_2013,
  author    = {Thomas Aynaud and Eric Fleury and Jean-Loup Guillaume and Qing Wang},
  title     = {Communities in Evolving Networks: Definitions, Detection, and Analysis Techniques},
  booktitle = {Dynamics On and Of Complex Networks, Volume 2},
  editor    = {Abhijit Mukherjee and Monojit Choudhury and Fernando Peruani and Niloy Ganguly and Bivas Mitra},
  series    = {Modeling and Simulation in Science, Engineering and Technology},
  publisher = {Birkh{\"a}user},
  address   = {New York, NY},
  year      = {2013},
  pages     = {159--200},
  doi       = {10.1007/978-1-4614-6729-8_9}
}

@Article{Hegselmann_2002,
journal={Journal of Artificial Societies and Social Simulation},
author={Rainer Hegselmann and Ulrich Krause},
title={Opinion Dynamics and Bounded Confidence Models, Analysis and Simulation},
year={2002},
pages={1-2},
volume={5},
number={3},
abstract={When does opinion formation within an interacting group lead to consensus, polarization or fragmentation? The article investigates various models for the dynamics of continuous opinions by analytical methods as well as by computer simulations. Section 2 develops within a unified framework the classical model of consensus formation, the variant of this model due to Friedkin and Johnsen, a time-dependent version and a nonlinear version with bounded confidence of the agents. Section 3 presents for all these models major analytical results. Section 4 gives an extensive exploration of the nonlinear model with bounded confidence by a series of computer simulations. An appendix supplies needed mathematical definitions, tools, and theorems.},
keywords={opinion dynamics; consensus/dissent; bounded confidence; nonlinear dynamical systems.},
%doi={None},
url={https://ideas.repec.org/a/jas/jasssj/2002-5-2.html},
}

@article{Cahill_2025,
title = {A modified {Hegselmann}–{Krause} model for interacting voters and political parties},
journal = {Physica A: Statistical Mechanics and its Applications},
volume = {665},
pages = {130490},
year = {2025},
issn = {0378-4371},
doi = {10.1016/j.physa.2025.130490},
url = {https://www.sciencedirect.com/science/article/pii/S0378437125001426},
author = {Patrick Cahill and Georg A. Gottwald}
}

@article{Blaskovic_2025,
  author  = {Bla{\v{s}}kovi{\'c}, Filip and Conrad, Tim O. F. and Klus, Stefan and Djurdjevac Conrad, Nata{\v{s}}a},
  title   = {Random Walk Based Snapshot Clustering for Detecting Community Dynamics in Temporal Networks},
  journal = {Scientific Reports},
  year    = {2025},
  volume  = {15},
  number  = {1},
  pages   = {24414},
  doi     = {10.1038/s41598-025-09340-0},
  url     = {https://doi.org/10.1038/s41598-025-09340-0}
}

@book{Golub_2013,
author = {Golub, Gene H. and Van Loan, Charles F.},
title = {Matrix Computations - 4th Edition},
publisher = {Johns Hopkins University Press},
year = {2013},
doi = {10.1137/1.9781421407944},
address = {Philadelphia, PA},
edition   = {},
URL = {https://epubs.siam.org/doi/abs/10.1137/1.9781421407944},
%eprint = {https://epubs.siam.org/doi/pdf/10.1137/1.9781421407944}
}

@article{Busiello_2017,
  author  = {Busiello, Daniel M. and Suweis, Samir and Hidalgo, Jorge and Maritan, Amos},
  title   = {Explorability and the Origin of Network Sparsity in Living Systems},
  journal = {Scientific Reports},
  year    = {2017},
  volume  = {7},
  number  = {1},
  pages   = {12323},
  doi     = {10.1038/s41598-017-12521-1},
  url     = {https://doi.org/10.1038/s41598-017-12521-1}
}

@article{Demaine_2019,
title = {Structural sparsity of complex networks: Bounded expansion in random models and real-world graphs},
journal = {Journal of Computer and System Sciences},
volume = {105},
pages = {199-241},
year = {2019},
issn = {0022-0000},
doi = {10.1016/j.jcss.2019.05.004},
url = {https://www.sciencedirect.com/science/article/pii/S0022000019300418},
author = {Erik D. Demaine and Felix Reidl and Peter Rossmanith and Fernando {Sánchez Villaamil} and Somnath Sikdar and Blair D. Sullivan}
}

@article{Froyland_2010,
    author = {Froyland, Gary and Santitissadeekorn, Naratip and Monahan, Adam},
    title = {Transport in time-dependent dynamical systems: Finite-time coherent sets},
    journal = {Chaos: An Interdisciplinary Journal of Nonlinear Science},
    volume = {20},
    number = {4},
    pages = {043116},
    year = {2010},
    month = {11},
    abstract = {We study the transport properties of nonautonomous chaotic dynamical systems over a finite-time duration. We are particularly interested in those regions that remain coherent and relatively nondispersive over finite periods of time, despite the chaotic nature of the system. We develop a novel probabilistic methodology based upon transfer operators that automatically detect maximally coherent sets. The approach is very simple to implement, requiring only singular vector computations of a matrix of transitions induced by the dynamics. We illustrate our new methodology on an idealized stratospheric flow and in two and three-dimensional analyses of European Centre for Medium Range Weather Forecasting (ECMWF) reanalysis data.},
    issn = {1054-1500},
    doi = {10.1063/1.3502450},
    url = {https://doi.org/10.1063/1.3502450},
    %eprint = {https://pubs.aip.org/aip/cha/article-pdf/doi/10.1063/1.3502450/13278544/043116_1_online.pdf},
}

@article{Fackeldey_2019,
    author = {Fackeldey, Konstantin and Koltai, Péter and Névir, Peter and Rust, Henning and Schild, Axel and Weber, Marcus},
    title = {From metastable to coherent sets---Time-discretization schemes},
    journal = {Chaos: An Interdisciplinary Journal of Nonlinear Science},
    volume = {29},
    number = {1},
    pages = {012101},
    year = {2019},
    month = {01},
    abstract = {Given a time-dependent stochastic process with trajectories x(t) in a space Ω, there may be sets such that the corresponding trajectories only very rarely cross the boundaries of these sets. We can analyze such a process in terms of metastability or coherence. Metastable sets M are defined in space M⊂Ω, and coherent sets M(t)⊂Ω are defined in space and time. Hence, if we extend the space Ω by the time-variable t, coherent sets are metastable sets in Ω×[0,∞) of an appropriate space-time process. This relation can be exploited, because there already exist spectral algorithms for the identification of metastable sets. In this article, we show that these well-established spectral algorithms (like PCCA+, Perron Cluster Cluster Analysis) also identify coherent sets of non-autonomous dynamical systems. For the identification of coherent sets, one has to compute a discretization (a matrix T) of the transfer operator of the process using a space-time-discretization scheme. The article gives an overview about different time-discretization schemes and shows their applicability in two different fields of application.},
    issn = {1054-1500},
    doi = {10.1063/1.5058128},
    url = {https://doi.org/10.1063/1.5058128},
    %eprint = {https://pubs.aip.org/aip/cha/article-pdf/doi/10.1063/1.5058128/13761746/012101_1_online.pdf},
}

@Article{Padberg-Gehle_2017,
AUTHOR = {Padberg-Gehle, K. and Schneide, C.},
TITLE = {Network-based study of {Lagrangian} transport and mixing},
JOURNAL = {Nonlinear Processes in Geophysics},
VOLUME = {24},
YEAR = {2017},
NUMBER = {4},
PAGES = {661--671},
URL = {https://npg.copernicus.org/articles/24/661/2017/},
DOI = {10.5194/npg-24-661-2017}
}

@misc{Wehlitz_2026,
      title={Data-driven Reduction of Transfer Operators for Particle Clustering Dynamics}, 
      author={Nathalie Wehlitz and Grigorios A. Pavliotis and Christof Schütte and Stefanie Winkelmann},
      year={2026},
      eprint={2601.02932},
      archivePrefix={arXiv},
      primaryClass={cond-mat.stat-mech},
      url={https://arxiv.org/abs/2601.02932}, 
      %note={Available at: \url{https://arxiv.org/abs/2601.02932}}
}

@article{Mucha_2010,
author = {Peter J. Mucha  and Thomas Richardson  and Kevin Macon  and Mason A. Porter  and Jukka-Pekka Onnela },
title = {Community Structure in Time-Dependent, Multiscale, and Multiplex Networks},
journal = {Science},
volume = {328},
number = {5980},
pages = {876-878},
year = {2010},
doi = {10.1126/science.1184819},
URL = {https://www.science.org/doi/abs/10.1126/science.1184819},
%eprint = {https://www.science.org/doi/pdf/10.1126/science.1184819},
abstract = {Networks are often characterized by clusters of constituents that interact more closely with each other and have more connections to one another than they do with the rest of the components of the network. However, systematically identifying and studying such community structure in complicated networks is not easy, especially when the network interactions change over time or contain multiple types of connections, as seen in many biological regulatory networks or social networks. Mucha et al. (p. 876) developed a mathematical method to allow detection of communities that may be critical functional units of such networks. Application to real-world tasks—like making sense of the voting record in the U.S. Senate—demonstrated the promise of the method. A general mathematical method used to identify closely interacting groups can explain the behavior of complicated networks. Network science is an interdisciplinary endeavor, with methods and applications drawn from across the natural, social, and information sciences. A prominent problem in network science is the algorithmic detection of tightly connected groups of nodes known as communities. We developed a generalized framework of network quality functions that allowed us to study the community structure of arbitrary multislice networks, which are combinations of individual networks coupled through links that connect each node in one network slice to itself in other slices. This framework allows studies of community structure in a general setting encompassing networks that evolve over time, have multiple types of links (multiplexity), and have multiple scales.}}

@article{Fortunato_2016,
title = {Community detection in networks: A user guide},
journal = {Physics Reports},
volume = {659},
pages = {1-44},
year = {2016},
note = {Community detection in networks: A user guide},
issn = {0370-1573},
doi = {10.1016/j.physrep.2016.09.002},
url = {https://www.sciencedirect.com/science/article/pii/S0370157316302964},
author = {Santo Fortunato and Darko Hric}
}

@article{Knyazev_1997,
 ISSN = {00255718, 10886842},
 URL = {http://www.jstor.org/stable/2153754},
 author = {Andrew V. Knyazev},
 journal = {Mathematics of Computation},
 number = {219},
 pages = {985--995},
 publisher = {American Mathematical Society},
 title = {New Estimates for Ritz Vectors},
 urldate = {2026-08-12},
 volume = {66},
 year = {1997},
 doi  = {10.1090/S0025-5718-97-00855-7}
}

@misc{Nakatsukasa_2019,
      title={Sharp error bounds for {R}itz vectors and approximate singular vectors}, 
      author={Yuji Nakatsukasa},
      year={2019},
      eprint={1810.02532},
      archivePrefix={arXiv},
      primaryClass={math.NA},
      url={https://arxiv.org/abs/1810.02532}, 
      %note = {Available at \url{https://arxiv.org/abs/1810.02532}}
}

@article{Hubert_1985,
    author  = {Hubert, Lawrence and Arabie, Phipps},
    title   = {Comparing Partitions},
    journal = {Journal of Classification},
    volume  = {2},
    number  = {1},
    pages   = {193--218},
    year    = {1985},
    doi     = {10.1007/BF01908075}
}

@Inbook{Cazabet_2019,
author="Cazabet, Remy
and Rossetti, Giulio",
title="Challenges in Community Discovery on Temporal Networks",
bookTitle="Temporal Network Theory",
year="2019",
publisher="Springer International Publishing",
address="Cham",
pages="181--197",
isbn="978-3-030-23495-9",
doi="10.1007/978-3-030-23495-9_10",
url="https://doi.org/10.1007/978-3-030-23495-9_10"
}

@misc{Blaskovic_code_2026,
  author       = {Blašković, Filip},
  title        = {Supplementary code for spectral clustering of time-evolving networks
                   using spatio-temporal random walks
                  },
  year         = 2026,
  publisher    = {Zenodo},
  version      = {v1.0},
  doi          = {10.5281/zenodo.22713370},
  url          = {https://doi.org/10.5281/zenodo.22713370},
  %note         = {\url{https://doi.org/10.5281/zenodo.22713370}}
}

\appendix

\section{Proofs of lemmas from sections~\ref{sec:Spatio-temporal_random_walk} and~\ref{section:covariance_operators}}
\label{sec:appendix_1}

\begin{proof}[Proof of lemma~\ref{le:H_random_walk}]
Define $\bpi = [\pi_1, \dots, \pi_M]^\top \in \mathbb{R}^M$ by
\[
\pi_t = \frac{r_t}{R}>0, \qquad R = \sum_{t=1}^M r_t.
\]
Then $\sum_{t=1}^M \pi_t = 1$, so $\bpi$ is a probability distribution. Moreover, for all $t,s$, we have
\[
\pi_t h_{ts}
= \frac{r_t}{R} \cdot \frac{w_{ts}}{r_t}
= \frac{r_s}{R} \cdot \frac{w_{st}}{r_s}
= \pi_s h_{st},
\]
which proves the detailed balance condition. The weighted degree of node $t$ in $G_H$ is
\[
\deg_H(t) = \sum_{s=1}^M (W_H)_{ts} = r_t,
\]
so that the transition matrix of a random walk on $G_H$ is indeed
\[
D_H^{-1} W_H = H,
\]
where $D_H = \diag(\deg_H(1), \dots, \deg_H(M))$. This random walk is reversible by the detailed balance condition and the uniqueness of $\bpi$ follows if $G_H$ is connected.
\end{proof}

\begin{proof}[Proof of lemma~\ref{thm:row_stochasticity}]
Since all $S_t$ are row-stochastic, we have $S_t\bone=\bone$ for all $1\le t\le M$. Consequently, for $t<s$, we obtain $K_{ts}\bone=(S_t\cdots S_{s-1})\bone=\bone$ and by the definition of probability distributions $\mu_t$, it holds that $(K_{ts})^{\top}\bmu_t=\bmu_s$. Hence, $T_{ts}\bone=D_{\mu_s}^{-1}(K_{ts})^{\top}D_{\mu_t}\bone=D_{\mu_s}^{-1}(K_{ts})^{\top}\bmu_t=D_{\mu_s}^{-1}\bmu_{s}=\bone$. Finally, since the coupling matrix $H$ is row-stochastic, we obtain for $\bone_{MN}\in\mathbb R^{MN}$ and $\bone_N\in\mathbb R^N$ that $\bC_H\bone_{MN}=[(\sum_{s=1}^Mh_{1s})\bone_N^{\top},\dots,(\sum_{s=1}^Mh_{Ms})\bone_N^{\top}]^{\top}=\bone_{MN}$. Therefore, $\bC_H$ is row-stochastic as well. The claim $\rho(\bC_H)=1$ then follows from standard properties of row-stochastic matrices and the largest eigenvalue is $\lambda=1$.
\end{proof}

\begin{lemma}
\label{lem:spatial_temporal_part}
    Let $(X_\tau^{(\bC_H)})_{\tau=1}^\infty$ denote the spatio-temporal random walk on $V\times\{1,\dots,M\}$ given by the transition matrix $\bC_H$. Its dynamics in the temporal dimension is governed by the transition matrix $H$. More precisely, let
    \[
    \bar{X}_{\tau}^{(\bC_H)}=\operatorname{pr}_2(X_{\tau}^{(\bC_H)})
    \]
    where $\operatorname{pr}_2$ denotes the projection onto the second coordinate. Then, $(\bar{X}_{\tau}^{(\bC_H)})_{\tau=1}^{\infty}$ is a well-defined random walk and
    \[
    p\Big(\bar{X}_{\tau+1}^{(\bC_H)}=s|\bar{X}_{\tau}^{(\bC_H)}=t\Big)=h_{ts}.
    \] 
\end{lemma}

\begin{proof}
    Let
    \begin{equation*}
    B_{ts}=
        \begin{cases}
           K_{ts}, & t<s,\\
           0, &t=s,\\
           T_{st}, &t>s.
        \end{cases}
    \end{equation*}
    Then
    \[
    p\Big(X_{\tau+1}^{(\bC_H)}=(v_j,s)\mid X_{\tau}^{(\bC_H)}=(v_i,t)\Big)=h_{ts}(B_{ts})_{ij}.
    \]
    Thus, the expression
    \[
    \sum_{j=1}^Np\Big(X_{\tau+1}^{(\bC_H)}=(v_j,s)\mid X_{\tau}^{(\bC_H)}=(v_i,t)\Big)=\sum_{j=1}^Nh_{ts}(B_{ts})_{ij}=h_{ts}\sum_{j=1}^N(B_{ts})_{ij}=h_{ts}
    \]
    is independent of $i$, and therefore the following is well-defined 
    \begin{equation*}
    \begin{split}
        h_{ts}&=\sum_{j=1}^Np\Big(X_{\tau+1}^{(\bC_H)}=(v_j,s)\mid X_{\tau}^{(\bC_H)}=(v_i,t)\Big)=\sum_{j=1}^Np\Big(X_{\tau+1}^{(\bC_H)}=(v_j,s)\mid \bar{X}_{\tau}^{(\bC_H)}=t\Big)\\&=p\Big(\bar{X}_{\tau+1}^{(\bC_H)}=s|\bar{X}_{\tau}^{(\bC_H)}=t\Big). \qedhere
    \end{split}
    \end{equation*}
\end{proof}

\begin{proof}[Proof of lemma~\ref{le:spatio_temporal_time_reversibility}]
Let $\bpi=[\pi_1,\dots,\pi_M]^{\top}\in\mathbb R^M$ be defined as in lemma~\ref{le:H_random_walk}. Let
\[
\bnu
=
\big[\pi_1 \bmu_1^\top ,\,
      \pi_2 \bmu_2^\top ,\,
      \dots,\,
      \pi_M \bmu_M^\top \big]^\top 
\in \mathbb{R}^{MN}.
\]
In the proof we use the following identities: For $t<s$, it holds that
\[
\bmu_{s}^\top  T_{ts} = \bmu_sD_{\mu_s}^{-1}(K_{ts})^{\top}D_{\mu_t}=\bone^{\top}(K_{ts})^{\top}D_{\mu_t}=\bone^{\top} D_{\mu_t}=\bmu_t^\top
\]
and
\[
\bmu_t^\top  K_{ts} = \bmu_t^{\top}(S_t\cdots S_{s-1})=\bmu_{t+1}^{\top}(S_{t+1}\dots S_{s-1})=\cdots=\bmu_{s-1}^{\top}S_{s-1}=\bmu_{s}^\top  .
\]
Let $(\bnu^\top  \bC_H)_t \in \mathbb{R}^N$ and $\bnu_t$ denote the $t$th segments (of length $N$) of the vectors
$\bnu^\top  \bC_H$ and $\bnu$, respectively. For every $1 \le t \le M$, we have
\begin{align*}
(\bnu^\top  \bC_H)_t
=
\sum_{k=1}^{t-1} h_{kt}\,\pi_k\,\bmu_k^\top  K_{kt}
\;+\;
\sum_{k=t+1}^{M} h_{kt}\,\pi_k\,\bmu_k^\top  T_{tk}
=
\sum_{k=1}^{M} h_{kt}\,\pi_k \,\bmu_t^\top .
\end{align*}
As shown in lemma~\ref{le:H_random_walk}, $\pi_t h_{tk}=\pi_k h_{kt}$, which, with the row-stochasticity of $H$, gives
\[
\Big( \sum_{k=1}^{M} \pi_k h_{kt} \Big)\bmu_t^\top 
=
\Big( \sum_{k=1}^{M} h_{tk} \Big)\pi_t \bmu_t^\top 
=
\pi_t \bmu_t^\top 
=
\bnu_t^\top  .
\]
Hence, $\bnu^\top  \bC_H = \bnu^\top $, and $\bnu$ is a stationary distribution of $\bC_H$. Since $T_{ts}^{\top}D_{\mu_{s}}=D_{\mu_t}K_{ts}$ and $D_{\mu_{s}}T_{ts}=K_{ts}^{\top}D_{\mu_t}$, the detailed balance condition 
\[
D_{\nu}\bC_H=\bC_H^\top D_{\nu}
\]
holds so the random walk induced by $\bC_H$ is time-reversible with respect to $\bnu$. Consequently, the matrix $\bC_H$ is then self-adjoint with respect to $\bnu$-weighted inner product. Therefore, its spectrum is real. Since its spectral radius additionally satisfies $\rho(\bC_H)=1$ (see lemma~\ref{thm:row_stochasticity}) all eigenvalues are contained in the interval $[-1,1]$.
\end{proof}

\begin{proof}[Proof of lemma~\ref{le:Pi_avg_projection}]
    From \eqref{eq:Pi_avg} it is easy to see that $\bPi_{\mathrm{avg}}^2=\bPi_{\mathrm{avg}}$. Furthermore, $\bPi_{\mathrm{avg}}$ is self-adjoint with respect to the $\bnu$-weighted inner product. Indeed, using \eqref{eq:Pi_avg_matrix}, for any two observables $\bbf, \bg$, we obtain
    \begin{equation*}
        \la\bPi_{\mathrm{avg}}\bbf,\bg \ra_{\nu}=
        \sum_{t=1}^M\pi_t\la\bone\bmu_t^{\top}\bbf_t,\bg_t\ra_{\mu_t}=\sum_{t=1}^M\pi_t\la\bone,\bbf_t\ra_{\mu_t}\la\bone,\bg_t\ra_{\mu_t}=\sum_{t=1}^M\pi_t\la\bbf_t,\bone\mu_t^{\top}\bg_t\ra_{\mu_t}=\la\bbf,\bPi_{\mathrm{avg}}\bg\ra_{\nu}.
    \end{equation*}
    Since $\bPi_{\mathrm{avg}}$ is idempotent and $\bnu$-self-adjoint, it is a $\bnu$-orthogonal projection. Finally, 
    \begin{equation*}
    \operatorname{Im}(\bPi_{\mathrm{avg}})
    =
    \operatorname{span}\{e_t\otimes\bone \mid 1\le t\le M\},
    \end{equation*}
    follows directly from the definition of $\bPi_{\mathrm{avg}}$.
\end{proof}

\begin{proof}[Proof of lemma~\ref{le:projected_self_adjointness}]
We have
\begin{equation*}
\begin{split}
    \la\widehat{\bC}_H\ba,\bb\ra_{\bG_{\cW}}&=\ba^{\top}\widehat{\bC}_H^{\top}\bW^{\top}D_{\nu}\bW\bb \\
    &=\ba^{\top}\bW^{\top}\bC_H^{\top}D_{\nu}\bW(\bW^{\top}D_{\nu}\bW)^{-{\top}}\bW^{\top} D_{\nu}\bW\bb\\&=\ba^{\top}\bW^{\top} D_{\nu}\bC_H\bW\bb\\
    &=\ba^{\top}\bW^{\top} D_{\nu}\bW(\bW^{\top} D_{\nu}\bW)^{-1}\bW^{\top} D_{\nu}\bC_H\bW\bb \\
    &=\la\ba,\widehat{\bC}_H\bb\ra_{\bG_{\cW}},
\end{split}
\end{equation*}
where the third equality follows from the detailed balance condition of $\bC_H$ with respect to $\nu$ (see lemma~\ref{le:spatio_temporal_time_reversibility}), and in the fourth equality we insert the identity $I = \bW^{\top} D_{\nu}\bW \, (\bW^{\top} D_{\nu}\bW)^{-1}$.
\end{proof}

\begin{proof}[Proof of lemma~\ref{lem:coarse_grained_spectral_radius_thm}]
Since $\widehat{\bC}_H$ is self-adjoint with respect to the $\bG_{\mathcal W}$-weighted inner product, its spectrum is real-valued. For any $\ba\in\mathbb R^d$, we have
\[
\|\bW\ba\|_\nu^2 = \ba^{\top} \bW^{\top} D_\nu \bW \ba = \|\ba\|_{\bG_{\mathcal W}}^2 .
\]
Let $\bPi_{\mathcal W}$ denote the $\bnu$-orthogonal projector onto
$\mathcal W$. Using the definition of the Galerkin projection,
\[
\bW\widehat{\bC}_H \ba
= \bPi_{\mathcal W}\bC_H \bW\ba.
\]
Hence,
\[
\|\widehat{\bC}_H \ba\|_{\bG_{\mathcal W}}
= \|\bW\widehat{\bC}_H \ba\|_\nu 
= \|\bPi_{\mathcal W}\bC_H \bW\ba\|_\nu 
\le \|\bC_H \bW\ba\|_\nu
\le \|\bC_H\|_\nu\,\|\bW\ba\|_\nu,
\]
where the first inequality follows from the contractivity of the
$\bnu$-orthogonal projector $\bPi_{\mathcal W}$. Since $\|\bC_H\|_\nu\le 1$, we conclude that
\[
\|\widehat{\bC}_H \ba\|_{\bG_{\mathcal W}}
\le \|\bW\ba\|_\nu
= \|\ba\|_{\bG_{\mathcal W}}.
\]
Therefore,
\[
\|\widehat{\bC}_H\|_{\bG_{\mathcal W}} \le 1 .
\]
Finally, self-adjointness of $\widehat{\bC}_H$ with respect to the $\bG_{\cW}$-weighted inner product implies that its induced norm coincides with
its spectral radius. Hence
\[
\rho(\widehat{\bC}_H) = \|\widehat{\bC}_H\|_{\bG_{\mathcal W}} \le 1. \qedhere
\]
\end{proof}

\section{Formal description of a space--time network}
\label{sec:appendix_2}

We formalize the construction of the space--time network $\mathscr G_H=(\mathscr V, \mathscr E, \omega)$ as follows. The bijection
\[
\zeta : V \times \{1,\dots, M\} \rightarrow \{1,\dots, MN\}, 
\qquad
(v_i,t) \mapsto (t-1)N + i,
\]
flattens the space--time nodes and allows us to define a new state space
\[
\mathscr V 
= 
\zeta\big(V \times \{1,\dots,M\}\big)
=
\{1,\dots ,MN\}.
\]
Thus, we can interpret $\bC_H$ as the transition matrix of a random walk on a weighted undirected network $\mathscr G_H$, where the edge weights are given by the function $\omega$ as summarized in lemma~\ref{lem:multilayer_network}.

\begin{lemma}
\label{lem:multilayer_network}
    Let $t<s$. Within the network $\mathscr G_H=(\mathscr V, \mathscr E, \omega)$ associated to a temporal network $\mathbb G$ with a coupling matrix $H$, an edge
    \[
    e_{ij}^{ts}
    = \big(\zeta(v_i,t),\,\zeta(v_j,s)\big) \in \mathscr E\subset \mathscr V\times\mathscr V
    \]
    exists if and only if snapshots $G_t$ and $G_s$ are coupled with nonzero weight $h_{ts}$ and a temporal path of length $s-t$ exists from
    $v_i \in G_t$ to $v_j \in G_s$, that is, if and only if
    \[
    h_{ts}(K_{ts})_{ij} \neq 0.
    \]
    The weight of edge $e_{ij}^{ts}$ is defined as
    \[
    \omega(e_{ij}^{ts})
    = \nu(\zeta(v_i,t))(\bC_H)_{\zeta(v_i,t),\,\zeta(v_j,s)}=\pi_t\mu_t(v_i)(\bC_H)_{\zeta(v_i,t),\,\zeta(v_j,s)}.
    \]
    where $\mu_t$ is defined in \eqref{eq:evolution_of_prob_distr}, and $\bpi$ and $\bnu$ are as defined in lemmas~\ref{le:H_random_walk} and~\ref{le:spatio_temporal_time_reversibility}, respectively. Then, the transition matrix of a random walk on $\mathscr G$ is given by $\bC_H$.
\end{lemma}
\begin{proof}

First, let us show that the edge weight is well defined. Let $t<s$. Then
\[
\omega(e_{ij}^{ts})
=
\pi_t\mu_t(v_i)h_{ts}(K_{ts})_{ij}
=
\pi_sh_{st}\mu_t(v_i)(K_{ts})_{ij}
=
\pi_sh_{st}\mu_s(v_j)(T_{ts})^{\top}_{ij}
=
\pi_s\mu_s(v_j)h_{st}(T_{ts})_{ji}
=
\omega(e_{ji}^{st}),
\]
where the second equality follows from the detailed balance condition
$D_{\pi} H = H^{\top} D_{\pi}$, and the third equality follows from the relation $D_{\mu_t} K_{ts} = T_{ts}^{\top} D_{\mu_s}$.

An edge $e_{ij}^{ts}$ exists if and only if its weight is nonzero. Since
\[
\omega(e_{ij}^{ts})
=
\pi_t \mu_t(v_i)\,
(\bC_H)_{\zeta(v_i,t),\,\zeta(v_j,s)},
\]
and $\pi_t \mu_t(v_i) > 0$, we have
\[
\omega(e_{ij}^{ts}) \neq 0
\;\Longleftrightarrow\;
(\bC_H)_{e_{ij}^{ts}} \neq 0
\;\Longleftrightarrow\;
h_{ts}(K_{ts})_{ij} \neq 0.
\]

Lastly, for a vertex $\zeta(v_i,t)$ in $\mathscr V$ we have
\[
\deg(\zeta(v_i,t))=\sum_{s=1}^M\sum_{j=1}^N\omega(e_{ij}^{ts})=\sum_{s=1}^M\sum_{j=1}^N\pi_t\mu_t(v_i)(\bC_H)_{\zeta(v_i,t),\,\zeta(v_j,s)}=\pi_t\mu_t(v_i)\sum_{s=1}^Mh_{ts}=\pi_t\mu_t(v_i).
\]
Let the degree matrix $D_{\mathscr G_H} \in \mathbb R^{MN \times MN}$ be defined as
\[
D_{\mathscr G_H}
=
\diag(\deg(\zeta(v_1,1)),\dots,\deg(\zeta(v_N,M)))
\]
and let $\mathscr W\in\mathbb R^{MN\times MN}$ be a weighted adjacency matrix such that $\mathscr W_{\zeta(v_i,t),\,\zeta(v_j,s)}=\omega(e_{ij}^{ts})$. Then the transition matrix of the random walk on $\mathscr G_H$ is
\[
D_{\mathscr G_H}^{-1} \mathscr W
=
\bC_H. \qedhere
\]
\end{proof}

In this way, by connecting nodes across different network snapshots, we obtain a multilayer static network representation of the temporal network $\mathbb G$ that preserves both the temporal coupling of snapshots and their internal structural properties. 

\section{Computational complexity analysis}
\label{sec:appendix_3}
Here, we discuss and compare the computational complexity and memory requirements of the full and reduced formulations of our algorithm. Throughout this discussion, we assume that the temporal network is sparse, i.e., that the transition matrices $S_t$ are sparse. This is a natural assumption in many real-world applications \cite{Busiello_2017, Demaine_2019}.

\paragraph{Full model.}
The main computational bottleneck in both the full and reduced models is the construction of the matrices $\bC_H$ and $\widehat{\bC}_H$, respectively. The computational complexity depends on the sparsity pattern of the snapshot coupling matrix $H$. Suppose that snapshots separated by at most $r$ time steps are coupled and define
\[
r=\max\{|t-s|:h_{ts}\neq0\}.
\]
Then $H$ contains $\mathcal O(Mr)$ nonzero entries. Consequently, for each $1\le t\le M$, at most $r$ block matrices $K_{ts}$ must be computed. These building blocks of $\bC_H$ can be computed recursively via $K_{t,s+1}=K_{ts}S_s$, so that already computed matrices can be reused. Hence, for each block row of $\bC_H$, at most $r$ matrix multiplications are required. Although the matrices $S_t$ are sparse, their products $S_{ts}$ generally become dense for larger temporal distances. Therefore, each recursive multiplication has computational complexity $\mathcal O(N^2)$, resulting in a total cost of $\mathcal O(rN^2)$ per block row and $\mathcal O(MrN^2)$ for computing all blocks $K_{ts}$. Since the blocks $T_{ts}$ are obtained from $K_{ts}$ by transposition and diagonal rescaling, their computation has the same asymptotic complexity. Consequently, the total complexity of constructing $\bC_H$ scales quadratically with the number of nodes $N$ as $\mathcal O(MrN^2)$. The dominant eigenvectors of the sparse matrix $\bC_H$ can be computed efficiently using e.g. iterative eigensolvers such as the Lanczos algorithm~\cite{Golub_2013} or the Nyström method~\cite{Fowlkes_2004}. Regarding memory requirements, the main bottleneck is storing the matrix $\bC_H$. Since $\bC_H$ contains $\mathcal O(Mr)$ nonzero blocks of size $N\times N$, the total memory requirements are $\mathcal O(MrN^2)$.

\paragraph{Reduced model.}
Let
\[
d_{\mathrm{max}}=\max\{d_t:1\le t\le M\}.
\]
Using similar arguments as above together with the recursive relation $K_{ts}W_s=S_t(K_{t+1,s}W_s)$, the $\mathcal O(Mr)$ projected covariance blocks $\widehat{C}_{ts}$ can be computed in $\mathcal O(MrNd_{\mathrm{max}}^2)$
time. Furthermore, computing the matrices $\widehat{C}_{tt}^{-1}$ requires $\mathcal O(MNd_{\mathrm{max}}^2+Md_{\mathrm{max}}^3)$ operations. Once these matrices are available, assembling all projected transfer operators $\widehat{K}_{ts}$ and $\widehat{T}_{ts}$ requires $\mathcal O(Mrd_{\mathrm{max}}^3)$ operations. Since typically $d_{\mathrm{max}}\ll N$, summing everything up we obtain the overall computational complexity for computing $\widehat{\bC}_H$ to be $\mathcal O(MrNd_{\mathrm{max}}^2)$ which scales linearly with $N$. Furthermore, the reduced formulation provides substantial memory savings. Indeed, storing the sparse matrices $S_t, D_{\mu_t}$, the basis vectors that define the subspaces $\mathcal W_t$, and the reduced operator $\widehat{\bC}_H$ in total require $\mathcal O(MrNd_{\mathrm{max}})$ space which also scales only linearly with $N$.

\section{Heuristics for selecting spatial eigenvectors for clustering}
\label{sec:appendix_4}

Since the spatio-temporal random walk can be decomposed into temporal and spatial components (see lemma~\ref{lem:spatial_temporal_part}), the spectral properties of the transition matrix $\bC_H$ are influenced by both the spectral properties of the coupling matrix $H$ and the multi-step transfer operators. In particular, slowly decaying non-constant eigenvectors of $H$ that do not correspond to metastable behavior of the temporal component of the spatio-temporal random walk can still influence some of the leading spatial eigenvectors of $\bC_H$. As a consequence, not every leading spatial eigenvector can be interpreted as describing coherent sets of the spatio-temporal random walk.

Such effects can be recognized in spatial eigenvectors whose snapshot segments are approximately given by scalar multiples of a base observable carrying information about the spatial organization of nodes. In figure~\ref{fig:heuristics}\textbf{d} we show the leading spatial eigenvectors of the example 1 (see section~\ref{sec:example_1}). The segments of the first and second spatial eigenvectors can both be interpreted approximately as scalar multiples of observables separating the first and last 60 nodes. While the first spatial eigenvector is informative, in the sense that it indicates that these two groups remain coherent and separated from each other throughout the whole evolution, this information is lost in the second spatial eigenvector. More precisely, the corresponding base observable is modulated in this case by a non-constant, slowly decaying eigenvector of the coupling matrix $H$ shown in figure~\ref{fig:heuristics}\textbf{a}. This eigenvector, plotted in green in figure~\ref{fig:heuristics}\textbf{c}, corresponds to the third largest eigenvalue of $H$ (figure~\ref{fig:heuristics}\textbf{b}) and does not reflect metastability of the temporal component of the spatio-temporal random walk. Consequently, the second spatial eigenvector is not useful as a feature coordinate for clustering.

On the other hand, some informative spatial eigenvectors may correspond to relatively small eigenvalues of $\bC_H$. This occurs when they capture structural changes only during specific phases of the evolution, while carrying no spatial information during the remaining time intervals. For example, the fourth spatial eigenvector in figure~\ref{fig:heuristics}\textbf{d} captures spatial organization of the nodes only during phase~2, which accounts for approximately one third of the total time. Due to this interplay between temporal and spatial effects, there is typically no clear spectral gap among the leading eigenvalues of $\bC_H$ (see figure~\ref{fig:numerical_example_1}\textbf{d}) as it is expected in community detection algorithms for static networks. Consequently, standard approaches based on selecting spatial eigenvectors corresponding to the slowest time scales before a spectral gap are not directly applicable. Although a definitive solution to this problem remains open and will be addressed in future work, we provide heuristic guidelines for the coupling schemes considered in this work.

To detect these spurious metastability effects among spatial eigenvectors, we consider the matrices 
\[
F_k=
\begin{bmatrix}
(\bbf^{\mathrm{spat},k}_1)^\top\\
\vdots\\
(\bbf^{\mathrm{spat},k}_M)^\top
\end{bmatrix}
\in\mathbb{R}^{M\times N},
\]
whose rows are the snapshot segments of the spatial eigenvectors $\bbf^{\mathrm{spat},k}$. We now look at its singular value decomposition 
\[
F_k=\Phi_k\Sigma_k\Psi_k^{\top}=
\sum_{i=1}^{\min(M,N)}
\varsigma_i^k\cdot\phi_i^k\otimes_{\mathrm{out}}\psi_i^k,
\]
where the columns of $\Phi_k=[\phi_1^k,\dots,\phi_M^k]\in\mathbb{R}^{M\times M}$ and $\Psi_k=[\psi_1^k,\dots,\psi_N^k]\in\mathbb{R}^{N\times N}$ are the left and right singular vectors of $F_k$, respectively, $\Sigma_k\in\mathbb{R}^{M\times N}$ is the rectangular diagonal matrix containing the singular values
\[
\varsigma_1^k\geq \varsigma_2^k\geq \dots \geq \varsigma_{\min(M,N)}^k
\]
and $\otimes_{\mathrm{out}}$ denotes the outer product of vectors.

Spatial eigenvectors can now be interpreted as modulations of observables given by the dominant right singular vectors $\psi_1^k$ (figure~\ref{fig:heuristics}\textbf{g}), while the type of the modulation by slowly decaying eigenvectors of $H$ is reflected in the corresponding left singular vectors $\phi_1^k$ (figure~\ref{fig:heuristics}\textbf{f}).

In this example we consider the first six spatial eigenvectors. Since
\[
\varsigma_1^k \gg \varsigma_2^k \geq \dots \geq \varsigma_{32}^k,
\]
as shown in figure~\ref{fig:heuristics}\textbf{e}, we consider the approximations
\[
F_k\approx \varsigma_1^k\cdot\phi_1^k\otimes_{\mathrm{out}}\psi_1^k,
\]
for $1\leq k\leq 6$ (compare figure~\ref{fig:heuristics}\textbf{d} and figure~\ref{fig:heuristics}\textbf{h}). By inspecting the dominant left singular vectors (figure~\ref{fig:heuristics}\textbf{f}), we observe that, as shown in figure~\ref{fig:heuristics}\textbf{h}, the second spatial eigenvector is approximately obtained by modulating a base observable (figure~\ref{fig:heuristics}\textbf{g}) with the third eigenvector of $H$ (plotted in green in figure~\ref{fig:heuristics}\textbf{c}), while the fifth and sixth spatial eigenvectors are approximately obtained by modulating base observables with the second eigenvector of $H$ (plotted in orange in figure~\ref{fig:heuristics}\textbf{c}). We therefore exclude these spatial eigenvectors from the feature coordinates used for clustering and apply a clustering algorithm to spatial eigenvectors 1, 3 and 4.

Furthermore, by examining the support of the dominant left singular vectors $\phi_1^3$ and $\phi_1^4$, we observe that the third spatial eigenvector carries structural information during the second and third phases of the evolution, whereas the fourth spatial eigenvector carries structural information only during the second phase.

\begin{figure}
    \centering
    \includegraphics[width=0.92\linewidth]{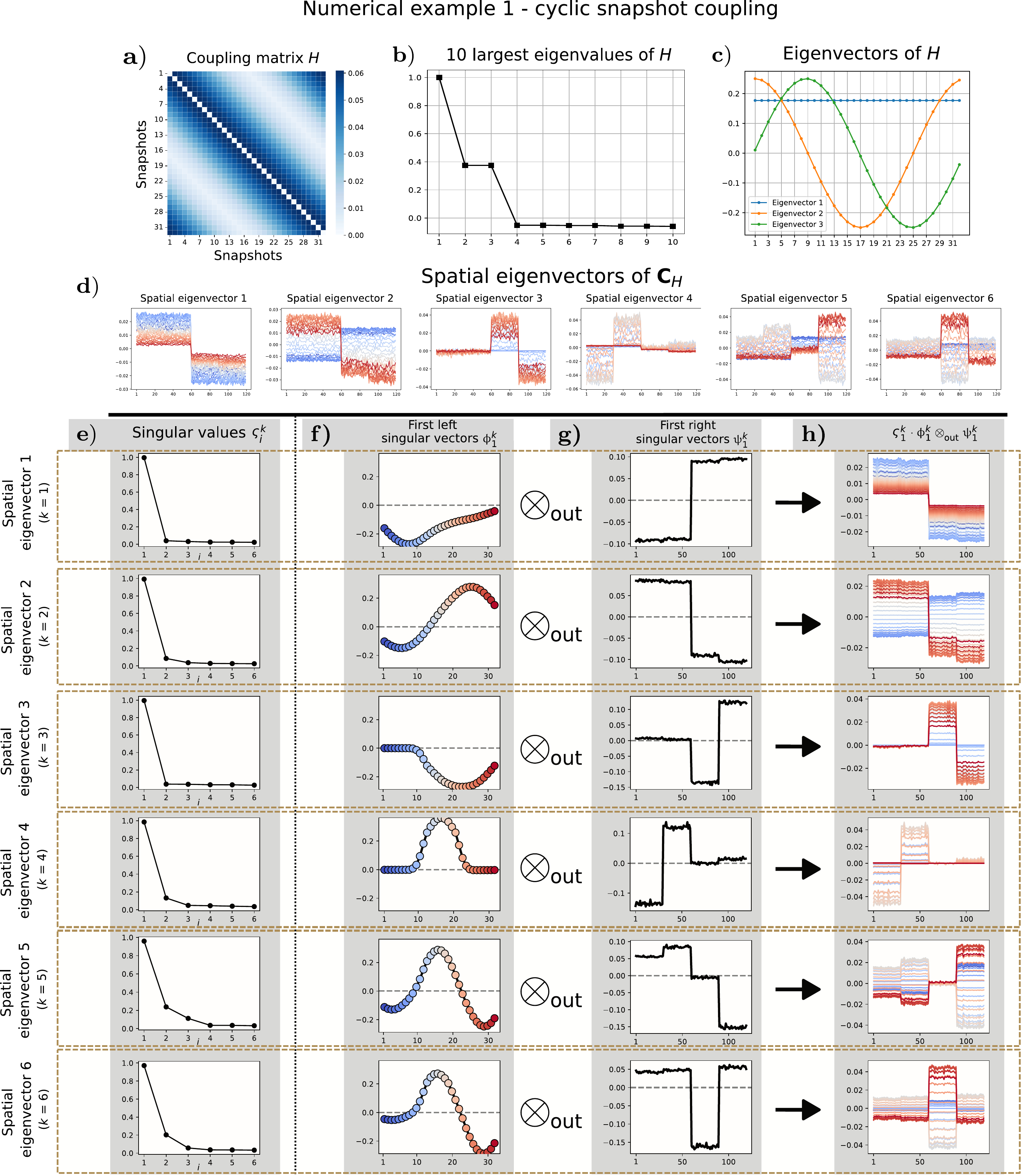}
    \caption{Illustration of the heuristics used for the selection of spatial eigenvectors as feature coordinates for the clustering algorithm in example~1.
    \textbf{a)} Coupling matrix $H$.
    \textbf{b)} Ten largest eigenvalues of $H$.
    \textbf{c)} Dominant eigenvectors of $H$.
    \textbf{d)} Six leading spatial eigenvectors of $\bC_H$ from example~1.
    \textbf{e)} Singular values of the matrices $F_k$ for $1\leq k\leq 6$.
    \textbf{f)} Dominant left singular vectors $\phi_1^k$ of the matrices $F_k$ for $1\leq k\leq 6$.
    \textbf{g)} Dominant right singular vectors $\psi_1^k$ of the matrices $F_k$ for $1\leq k\leq 6$.
    \textbf{h)} Rows of the rank-one approximations $\varsigma_1^k\cdot\phi_1^k\otimes_{\mathrm{out}}\psi_1^k$.}
    \label{fig:heuristics}
\end{figure}

\section{Performance analysis}
\label{sec:appendix_5}

Since in sections~\ref{sec:example_1} and~\ref{sec:example_2} we constructed the models in such a way that they have well-defined time-evolving communities, we use them to assess the sensitivity of our method with respect to the intra- and inter-community probabilities $p_{\mathrm{in}}$ and $p_{\mathrm{out}}$ as well as the coupling parameter $\alpha$. In example~\ref{sec:example_1}, the transition periods of length four between consecutive stability phases are gradual. We therefore use the clustering shown in figure~\ref{fig:reference_clustering}\textbf{a} as the reference solution, assigning the first two snapshots to the preceding phase and the last two to the following phase. In example~\ref{sec:example_2}, the situation is clearer since the community structure is well-defined in every snapshot; the corresponding reference solution is shown in figure~\ref{fig:reference_clustering}\textbf{b}. 

The rest of the section is organized as follows: in section~\ref{sec:comparison_mmm} we compare our approach with a method for community detection in temporal networks by multilayer modularity maximization introduced in~\cite{Mucha_2010}. In sections~\ref{sec:parameter_analysis_alpha} and~\ref{sec:parameter_analysis_p} we discuss how the performance of our method depends on the temporal coupling parameter $\alpha$ and the connection probabilities $p_{\mathrm{in}}$ and $p_{\mathrm{out}}$, and mention the resulting limitations.

\subsection{Comparison with multilayer modularity maximization}
\label{sec:comparison_mmm}
To provide a baseline comparison with an established temporal community detection method, we compare our approach with the multilayer modularity maximization method (MMM) introduced in~\cite{Mucha_2010}. This method extends the well-established modularity quality function from static networks to multilayer networks. In addition to the resolution parameter $\gamma$, it introduces an interslice coupling parameter $\omega$, which controls the coupling between copies of the same node in different snapshots; see \cite{Mucha_2010} for details. In table~\ref{tab:method_comparison}, we evaluate the performance of both algorithms on examples~\ref{sec:example_1} and~\ref{sec:example_2} by reporting the mean ARI between the obtained and reference clusterings over 20 network realizations. For MMM, we considered parameter combinations on a grid $(\gamma,\omega)\in[0,3]\times[0,3]$. The choice $\gamma=\omega=1$ produced the highest ARI score for both examples and is therefore used in the comparison.

In example~\ref{sec:example_1}, MMM correctly identifies the recurring community, while the locations of the structural changes fluctuate slightly around the prescribed transition points for both methods. In particular, one community from each pair $C_{11}/C_{12}$ and $C_{21}/C_{22}$ is typically merged with $C_1$ and $C_2$, respectively, with some variations across realizations. We note that this should not necessarily be interpreted as a drawback of the method since it also reflects a certain ambiguity in the prescribed reference clustering. However, we were not able to reproduce with MMM the partitioning obtained by our method under non-cyclic temporal coupling, where $C_1$ is identified as a new community after undergoing structural changes. This highlights the adaptability of our method to various scenarios: the choice of temporal coupling in our framework can lead to different interpretations of recurring and evolving structures. In example~\ref{sec:example_2}, MMM recovers the changes in node memberships entirely, resulting in a mean ARI of
$1.00$, compared with $0.97$ for our method. 
The two approaches are based on different principles: MMM formulates community detection as an optimization of a multilayer modularity quality function, whereas our method is derived from the dynamics of a random walk on the spatio-temporal network. For a broader discussion of modularity-based and dynamics-based approaches to community detection, we refer to \cite{Fortunato_2016}.

\begin{figure}
    \centering
    \includegraphics[width=1\linewidth]{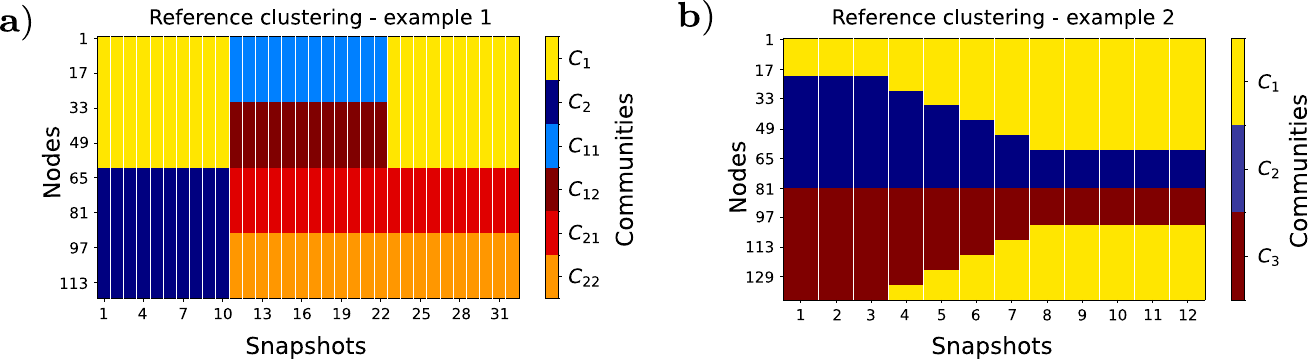}
    \caption{Reference clusterings of the temporal networks from \textbf{a)} example~\ref{sec:example_1} and \textbf{b)} example~\ref{sec:example_2}.}
    \label{fig:reference_clustering}
\end{figure}

\subsection{Parameter analysis: \texorpdfstring{$\alpha$}{alpha}}
\label{sec:parameter_analysis_alpha}
In this work, we use the snapshot-coupling scheme of the general form given in \eqref{eq:decay_coupling}, controlled by the decay parameter $\alpha>0$. The role of $\alpha$ is to set the temporal scale between the snapshots: the coupling should allow information to propagate across snapshots over moderate temporal distances, on a scale comparable to the duration of the relevant community phases.

If $\alpha$ is too large, the coupling decays too rapidly and only very close snapshots influence each other. Example~\ref{sec:example_2} and figure~\ref{fig:numerical_example_2}\textbf{j} illustrate the limitations of such strongly localized coupling. On the other hand, if $\alpha$ is too small, all snapshots are coupled with approximately equal strength, removing the temporal locality that motivates the construction. When determining the community structure of a given snapshot, the network structure in its temporal neighborhood should generally be more informative than the structure in the distant past or future. Whether a value of $\alpha$ is considered too large or too small therefore depends on the duration of the relevant community phases and the temporal extent of the dataset; no fixed numerical thresholds are implied.

Figure~\ref{fig:parameter_analysis}\textbf{a} presents a sensitivity analysis with respect to $\alpha$. In both examples, the best agreement with the prescribed reference clustering is obtained for moderate values of $\alpha$, consistent with the choices used in examples~\ref{sec:example_1} and~\ref{sec:example_2}. This supports the claim that an effective coupling should balance temporal locality with sufficient communication across snapshots. A more systematic, data-driven
strategy for choosing
$\alpha$ is an interesting direction for future work.

\begin{figure}
    \centering
    \includegraphics[width=1\linewidth]{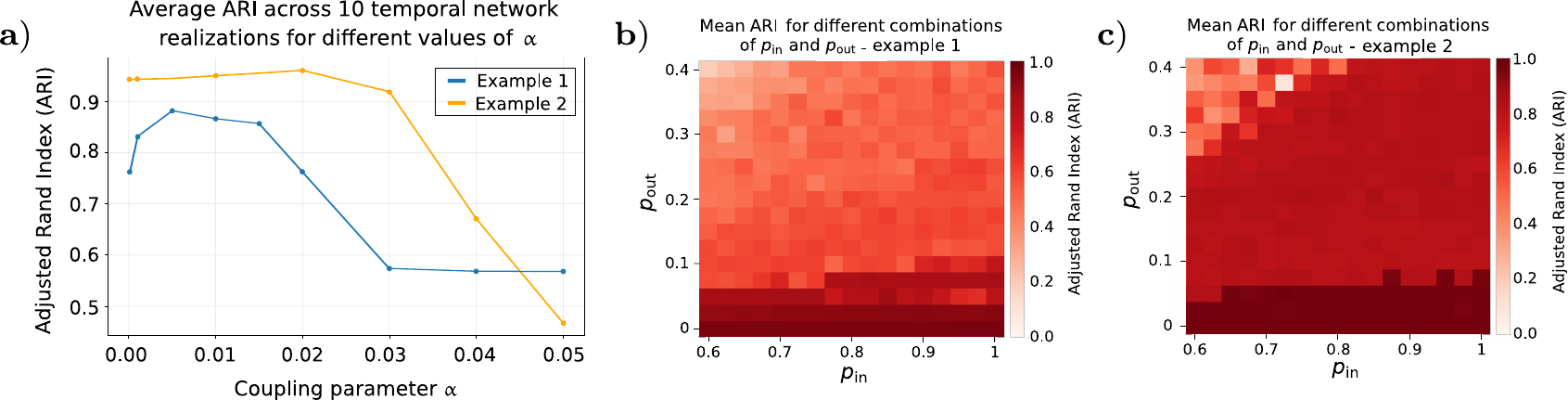}
    \caption{Illustration of the performance of our algorithm for different parameter regimes in examples~\ref{sec:example_1} and~\ref{sec:example_2}.
    \textbf{a)} Mean ARI between the reference partitions and those obtained by our algorithm, averaged over ten temporal network realizations, for different values of $\alpha$ in examples~\ref{sec:example_1} and~\ref{sec:example_2}.
    \textbf{b)} Mean ARI between the reference partitions and those obtained by our algorithm, averaged over five temporal network realizations, for different combinations of $p_{\mathrm{in}}$ and $p_{\mathrm{out}}$ in example~\ref{sec:example_1} and \textbf{c)} example~\ref{sec:example_2}.}
    \label{fig:parameter_analysis}
\end{figure}

\subsection{Parameter analysis: \texorpdfstring{$p_{\mathrm{in}}$ and $p_{\mathrm{out}}$}{p-in and p-out}.}
\label{sec:parameter_analysis_p}
In figure~\ref{fig:parameter_analysis}\textbf{b,c}, we report the mean ARI between the reference partitions and the partitions obtained by our method, averaged over five SBM network realizations in examples~\ref{sec:example_1} and~\ref{sec:example_2}, for different combinations of $p_{\mathrm{in}}$ and $p_{\mathrm{out}}$. As expected, the method performs best when the communities are clearly separated, while the ARI decreases as the distinction between intra- and inter-community connectivity becomes less pronounced. Some exceptions occur even in seemingly favorable parameter regimes with large $p_{\mathrm{in}}$ and small $p_{\mathrm{out}}$. For instance, similarly to the behavior observed for MMM, our method may sometimes merge communities such as $C_2$ and $C_{21}$, or $C_2$ and $C_{22}$. As discussed above, such outcomes are consistent with the ambiguity of the transition after snapshot $G_{10}$, which can be interpreted either as a genuine split of $C_2$ into two new communities $C_{21}$ and $C_{22}$, or as one part detaching while the remaining nodes continue to represent the original community. 

This trend can be interpreted through the metastability of the underlying spatio-temporal random walk: identifying a community across time requires the walker to remain confined within it while exploring different snapshots. The expected exit time from a community provides an intuitive measure of this persistence; see \cite{Djurdjevac_2012_phd} for further details. To illustrate this, we consider the community $C_1=\{v_1,\ldots,v_{60}\}$ in example~\ref{sec:example_1}, which remains separated from the rest of the network throughout all 32 snapshots. Figure~\ref{fig:exit_times}\textbf{a} shows its expected exit times, averaged across snapshots, for different combinations of $p_{\mathrm{in}}$ and $p_{\mathrm{out}}$. As this separation weakens, exit times decrease, and the walker is expected to explore only snapshots within a limited time interval before leaving the community; that is, it is expected to leave the community $C_1$ before visiting all snapshots. Consequently, the information about this community in the network may be represented by several spatial eigenvectors rather than by a single one, with each capturing its separation from the rest of the network over a different temporal interval. This interpretation is intended to provide intuition for the observed behavior rather than a precise detectability criterion.

The resulting change in the dominant spatial eigenvectors is illustrated in figure~\ref{fig:exit_times}\textbf{c}. We compare temporal networks from example~\ref{sec:example_1} with cyclic snapshot coupling for the two parameter regimes: $p_{\mathrm{in}}=0.9, p_{\mathrm{out}}=0.025$ and $p_{\mathrm{in}}=0.7, p_{\mathrm{out}}=0.2$. In the first regime, the contrast between intra- and inter-community connectivity is strong and the relevant exit time is approximately $37$, which is larger than the total duration of 32 snapshots, and the dominant spatial eigenvector successfully captures the separation of $C_1$ from the remainder of the network across all snapshots.

In the second regime, the corresponding exit time is only approximately $5$. Instead of being represented by a single global spatial eigenvector, the metastable behavior is distributed across two dominant spatial eigenvectors, each describing the separation of $C_1$ from the rest of the network over a different temporal subinterval. In particular, one eigenvector captures this separation during the first half of the temporal network, while another captures it during the second half. However, the information that $C_1$ represents the same recurring community in snapshots $G_1$--$G_{10}$ and $G_{23}$--$G_{32}$ is no longer contained in a single dominant mode. Importantly, this happens despite the cyclic coupling between snapshots. This limitation therefore does not originate from the temporal coupling itself, but from the lack of sufficiently long-lived spatial metastability in the underlying random walk.

One possible way to mitigate this limitation is to consider random walks on weighted networks that further favor transitions between nodes with similar neighborhoods. Such a weighting can increase the tendency of random walkers to remain inside their initial community and therefore increase the corresponding exit times, even in less favorable regimes of $p_{\mathrm{in}}$ and $p_{\mathrm{out}}$. Examples of random walks with related properties can be found in \cite{Djurdjevac_2012_phd}. To illustrate this idea, figure~\ref{fig:exit_times}\textbf{b} shows the exit times for community $C_1$ in example~\ref{sec:example_1} when the edges are weighted by the Jaccard coefficient, which measures the similarity between the neighborhoods of two adjacent nodes:
\[
J(v_i,v_j)
=
\frac{|N(v_i)\cap N(v_j)|}
{|N(v_i)\cup N(v_j)|}.
\]
The resulting increase in exit times suggests that such weighting schemes may extend the range of parameter regimes in which network communities can be identified reliably. A more systematic analysis of this possibility is left for future work.

\begin{figure}
    \centering
    \includegraphics[width=0.8\linewidth]{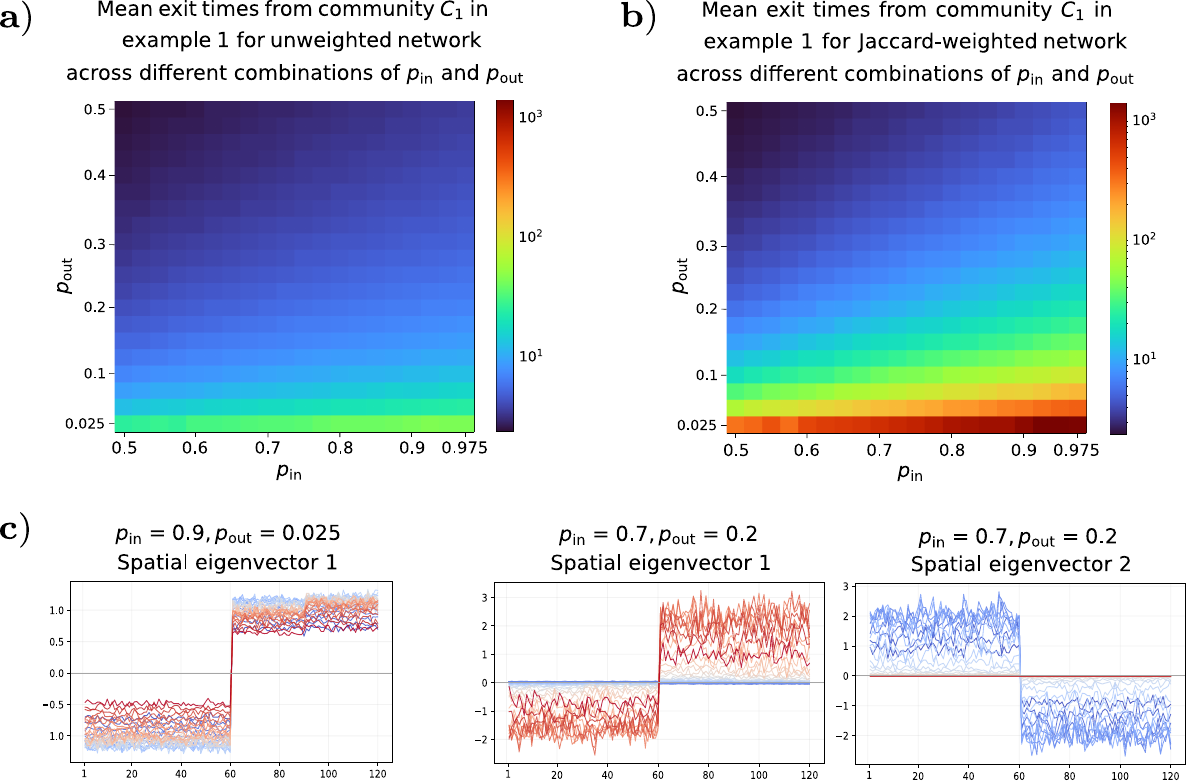}
    \caption{Illustration of the performance of our algorithm in two parameter regimes, and the mean exit times from community $C_1$ in example~\ref{sec:example_1} for unweighted and Jaccard-weighted random walks.
    \textbf{a)} Mean exit times from community $C_1$ for the unweighted network across different combinations of $p_{\mathrm{in}}$ and $p_{\mathrm{out}}$.
    \textbf{b)} Mean exit times from community $C_1$ for the Jaccard-weighted network across different combinations of $p_{\mathrm{in}}$ and $p_{\mathrm{out}}$.
    \textbf{c)} Spatial eigenvectors of $\bC_H$ in example~\ref{sec:example_1} for two different regimes of $p_{\mathrm{in}}$ and $p_{\mathrm{out}}$.}
    \label{fig:exit_times}
\end{figure}

\end{document}